\documentclass[a4paper,11pt]{article} 

\usepackage[english]{babel}
\usepackage[a4paper]{geometry}
\usepackage{enumerate}  
\usepackage{amsmath}
\usepackage{graphicx}
\usepackage{color}
\usepackage{amsthm}
\usepackage{amssymb}

\usepackage{hyperref}    

\newcommand{\di}{\textrm{d}}
\newcommand{\hc}{\mbox{h.c.}}

\usepackage{xcolor}
\usepackage{ulem}  

\let\a=\alpha     \let\g=\gamma     \let\d=\delta     \let\e=\varepsilon
             \let\l=\lambda
          \let\x=\xi        \let\p=\pi        \let\r=\rho
\let\s=\sigma \let\t=\tau         \let\ph=\varphi   \let\c=\chi
\let\ps=\psi   \let\o=\omega     
 \let\D=\Delta       \let\L=\Lambda    
             
\let\O=\Omega

\def\cE{{\cal E}}\def\cV{{\cal V}}
\def\cC{{\cal C}}\def\cF{{\cal F}}\def\cH{{\cal H}}
\def\cN{{\cal N}}\def\cZ{{\cal Z}}
\def\cR{{\cal R}}\def\cL{{\cal L}}
\def\cG{{\cal G}}
\def\cO{{\cal O}}\def\cK{{\cal K}}

  \def\v0{{\vec 0}}

 \def\hw{{\widehat w}}  

\def\bal{{\bar \l}}

\def\tl#1{{\tilde{#1}}}

\def\bR{\mathbb{R}}

\def\cV{\mathcal{V}}
\def\cF{\mathcal{F}}
\def\cG{\mathcal{G}}
\def\cL{\mathcal{L}}
\def\cN{\mathcal{N}}
\def\cE{\mathcal{E}}
\def\cK{\mathcal{K}}
\def\cH{\mathcal{H}}
\def\eps{\varepsilon}
\def\ph{\varphi}

\def\bN{\mathbb{N}}  
\def\bZ{\mathbb{Z}} 
\def\bR{\mathbb{R}}

\def\indic{\hbox{\raise-2pt \hbox{\indbf 1}}}

\let\dpr=\partial

\let\==\equiv

\let\io=\infty
\let\0=\noindent

\def\*{{\hfill\break\null\hfill\break}}

\def\bmedia#1{{\bigl\langle#1\bigr\rangle}}

\def\tende#1{\,\vtop{\ialign{##\crcr\rightarrowfill\crcr
             \noalign{\kern-1pt\nointerlineskip}
             \hskip3.pt${\scriptstyle #1}$\hskip3.pt\crcr}}\,}
\def\otto{\,{\kern-1.truept\leftarrow\kern-5.truept\to\kern-1.truept}\,}
\def\fra#1#2{{#1\over#2}}

\def\wt#1{\widetilde{#1}}

\def\sqt[#1]#2{\root #1\of {#2}}

\def\lis{\overline}

\def\hc{{\rm h.c.}\,}

\def\tr{{\rm tr}}

\def\wt{\widetilde}

\def\be{\begin{equation}}
\def\ee{\end{equation}}
\def\bea{\begin{eqnarray}}\def\eea{\end{eqnarray}}
\def\bean{\begin{eqnarray*}}\def\eean{\end{eqnarray*}}
\def\bfr{\begin{flushright}}\def\efr{\end{flushright}}
\def\bc{\begin{center}}\def\ec{\end{center}}
\def\bal{\begin{align}} 
\def\eal{\end{align}}

\def\spl#1\spl{\[ \begin{split}#1\end{split} \]}

\def\bd{\begin{description}}\def\ed{\end{description}}
\def\bv{\begin{verbatim}}\def\ev{\end{verbatim}}

\def\Halmos{\hfill\vrule height10pt width4pt depth2pt \par\hbox to \hsize{}}

\newtheorem{theorem}{Theorem}[section]
\newtheorem{prop}{Proposition}[section]
\newtheorem{lemma}[prop]{Lemma} 
\numberwithin{equation}{section}
\title{Bose-Einstein condensation for two dimensional bosons\\  in the Gross-Pitaevskii regime}
\date{\today}

\def \aa{{\mathfrak a}}

\begin{document}

\author{Cristina Caraci$^1$, Serena Cenatiempo$^1$, Benjamin Schlein$^2$ 
\\\\
Gran Sasso Science Institute, Viale Francesco Crispi 7 \\ 
67100 L'Aquila, Italy$^1$ 
\\ \\
Institute of Mathematics, University of Zurich\\
Winterthurerstrasse 190, 8057 Zurich, 
Switzerland$^2$}

\maketitle

\abstract{We consider systems of N bosons trapped on the two-dimensional unit torus, in the Gross-Pitaevskii regime, where the scattering length of the repulsive interaction is exponentially small in the number of particles. We show that low-energy states exhibit complete Bose-Einstein condensation, with almost optimal bounds on the number of orthogonal excitations. }

\section{Introduction}

We consider $N \in \bN$ bosons trapped in the two-dimensional box $\L = [-1/2;1/2]^2$ with periodic boundary conditions. In the Gross-Pitaevskii regime, particles interact through a repulsive pair 
potential, with a scattering length exponentially small in $N$. The Hamilton operator is given by 
\begin{equation}\label{eq:HN} H_N = \sum_{j=1}^N -\Delta_{x_j} + \sum_{i<j}^N e^{2N} V (e^N(x_i -x_j)) \end{equation}
and acts on a dense subspace of $L^2_s (\Lambda^N)$, the Hilbert space consisting of functions in $L^2 (\L^N)$ that are invariant with respect to permutations of the $N$ particles. We assume here $V \in L^3 (\bR^2)$ to be compactly supported and pointwise non-negative (i.e. $V(x) \geq 0$ for almost all $x \in \bR^2$). 

We denote by $\aa$ the scattering length of the unscaled potential $V$. We recall that in two dimensions and for a potential $V$ with finite range $R_0$, the scattering length is defined by
\be \label{eq:defa0}
\frac{2\pi}{\log(R/\aa)} =\inf_{\phi} \int_{B_R}  \left[ |\nabla \phi|^2 + \frac 1 2 V |\phi|^2   \right] dx 
\ee
where $R > R_0$, $B_R$ is the disk of radius $R$ centered at the origin and the infimum is taken over functions $\phi \in H^1(B_R)$ with $\phi (x)=1$ for all $x$ with $|x|=R$. The unique minimizer of the variational problem on the r.h.s. of \eqref{eq:defa0} is non-negative, radially symmetric and satisfies the scattering equation
\[
- \D \phi^{(R)} + \frac 12 V \phi^{(R)} =0
\]
in the sense of distributions. 
For $R_0 < |x| \leq R$, we have 
\[
\phi^{(R)} (x) = \frac{\log(|x|/\aa)}{\log (R/\aa) }\,.
\]
By scaling, $\phi_N (x) := \phi^{(e^N R)} (e^N x)$  is such that 
\[ -\D \phi_N + \frac{1}{2} e^{2N} V (e^N x) \phi_N = 0 \] 
We have 
\[
\phi_{N}(x) =\frac{\log(|x|/\aa_N)}{\log (R/\aa_N) } \qquad \forall x \in \bR^2 : e^{-N} R_0 < |x| \leq R\,,
\]
for all $x \in \bR^2$ with $e^{-N} R_0 < |x| \leq R$. Here $\aa_N= e^{-N} \aa$.

The properties of trapped two dimensional bosons in the Gross-Pitaevskii regime (in the more general case where the bosons are confined by external trapping potentials) have been first studied in \cite{LSeY2d, LS-BEC, LSe-GP-rotating}. These results can be translated to the Hamilton operator (\ref{eq:HN}), defined on the torus, with no external potential.
They imply that the ground state energy $E_N$ of (\ref{eq:HN}) is such that 
\begin{equation}\label{eq:en-ti} 
E_N = 2 \pi N \big( 1 + O(N^{-1/5})\big)\,.
\end{equation}

Moreover, they imply Bose-Einstein condensation in the zero-momentum mode $\ph_0 (x) = 1$ for all $x \in \Lambda$, for any approximate ground state of (\ref{eq:HN}). More precisely, it follows from \cite{LS-BEC} that, for any sequence $\psi_N \in L^2_s (\L^N)$ with $\| \psi_N \| = 1$ and  
\begin{equation} \label{eq:appro-gs} 
\lim_{N \to \infty} \frac{1}{N} \langle \psi_N , H_N \psi_N \rangle = 2\pi  ,
\end{equation}
the one-particle reduced density matrix $\gamma_N = \tr_{2, \dots , N} |\psi_N \rangle \langle \psi_N |$ is such that 
\begin{equation}\label{eq:BEC1} 
1 - \langle \ph_0 , \gamma_N \ph_0 \rangle \leq CN^{-\bar{\d}} 
\end{equation}
for a sufficiently small $\bar{\d}>0$. The estimate (\ref{eq:BEC1}) states that, in many-body states satisfying (\ref{eq:appro-gs}) (approximate ground states), almost all particles are described by the 
one-particle orbital $\ph_0$, with at most $N^{1-\delta} \ll N$ orthogonal excitations.

For $V \in L^3(\bR^2)$, our main theorem improves (\ref{eq:en-ti}) and (\ref{eq:BEC1})  by providing more precise bounds on the ground state energy and on the number of excitations.
\begin{theorem}\label{thm:main}
	Let $V \in L^3 (\bR^2)$ have compact support, be spherically symmetric and pointwise non-negative. Then there exists a constant $C > 0$ such that the ground state energy $E_N$ of (\ref{eq:HN}) satisfies 
	\be\label{eq:Enbd}
	2\pi N - C  \leq E_N  \leq 2\pi N + C \log N
	\ee
	Furthermore, consider a sequence $\psi_N \in L^2_s (\Lambda^N)$ with $\| \psi_N \| = 1$ and such that   
	\[\label{eq:apprH} \langle \psi_N , H_N \psi_N \rangle  \leq 2\pi N + K \]
	for a $K > 0$. Then the reduced density matrix $\gamma_N = \tr_{2,\dots , N} | \psi_N \rangle \langle\psi_N |$ associated with $\psi_N$ is such that 
	\begin{equation}\label{eq:BEC} 1 - \langle \ph_0 , \gamma_N \ph_0 \rangle \leq \frac{C(1+ K)}{N} \end{equation}
	for all $N \in \bN$ large enough. 
\end{theorem}

\vskip 0.5cm

It is interesting to compare the Gross-Pitaevskii regime with the thermodynamic limit, where a Bose gas of $N$ particles interacting through a fixed potential with scattering length $\aa$ is confined in a box with volume $L^2$, so that $N, L \to \infty$ with the density $\r=N/L^2$ kept fixed. Let 
$b=|\log(\r  \aa^2)|^{-1}$. Then, in the dilute limit $\r \aa^2 \ll 1$, the ground state energy per particle in the thermodynamic limit is expected to satisfy 
\be \label{e0}
e_0(\r) = 4  \pi \r^2 b \Big( 1 + b \log b  + \big(1/2 + 2\g + \log \pi \big) b + o(b) \Big)\,,
\ee
with $\g$ the Euler's constant. The leading order term on the r.h.s. of \eqref{e0} has been first derived in  \cite{Schick} and then rigorously established in \cite{LY2d}, with an error rate $b^{-1/5}$. The corrections up to order $b$ have been predicted in \cite{GroundState2d-1, GroundState2d-2, GroundState2d-3}. To date, there is no rigorous proof of \eqref{e0}. Some partial result, based on the restriction to quasi-free states, has been recently obtained in  \cite[Theorem 1]{FNRS-2d}. 

Extrapolating from (\ref{e0}), in the Gross-Pitaevskii regime we expect $|E_N - 2\pi N| \leq C$. While our estimate (\ref{eq:Enbd}) captures the correct lower bound, the upper bound is off by a logarithmic correction. Eq. (\ref{eq:BEC}), on the other hand, is expected to be optimal (but of course, by (\ref{eq:Enbd}), we need to choose $K = C \log N$ to be sure that (\ref{eq:apprH}) can be satisfied). 
This bound can be used as starting point to investigate the validity of Bogoliubov theory for two dimensional bosons in the Gross-Pitaevskii regime, following the strategy developed in \cite{BBCS4} for the three dimensional case; we plan to proceed in this direction in a separate paper. 

The proof of Theorem \ref{thm:main} follows the strategy  that has been recently introduced in \cite{BBCS3} to prove condensation for three-dimensional bosons in the Gross-Pitaevskii limit. There are, however, additional obstacles in the two-dimensional case, requiring new ideas. To appreciate the difference between the Gross-Pitaevskii regime in two- and three-dimensions, we can compute the energy of the trivial wave function $\psi_N \equiv 1$. The expectation of (\ref{eq:HN}) in this state is of order $N^2$. It is only through correlations that the energy can approach (\ref{eq:Enbd}). Also in three dimensions, uncorrelated many-body wave functions have large energy, but in that case the difference with respect to the ground state energy is only of order $N$ ($N \widehat{V} (0)/2$ rather than $4\pi \frak{a}_0 N$). This observation is a sign that correlations in two-dimensions are stronger and play a more important role than in three dimensions (this creates problems in handling error terms that, in the three dimensional setting, were simply estimated in terms of the integral of the potential). 

The paper is organized as follows. In Sec. \ref{sec:fock} we introduce our setting, based on a description of orthogonal excitations of the condensate on a truncated Fock space. In Sec. \ref{sec:3} and \ref{sec:4} we show how to renormalize the excitation Hamiltonian, to regularise the singular interaction. Finally, in Sec. \ref{sec:main}, we show our main theorem. Sec. \ref{sec:RN} and App. \ref{sec:GN} contain the proofs of the results stated in \ref{sec:3} and \ref{sec:4}, respectively. Finally, in App. \ref{App:omega} we establish some  properties of the solution of the Neumann problem associated with the two-body potential $V$. 

\medskip

\thanks{{\it Acknowledgment.}  We are thankful to A. Olgiati for discussions on the two dimensional scattering equation.  C.C. and S.C. gratefully acknowledge the support from the GNFM Gruppo Nazionale per la Fisica Matematica - INDAM through the project ``Derivation of effective theories for large quantum systems''. B. S. gratefully acknowledges partial support from the NCCR SwissMAP, from the Swiss National Science Foundation through the Grant ``Dynamical and energetic properties of Bose-Einstein condensates'' and from the European Research Council through the ERC-AdG CLaQS.}

\section{The Excitation Hamiltonian}
\label{sec:fock}

Low-energy states of (\ref{eq:HN}) exhibit condensation in the zero-momentum mode $\ph_0$ defined by $\ph_0 (x) = 1$ for all $x \in \Lambda = [-1/2;1/2]^2$. Similarly as in \cite{LNSS,BBCS1,BBCS3}, we are going to describe excitations of the condensate on the truncated bosonic Fock space 
\[ \cF^{\leq N}_+ = \bigoplus_{k=0}^N L^2_\perp (\Lambda)^{\otimes_s k} \]
constructed on the orthogonal complement $L^2_\perp (\Lambda)$ of $\ph_0$ in $L^2 (\Lambda)$. To reach this goal, we define a unitary map $U_N : L^2_s (\Lambda^N) \to \cF_+^{\leq N}$ by requiring that $U_N \psi_N = \{ \alpha_0, \alpha_1, \dots , \alpha_N \}$, with $\alpha_j \in L^2_\perp (\Lambda)^{\otimes_s j}$, if 
\[ \psi_N= \alpha_0 \ph_0^{\otimes N} + \alpha_1 \otimes_s \ph_0^{\otimes (N-1)} + \dots + \alpha_N \]
With the usual creation and annihilation operators, we can write 
\[ U_N \, \psi_N = \bigoplus_{n=0}^N  (1-|\ph_0 \rangle \langle \ph_0|)^{\otimes n} \frac{a (\ph_0)^{N-n}}{\sqrt{(N-n)!}} \, \psi_N \]
for all $\psi_N \in L^2_s (\Lambda^N)$. It is then easy to check that $U_N^* : \cF_{+}^{\leq N} \to L^2_s (\Lambda^N)$ is given by 
\[ U_N^* \, \{ \alpha^{(0)}, \dots , \alpha^{(N)} \} = \sum_{n=0}^N \frac{a^* (\ph_0)^{N-n}}{\sqrt{(N-n)!}} \, \alpha^{(n)} \]
and that $U_N^* U_N = 1$, i.e. $U_N$ is unitary. 

With $U_N$, we can define the excitation Hamiltonian $\cL_N := U_N H_N U_N^*$, acting on a dense subspace of $\cF_+^{\leq N}$. To compute the operator $\cL_N$, we first write the Hamiltonian (\ref{eq:HN}) in momentum space, in terms of creation and annihilation operators $a_p^*, a_p$, for momenta $p \in \Lambda^* = 2\pi \bZ^2$. We find 
\begin{equation}\label{eq:Hmom} H_N = \sum_{p \in \Lambda^*} p^2 a_p^* a_p + \frac 12\sum_{p,q,r \in \Lambda^*} \widehat{V} (r/e^N) a_{p+r}^* a_q^* a_{p} a_{q+r} 
\end{equation}
where 
\[ \widehat{V} (k) = \int_{\bR^2} V (x) e^{-i k \cdot x} dx \] 
is the Fourier transform of $V$, defined for all $k \in \bR^2$ (in fact, (\ref{eq:HN}) is the restriction of (\ref{eq:Hmom}) to the $N$-particle sector of the Fock space). We can now determine $\cL_N$ using the following rules, describing the action of the unitary operator $U_N$ on products of a creation and an annihilation operator (products of the form $a_p^* a_q$ can be thought of as operators mapping $L^2_s (\Lambda^N)$ to itself). For any $p,q \in \Lambda^*_+ = 2\pi \bZ^2 \backslash \{ 0 \}$, we find (see \cite{LNSS}):
\begin{equation}\label{eq:U-rules}
\begin{split} 
U_N \, a^*_0 a_0 \, U_N^* &= N- \cN_+  \\
U_N \, a^*_p a_0 \, U_N^* &= a^*_p \sqrt{N-\cN_+ } \\
U_N \, a^*_0 a_p \, U_N^* &= \sqrt{N-\cN_+ } \, a_p \\
U_N \, a^*_p a_q \, U_N^* &= a^*_p a_q \,.
\end{split} \end{equation}
where $\cN_+ = \sum_{p\in \Lambda^*_+} a_p^* a_p$ is the number of particles operator on $\cF_+^{\leq N}$. We conclude that 
\begin{equation}\label{eq:cLN} \cL_N =  \cL^{(0)}_{N} + \cL^{(2)}_{N} + \cL^{(3)}_{N} + \cL^{(4)}_{N} \end{equation}
with
\begin{equation}\label{eq:cLNj} \begin{split} 
\cL_{N}^{(0)} =\;& \frac 12 \widehat{V} (0) (N-1)(N-\cN_+ ) + \frac 12 \widehat{V} (0) \cN_+  (N-\cN_+ ) \\
\cL^{(2)}_{N} =\; &\sum_{p \in \Lambda^*_+} p^2 a_p^* a_p + N\sum_{p \in \Lambda_+^*} \widehat{V} (p/e^N) \left[ b_p^* b_p - \frac{1}{N} a_p^* a_p \right] \\ &+ \frac N2 \sum_{p \in \Lambda^*_+} \widehat{V} (p/e^N) \left[ b_p^* b_{-p}^* + b_p b_{-p} \right] \\
\cL^{(3)}_{N} =\; & \sqrt{N} \sum_{p,q \in \Lambda_+^* : p+q \not = 0} \widehat{V} (p/e^N) \left[ b^*_{p+q} a^*_{-p} a_q  + a_q^* a_{-p} b_{p+q} \right] \\
\cL^{(4)}_{N} =\; & \frac 12\sum_{\substack{p,q \in \Lambda_+^*, r \in \Lambda^*: \\ r \not = -p,-q}} \widehat{V} (r/e^N) a^*_{p+r} a^*_q a_p a_{q+r} \,,
\end{split} \end{equation}
where we introduced generalized creation and annihilation operators  
\[\label{eq:bp-de} 
b^*_p = U_N a_p^* U_N^* = a^*_p \, \sqrt{\frac{N-\cN_+}{N}} , \qquad \text{and } \quad  b_p = U_N a_p U_N^* =  \sqrt{\frac{N-\cN_+}{N}} \, a_p 
\]
for all $p \in \Lambda^*_+$. 

On states exhibiting complete Bose-Einstein condensation in the zero-momentum mode $\ph_0$, we have $a_0 , a_0^* \simeq \sqrt{N}$ and we can therefore expect that $b_p^* \simeq a_p^*$ and that $b_p \simeq a_p$. From the canonical commutation relations for the standard creation and annihilation operators $a_p, a_p^*$, we find 
\begin{equation}\label{eq:comm-bp} \begin{split} [ b_p, b_q^* ] &= \left( 1 - \frac{\cN_+}{N} \right) \delta_{p,q} - \frac{1}{N} a_q^* a_p 
\\ [ b_p, b_q ] &= [b_p^* , b_q^*] = 0 \,.
\end{split} \end{equation}
Furthermore,
\[\label{eq:comm2} [ b_p, a_q^* a_r ] = \delta_{pq} b_r, \qquad  [b_p^*, a_q^* a_r] = - \delta_{pr} b_q^* \]
for all $p,q,r \in \Lambda_+^*$; this implies in particular that $[b_p , \cN_+] = b_p$, $[b_p^*, \cN_+] = - b_p^*$. It is also useful to notice that the operators $b^*_p, b_p$, like the standard creation and annihilation operators $a_p^*, a_p$, can be bounded by the square root of the number of particles operators; we find
\begin{equation*}
\| b_p \xi \|  \leq \| \cN_+^{1/2} \xi \| \,, \qquad 
\| b^*_p \xi \| \leq  \| (\cN_+ + 1)^{1/2} \xi \| 
\end{equation*}
for all $\xi \in \cF^{\leq N}_+$. Since $\cN_+  \leq N$ on $\cF_+^{\leq N}$, the operators $b_p^* , b_p$ are bounded, with $\| b_p \|, \| b^*_p \| \leq (N+1)^{1/2}$. 

\section{Quadratic renormalization} \label{sec:3}

From \eqref{eq:cLNj} we see that conjugation with $U_N$ extracts, from the original quartic interaction in (\ref{eq:Hmom}), some large constant and quadratic contributions, collected in $\cL^{(0)}_N$ and $\cL^{(2)}_N$ respectively. In particular, the expectation of $\cL_N$ on the vacuum state $\O$ is of order $N^2$, this being an indication of the fact that there are still large contributions to the energy hidden among cubic and quartic terms in $\cL^{(3)}_N$ and $\cL^{(4)}_N$.  Since $U_N$ only removes products of the zero-energy mode $\ph_0$, correlations among particles remain in the excitation vector $U_N \psi_N$. Indeed, correlations  play a crucial role in the two dimensional Gross-Pitaevskii regime and carry an energy of order $N^2$.

To take into account the short scale correlation structure on top of the condensate, we consider the ground state $f_{\ell}$ of the Neumann problem 
\be \label{fell-1}
\Big( -\D + \frac 12 V(x) \Big)  f_{\ell}(x) =  \l_{\ell}\, f_{\ell}(x)
\ee
on the ball $|x| \leq e^N \ell$, normalized so that $f_{\ell}(x) = 1$ for $|x|= e^N\ell$. Here and in the following we omit the $N$-dependence in the notation for $f_\ell$ and for $\lambda_\ell$. By scaling, we observe that $f_{\ell}(e^N\cdot)$ satisfies
\[ 
\Big( -\D + \frac {e^{2N}}{2} V(e^N x) \Big) f_{\ell}(e^Nx) =  e^{2N}\l_{\ell}\,  f_{\ell}(e^Nx)
\]
on the ball $|x| \leq \ell$. We choose $0 < \ell < 1/2$, so that the ball of radius $\ell$ is contained in the box $\Lambda= [-1/2 ; 1/2]^2$. We extend then $f_\ell (e^N.)$ to $\Lambda$, by setting $f_{N,\ell} (x) = f_\ell (e^Nx)$, if $|x| \leq \ell$ and $f_{N,\ell} (x) = 1$ for $x \in \Lambda$, with $|x| > \ell$. Then  
\begin{equation}
\label{tlf}
\Big( -\D + \frac {e^{2N}}{2} V(e^N x ) \Big) f_{N,\ell}(x) =  e^{2N}\l_{\ell}\,  f_{N,\ell}(x)\chi_\ell(x)\,,\end{equation}
where $\chi_\ell$ is the characteristic function of the ball of radius $\ell$. The Fourier coefficients of the function $f_{N,\ell}$ are given by 

\begin{equation*}
 \widehat{f}_{N,\ell} (p) := \int_\Lambda f_\ell (e^Nx) e^{-i p \cdot x} dx \end{equation*}

for all $p \in \L^*$.
We introduce also the function $w_\ell(x)=1-f_\ell(x)$ for $|x| \leq e^N\ell$ and extend it by setting $w_\ell(x)=0$ for $|x| >e^N\ell$. 
Its re-scaled version is defined by $w_{N,\ell}: \L \rightarrow \bR$ $w_{N,\ell}(x)=w_\ell(e^Nx)$ if $|x| \leq \ell$ and $w_{N,\ell}=0$ if $x \in \L$ with $|x| > \ell$.  \\

The Fourier coefficients of the re-scaled function $w_{N,\ell}$ are given by

\begin{equation} \label{eq:def-whatNell}
\hw_{N,\ell}(p)=\int_\L w_\ell(e^Nx)e^{-ip\cdot x}dx = e^{-2N}\hw_\ell\left(e^{-N}p\right).
\end{equation}
We find $\widehat{f}_{N,\ell}(p) = \d_{p,0}- e^{-2N}\hw_{\ell}(e^{-N}p)$. From the Neumann problem \eqref{tlf} 
we obtain 
\begin{equation}
\label{scattp}
-p^2e^{-2N}\hw_{\ell}(e^{-N}p) +\frac 12\sum_{q \in \L^*}\widehat V(e^{-N}(p-q))\widehat{f}_{N,\ell}(q) =e^{2N} \l_\ell\sum_{q\in\L^*}\widehat{\chi}_\ell(p-q)\widehat{f}_{N,\ell}(q).
\end{equation}
where we used the notation $\widehat{\chi}_\ell$ for the Fourier coefficients of the characteristic function on the ball of radius $\ell$. Note that $\widehat{\chi}_\ell(p)= \ell^2\, \widehat \chi(\ell p)$ with $\widehat \chi(p)$ the Fourier coefficients of the characteristic function on the ball of radius one.\\

In the next lemma, we collect some important properties of the solution of \eqref{fell-1}.

\begin{lemma} \label{lm:propomega}
Let $V\in L^3(\bR^2)$ be non-negative, compactly supported (with range $R_0$) and spherically symmetric, and denote its scattering length by $\aa$. Fix $0<\ell<1/2$, $N$ sufficiently large and let $f_{\ell}$ denote the solution of  \eqref{tlf}.  Then
\begin{enumerate}[i)]
\item 
\[ \label{eq:bounds-tlfell}
0 \leq f_{\ell}(x) \leq 1   \qquad \forall \, |x| \leq e^N \ell\,.
\]
\item We have
\be  \label{eq:eigenvalue}
\left |\l_{\ell} - \frac{2}{(e^N\ell)^2 \log(e^N\ell/\aa)} \right| \leq  \frac{C} {(e^N\ell)^2 \log^2(e^N\ell/\aa)} 
\ee
\item There exist a constant $C>0$ such that
\be \label{eq:intpotf}  
 \left| 
 \int \di x\, V(x) f_{\ell}(x) - \frac{4\pi}{\log(e^N\ell/\aa)} \right| \leq   \frac{C} {\log^2(e^N\ell/\aa)} 
\ee
	\item There exists a constant $C>0$ such that 
\be \begin{split} \label{eq:boundsomega}
|w_{\ell}(x)| &\leq  \left\{ \begin{array}{ll} C \qquad &\text{if } |x| \leq R_0 \\ 
C  \, \frac{\log(e^N\ell/|x|) }{\log (e^N\ell/\aa)} \quad & \text{if } R_0 \leq |x|\leq e^N \ell \end{array} 
\right. \\
 |\nabla w_{\ell}(x)| &\leq \frac{C}{\log (e^N\ell /\aa)} \frac 1 {|x| + 1}  \qquad \text{for all } \, |x| \leq e^N \ell
\end{split}\ee

\item Let $w_{N,\ell}= 1- f_{N,\ell}$ with $f_{\ell, N}=f_\ell(e^N x)$. Then 
the Fourier coefficients of the function $w_{N,\ell}$ defined in \eqref{eq:def-whatNell} are such that 
\begin{equation}
\label{hwbound}
|\hw_{N,\ell}(p)| \leq \fra{C}{p^2 \log(e^N\ell/\aa)}.
\end{equation}
	\end{enumerate}
\end{lemma}
\begin{proof}
	 The proof of points i)-iv) is deferred in Appendix \ref{App:omega}. To prove point v) we use the scattering equation \eqref{scattp}:
	\begin{equation*}
	\hw_\ell(e^{-N}p) =\frac{e^{2N}}{2p^2}\sum_{q \in \L^*}\widehat V(e^{-N}(p-q))\widehat{f}_{N,\ell}(q) - \fra{e^{4N}}{p^2} \l_\ell\sum_{q\in\L^*}\widehat{\chi}_\ell(p-q)\widehat{f}_{N,\ell}(q).
	\end{equation*}
	Using the fact that $ e^{2N}\l_\ell \leq  C \ell^{-2} |\ln(e^N\ell/\aa)|^{-1}$ and that $0 \leq f_\ell \leq 1$, we end up with
	\begin{equation*}
	\begin{split}
	|\hw_\ell(e^{-N}p)| 	& \leq \frac{e^{2N}}{2p^2}\left[ \big|(\widehat V(e^{-N}\cdot)\ast\widehat{f}_{N,\ell})(p)\big| + 2 e^{2N} \l_\ell \, \big| (\widehat{\chi}_\ell\ast\widehat{f}_{N,\ell})(p)\big|\right]\\
	& \leq \frac{e^{2N}}{2p^2}\left[\int V(x)f_{\ell}(x)dx + C \ell^{-2} |\log(e^N\ell/\aa)|^{-1} \,\int \c_\ell(x)f_{\ell}(e^Nx)dx\right]\\
	& \leq \frac{Ce^{2N}}{p^2\log(e^N\ell/\aa)}.
	\end{split}
	\end{equation*}
\end{proof}

We now define $\check{\eta} : \L \to \bR$ through
\be \label{eq:defchecketa}
\check{\eta}(x)= - N w_{N,\ell}(x) = -N w_\ell(e^N x)  \,.
\ee
With (\ref{eq:boundsomega}) we find 
\be \label{eq:etax0}
|\check{\eta}(x)| \leq  \left\{ \begin{array}{ll} C N \qquad &\text{if } |x| \leq e^{-N} R_0 \\ C \log (\ell / |x| ) \qquad &\text{if } e^{-N} R_0 \leq |x| \leq \ell \end{array} \right. 
\ee
and in particular, recalling that $ e^{-N}R_0 < \ell \leq 1/2$,
\begin{equation}\label{eq:etax}
|\check{\eta} (x)| \leq C  \max (N , \log (\ell / |x|)) \leq C N
\end{equation}
for all $x \in \L$. Using \eqref{eq:etax0} we find 
\begin{equation*}\label{eq:L2eta} \| \eta \|^2 = \| \check{\eta} \|^2 \leq C \int_{|x|\leq \ell} |\log(\ell/|x|)|^2 d^2x \leq C \ell^2 \int_0^1  (\log r)^2 r dr   \leq C \ell^2\,. \end{equation*}
In the following we choose $\ell=N^{-\a}$, for some $\a>0$ to be fixed later, so that 
\begin{equation}\label{eq:etaHL2}
\| \eta \| \leq C N^{-\a}\,.
\end{equation}
This choice of $\ell$ will be crucial for our analysis, as commented below. 
 Notice, on the other hand, that the $H^1$-norms of $\eta$ diverge, as $N \to \infty$. From \eqref{eq:defchecketa} and Lemma \ref{lm:propomega}, part iv) we find 
\begin{equation*} \label{eq:H1eta}
\begin{split}
\| \check{\eta}\|^2_{H_1}  = \int_{|x|\leq \ell} e^{2N} N^2 | (\nabla w_{\ell}) (e^N x )|^2 d^2x &= \int_{|x|\leq e^{N}\ell} N^2 | \nabla w_{\ell}(x)|^2 d^2x \\
& \leq C \int_{|x|\leq e^{N} \ell} \frac{1}{(|x|+1)^2}\,d^2x \leq C N
\end{split}\end{equation*}
for $N \in \bN$ large enough. We denote with $\eta: \L^* \to \bR$ the Fourier transform of $\check{\eta}$, or equivalently
\be \label{eq:defeta}
\eta_p = -N \widehat{w}_{N,\ell} (p) = - N e^{-2N} \widehat{w}_\ell(p/e^N)\,.
\ee
With \eqref{hwbound} we can bound (since $\ell = N^{-\alpha}$) 
\be \label{eq:modetap}
|\eta_p| \leq \frac{C}{|p|^2}
\ee
for all $p \in \L_+^*=2\pi \bZ^2 \backslash \{0\}$, and  for some constant $C>0$ independent of $N$, if $N$ is large enough. From \eqref{eq:etaHL2} we also have
\be \label{eq:etapio}
\| \eta \|_\io \leq C N^{-\a}\,.
\ee
Moreover, (\ref{scattp}) implies the relation  
\begin{equation}\label{eq:eta-scat0}
\begin{split} 
p^2 \eta_p + \frac{N}{2} (\widehat{V} (./e^N) *\widehat{f}_{N,\ell}) (p) = Ne^{2N} \l_\ell (\widehat{\chi}_\ell * \widehat{f}_{N,\ell}) (p)
\end{split} \end{equation}
or equivalently, expressing also the other terms through the coefficients $\eta_p$, 
\begin{equation}\label{eq:eta-scat}
\begin{split} 
p^2 \eta_p + \frac{N}{2} \widehat{V} (p/e^N) & + \frac{1}{2} \sum_{q \in \Lambda^*} \widehat{V} ((p-q)/e^N) \eta_q \\ &\hspace{2cm} = Ne^{2N} \lambda_\ell \widehat{\chi}_\ell (p) +e^{2N}  \lambda_\ell \sum_{q \in \Lambda^*} \widehat{\chi}_\ell (p-q) \eta_q\,.
\end{split} \end{equation}

We will mostly use the coefficients $\eta_p$ with $p\neq 0$. Sometimes, however, it will be useful to have an estimate on $\eta_0$ (because Eq. \eqref{eq:eta-scat} involves $\eta_0$). From \eqref{eq:defeta} and Lemma \ref{lm:propomega}, part iv) we find
\begin{equation}\label{eq:wteta0}  |\eta_0| \leq N \int_{|x|\leq \ell} w_\ell (e^Nx) d^2x \leq C \int_{|x|\leq \ell} \log(\ell/|x|) d^2x  + C N e^{-N}  \leq C \ell^2\,. \end{equation}

With the coefficients \eqref{eq:defeta} we define the antisymmetric operator 
\begin{equation}\label{eq:defB} 
B = \frac{1}{2} \sum_{p\in \L^*_+}  \left( \eta_p b_p^* b_{-p}^* - \bar{\eta}_p b_p b_{-p} \right) \, \end{equation}
and we consider the unitary operator
\begin{equation}\label{eq:eBeta} 
e^{B} = \exp \left[ \frac{1}{2} \sum_{p \in \Lambda^*_+}   \left( \eta_p b_p^* b_{-p}^* - \bar{\eta}_p  b_p b_{-p} \right) \right] \,.
\end{equation}
We refer to operators of the form (\ref{eq:eBeta}) as generalized Bogoliubov transformations.  In contrast with the standard Bogoliubov transformations
\begin{equation}\label{eq:wteBeta} e^{\wt{B}} = \exp \left[  \frac{1}{2} \sum_{p\in \Lambda^*_+}  \left( \eta_p a_p^* a_{-p}^* - \bar{\eta}_p a_p a_{-p} \right) \right] \end{equation}
defined in terms of the standard creation and annihilation operators, operators of the form  (\ref{eq:eBeta}) leave the truncated Fock space $\cF_+^{\leq N}$ invariant. 
On the other hand, while the action of standard Bogoliubov transformation 
on creation and annihilation operators is explicitly given by 
\[\label{eq:act-Bog} e^{-\wt{B}} a_p e^{\wt{B}}  = \cosh (\eta_p) a_p + \sinh (\eta_p) a_{-p}^* \,   \]
there is no such formula describing the action of generalized Bogoliubov transformations.  

Conjugation with \eqref{eq:eBeta} leaves the number of particles essentially invariant, as confirmed by the following lemma.
\begin{lemma}\label{lm:Ngrow} 
Assume $B$ is defined as in \eqref{eq:defB}, with $\eta \in \ell^2 (\L^*)$ and $\eta_p = \eta_{-p}$ for all $p \in \L^*_+$. Then, for every $n \in \bN$ there exists a constant $C > 0$ such that, on $\cF_+^{\leq N}$, 
	\begin{equation}\label{eq:bd-Beta} 
e^{-B} (\cN_+ +1)^{n} e^{B} \leq C e^{C \| \eta \|} (\cN_+ +1)^{n}  \,.
	\end{equation}
as an operator inequality on $\cF^{\leq N}_+$. 
\end{lemma}

The proof of \eqref{eq:bd-Beta} can be found in \cite[Lemma 3.1]{BS} (a similar result has been previously established in \cite{Sei}).

With the generalized Bogoliubov transformation $e^{B} : \cF_+^{\leq N} \to \cF^{\leq N}_+$, we define a new, renormalized, excitation Hamiltonian $\cG_{N,\a} : \cF^{\leq N}_+ \to \cF^{\leq N}_+$ by setting 
\begin{equation}\label{eq:GN} 
\cG_{N,\a} = e^{-B} \cL_N e^{B} = e^{-B} U_N H_N U_N^* e^{B}\,.
\end{equation}

In the next proposition, we collect important properties $\cG_{N,\a}$. We will use the notation 
\begin{equation}\label{eq:KcVN}  
\cK = \sum_{p \in \Lambda^*_+} p^2 a_p^* a_p \qquad \text{and } \quad \cV_N = \frac{1}{2} \sum_{\substack{p,q \in \Lambda_+^*, r \in \Lambda^* : \\ r \not = -p, -q}} \widehat{V} (r/e^N)  a_{p+r}^* a_q^* a_{q+r} a_p \end{equation}
for the kinetic and potential energy operators, restricted on $\cF_+^{\leq N}$, and  $\cH_N = \cK + \cV_N$. We also introduce a renormalized interaction potential $\o_N\in L^\infty(\L)$, which is defined as the function with Fourier coefficients $\widehat{\o}_N  $
	\be  \label{eq:defomegaN}
\widehat{\o}_N(p) :=  g_N \,  \widehat{\chi}(p/N^\a)\,, \qquad  g_N =2  N^{1-2\a} e^{2N} \l_\ell
\ee
for any $p \in \L^*_+$, and
	\be \label{eq:defomega0}
 \widehat{\o}_N(0)= g_N  \widehat{\chi}(0)  =  \pi g_N\,.
\ee
with $\widehat \chi(p)$ the Fourier coefficients of the characteristic function of the ball of radius one. From \eqref{eq:eigenvalue} and $\ell=N^{-\a}$ one has $|g_N|\leq C$. Note in particular that the potential $\widehat \o_N(p)$ decays on momenta of order $N^\a$, which are much smaller than $e^N$.  From Lemma \ref{lm:propomega} parts i) and iii)  we find 
\be \label{eq:omegahat0}
\big| \widehat{\o}_N(0) -   N \| V f_\ell \|_1  \big| \leq \frac C {N}\,, \qquad  
\Big|\, \widehat{\o}_N(0) -  4 \pi \left(1 + \a \,\tfrac{\log N}{N} \right) \, \Big| \leq \frac{ C }{N}\,.
\ee

\begin{prop} \label{prop:GN} Let $V\in L^3(\bR^2)$ be compactly supported, pointwise non-negative and spherically symmetric. Let $\cG_{N,\a}$ be defined as in \eqref{eq:GN} and define
\begin{equation}\label{eq:GNeff} \begin{split}   \cG^{\text{eff}}_{N,\a}  :=\; &    \frac 12 \widehat{\o}_N(0) (N-1)\left(1-\frac{\cN_+}{N}\right) + \left[2 N\widehat{V} (0)-\frac 12 \widehat{\o}_N(0) \right] \, \cN_+ \, \left(1-\frac{\cN_+}{N}\right) \\
&+ \frac 1 2 \sum_{ p\in  \L_+^*}\widehat{\omega}_N(p)(b_pb_{-p}+\hc)   + \sqrt{N} \sum_{\substack{p,q \in \L^*_+ :\\ p + q \not = 0}} \widehat{V} (p/e^N) \left[ b_{p+q}^* a_{-p}^* a_q  + \hc \right]  \\
& +\cH_N  \,.
\end{split} \end{equation}
Then there exists a constant $C > 0$ such that $\cE_{\cG}= \cG_{N,\a} - \cG^\text{eff}_{N,\a}$ is bounded by 
	\begin{equation} \begin{split}\label{eq:GeffE}
| \langle \xi, \cE_{\cG}\, \xi \rangle | \leq\; & C \big( N^{1/2 -\a} + N^{-1}(\log N)^{1/2} \big) \| \cH_N^{1/2}\xi\| \| (\cN_++1)^{1/2}\xi \| \\
&+ C N^{1-\a}\|(\cN_++1)^{1/2}\xi \| ^2 +C \| \xi\|^2
\end{split}\end{equation}
for all $\a>1$, $\xi \in \cF_+^{\leq N}$ and $N\in \bN$ large enough.
\end{prop}

The proof of Prop. \ref{prop:GN} is very similar to the proof of \cite[Prop. 4.2]{BBCS4}. For completeness, we discuss the changes in Appendix \ref{sec:GN}. 

\section{Cubic Renormalization} \label{sec:4}

Conjugation through the generalized Bogoliubov transformation (\ref{eq:wteBeta}) renormalizes constant and off-diagonal quadratic terms on the r.h.s. of (\ref{eq:GNeff}). In order to estimate the number of excitations $\cN_+$ through the energy and show Bose-Einstein condensation, we still 
need to renormalize the diagonal quadratic term (the part proportional to $N \widehat{V} (0) \cN_+$, 
on the first line of (\ref{eq:GNeff})) and the cubic term on the last line of (\ref{eq:GNeff}). To this end, 
we conjugate $\cG_{N,\a}^{\text{eff}}$ with an additional unitary operator, given by the exponential of the anti-symmetric operator 
\begin{equation}\label{eq:defA} A := \frac1{\sqrt N} \sum_{r, v \in \L^*_+} 
\eta_r \big[b^*_{r+v}a^*_{-r}a_v - \text{h.c.}\big]\end{equation}
with $\eta_p$ defined in \eqref{eq:defeta}. 

An important observation is that while conjugation with $e^A$ allows to renormalize the large terms in $\cG_{N,\a}$, it does not substantially change the number of excitations. The following proposition can be proved similarly to \cite[Proposition 5.1]{BBCS3}. 
\begin{prop} \label{prop:ANgrow}
	Suppose that $A$ is defined as in (\ref{eq:defA}). Then, for any $k\in \bN$ there exists a constant $C >0$ such that the operator inequality  
	\[ e^{-A} (\cN_++1)^k e^{A} \leq C (\cN_+ +1)^k   \]
	holds true on $\cF_+^{\leq N}$, for any $\alpha > 0$ (recall the choice $\ell = N^{-\alpha}$ in the definition (\ref{eq:defeta}) of the coefficients $\eta_r$), and $N$ large enough. 
\end{prop}

We will also need to control the growth of the expectation of the energy $\cH_N$ with respect to the cubic conjugation. This is the content of the following proposition, which is proved in subsection \ref{sec:aprioribnds}.

\begin{prop} \label{prop:AHNgrow} Let $A$ be defined as in (\ref{eq:defA}). Then there exists a constant $C > 0$ such that 
	\begin{equation}\label{eq:expAHNexpA}
	e^{-sA} \cH_N e^{sA} \leq C \cH_N + C N (\cN_+ +1)
	\end{equation}
	for all $\a \geq 1$,  $s \in [0;1]$ and $N \in \bN$ large enough.
\end{prop}

We use now the cubic phase $e^{A}$ to introduce a new excitation Hamiltonian, obtained by conjugating the main part $\cG_{N, \a}^{\text{eff}}$ of $\cG_{N,\a}$. We define
\begin{equation} \label{eq:RNell}
\cR_{N, \a}:= e^{-A} \,\cG_{N,\a}^{\text{eff}}\,e^{A} \end{equation}
on a dense subset of $\cF_+^{\leq N}$.
Conjugation with $e^{A}$ renormalizes both the contribution proportional to $\cN_+$ (in the first line in the last line on the r.h.s. of (\ref{eq:GNeff})) and the cubic term on the r.h.s. of (\ref{eq:GNeff}), effectively replacing the singular potential $\widehat{V} (p/e^N)$ by the renormalized potential $\widehat \omega_N(p)$ defined in \eqref{eq:defomegaN}. This follows from the following proposition. 

\begin{prop}\label{prop:RN} Let $V\in L^3(\bR^2)$ be compactly supported, pointwise non-negative and spherically symmetric. Let $\cR_{N,\a}$ be defined in \eqref{eq:RNell} and define 
\begin{equation}\label{eq:cReff}
\begin{split} 
\cR_{N,\a}^{\text{eff}}= &\;  \frac 1 2  (N-1)\,  \widehat \o_N(0)   (1-\cN_+/N) + \frac 1 2  \widehat \o_N(0) \,\cN_+ \left( 1 - \cN_+/N \right) \\
& +  \widehat \o_N(0) \sum_{p\in \L^*_+}a^*_pa_p \Big(1-\frac{\cN_+}{N} \Big)+  \frac 12 \sum_{p\in \L^*_+}  \widehat{\o}_N(p)\big[ b^*_p b^*_{-p} + b_p b_{-p} \big]  \\
& +\frac 1 {\sqrt N} \sum_{\substack{r,v\in \L^*_+:\\ r\neq-v} } \widehat{\o}_N(r)\big[ b^*_{r+v}a^*_{-r} a_v + \text{h.c.}\big]  + \cH_N  \,.
\end{split}
\end{equation} 
Then for $\ell=N^{-\a}$ and $\a>2$ there exists a constant $C>0$ such that $\cE_\cR= \cR_{N,\a}- \cR_{N,\a}^{\text{eff}}$ is bounded by 
\be \label{eq:ReffE}
\pm  \cE_\cR \leq C [ N^{2-\a} + N^{-1/2} (\log N)^{1/2} ](\cH_N +1)  \, , 
\ee
for $N \in \bN$ sufficiently large.
\end{prop}

The proof of Proposition \ref{prop:RN} will be given in Section \ref{sec:RN}. We will also need more detailed information on $\cR_{N,\a}^{\text{eff}}$, as contained in the following proposition. 

\begin{prop}  \label{prop:RN-loc}
Let $\cR_{N,\a}^{\text{eff}}$ be defined in \eqref{eq:cReff}. Then, for every $c > 0$ there is a constant $C > 0$ (large enough) such that 
\begin{equation}
\begin{split} \label{eq:cRN-LB}
\cR_{N,\a}^{\text{eff}}\geq & \;2\pi N + \frac{\widehat{\omega}_N (0)}{2} \,\cN_+ + \frac c {\log N}\, \cH_N - C (\log N)^2 \, \frac{\,\cN_+^2}{N}\,  - C 
\end{split}\end{equation} 
for all $\a>2$ and $N \in \bN$ large enough. 

Moreover, let $f,g : \bR \to [0;1]$  be smooth, with $f^2 (x) + g^2 (x) =1$ for all $x \in \bR$. For $M \in \bN$, let $f_M := f(\cN_+/M)$ and $g_M:= g(\cN_+/M)$. Then there exists $C > 0$ such that  
	\be \label{eq:IMS}
	\cR_{N,\a}^{\text{eff}} =  f_M\, \cR_{N, \a}^{\text{eff}}\, f_M + g_M\, \cR_{N, \a}^{\text{eff}}\, g_M + \Theta_{M}
	\ee
	with 
	\begin{equation*}
	\pm \Theta_M \leq \frac{C\log N}{M^2}\big(\|f'\|^2_{\infty} +\|g'\|^2_{\infty}\big) \big( \cH_N +1 \big)
	\end{equation*}
	for all $\a > 2$, $M \in \bN$ and $N \in \bN$ large enough. 
\end{prop}

\begin{proof}
 
From \eqref{eq:cReff}, using that $| \widehat \o_N(0) |\leq C$ we have 
\begin{equation}\label{eq:RN-LB1}
\begin{split} 
\cR_{N,\a}^{\text{eff}} \geq &\;  \frac N 2 \,  \widehat \o_N(0)   +   \widehat \o_N(0) \,\cN_+  +  \frac 12 \sum_{p\in \L^*_+}  \widehat{\o}_N(p)\big[ b^*_p b^*_{-p} + b_p b_{-p} \big] \\
& +\frac 1 {\sqrt N} \sum_{\substack{r,v\in \L^*_+:\\ r\neq-v} } \widehat{\o}_N(r)\big[ b^*_{r+v}a^*_{-r} a_v + \text{h.c.}\big] +  \cH_N - C \,\frac{\cN_+^2}{N} - C\,. 
\end{split}
\end{equation} 
For the cubic term on the r.h.s. of \eqref{eq:RN-LB1}, with
\be \label{eq:intV-div-p2}
\sum_{p \in \L^*_+} \frac{|\widehat \o_N(p)|^2}{p^2} \leq C \log N
\ee we can bound  
\be \begin{split} \label{eq:RN-LB3}
 \Big|  \frac 1 {\sqrt N} & \sum_{\substack{r,v\in \L^*_+\\ r \neq -v}}  \widehat{\o}_N(r) \langle \xi, b^*_{r+v} a^*_{-r} a_v\xi \rangle \Big|  \\ & \leq  \frac 1 {\sqrt N} \sum_{\substack{r,v\in \L^*_+\\ r \neq -v}}|\widehat{\o}_N(r) | \| (\cN_+ + 1)^{-1/2}  b_{r+v} a_{-r} \xi \|  \| (\cN_+ + 1)^{1/2} a_v \xi \| \\
&   \leq  \frac 1 {\sqrt N} \, \bigg[ \sum_{\substack{r,v\in \L^*_+\\ r \neq -v}} |r|^2 \| (\cN_++1)^{-1/2}b_{r+v} a_{-r} \xi \|^2 \bigg]^{1/2} 
\\& \hskip 2cm \times 
\bigg[ \sum_{\substack{r,v\in \L^*_+\\ r \neq -v}} \frac{|\widehat{\o}_N(r) |^2}{|r|^2}   \| (\cN_++1)^{1/2}a_v \xi \|^2 \bigg]^{1/2}\\
& \leq  \frac{C (\log N)^{1/2}}{\sqrt N}  \,\| \cK^{1/2}\xi \| \| (\cN+1)\xi\|\,. 
\end{split}\ee
As for the off-diagonal quadratic term on the r.h.s of \eqref{eq:RN-LB1}, we combine it with part of the kinetic energy to estimate. For any $0 < \mu < 1$, we have 
\[ \begin{split} \frac{1}{2} \sum_{p \in \Lambda_+^*} \widehat{\o}_N (p) & \big[ b_p^* b_{-p}^* + b_{-p} b_p \big] + (1- \mu) \sum_{p \in \Lambda^*_+} p^2 a_p^* a_p \\ = \; &(1-\mu) \sum_{p \in \Lambda_+^*} p^2 \left[ b_p^* + \frac{ \widehat{\o}_N (p)}{2(1-\mu) p^2} b_{-p} \right] \left[ b_p + \frac{ \widehat{\o}_N (p)}{2(1-\mu) p^2} b^*_{-p} \right] \\ &- \frac{1}{4(1-\mu)} \sum_{p \in \Lambda^*_+} \frac{|\widehat{\o}_N (p)|^2}{p^2} b_p b_p^* + (1-\mu) \sum_{p \in \Lambda^*_+} p^2 a_p^* \frac{\cN_+}{N} a_p \end{split} \]
since $a_p^* a_p - b_p^* b_p = a_p^* (\cN_+ / N) a_p$. With (\ref{eq:comm-bp}), we conclude that 
\begin{equation*}
	\begin{split} 
\frac{1}{2} \sum_{p \in \Lambda_+^*} \widehat{\o}_N (p) &\big[ b_p^* b_{-p}^* + b_{-p} b_p \big] + (1- \mu) \sum_{p \in \Lambda^*_+} p^2 a_p^* a_p \\ \geq \; &- \frac{1}{4(1-\mu)} \sum_{p \in \Lambda^*_+} \frac{|\widehat{\o}_N (p)|^2}{p^2} a^*_p  a_p - \frac{1}{4(1-\mu)} \sum_{p \in \Lambda^*_+} \frac{|\widehat{\o}_N (p)|^2}{p^2}\,. \end{split} \end{equation*} 
With the choice $\mu = C / \log N$ and with (\ref{eq:intV-div-p2}), we obtain 
\[  \begin{split} 
\frac{1}{2} \sum_{p \in \Lambda_+^*} \widehat{\o}_N (p)& \big[ b_p^* b_{-p}^* + b_{-p} b_p \big] + (1- \mu) \sum_{p \in \Lambda^*_+} p^2 a_p^* a_p \\ \geq \; &- \frac{1}{4(1-\mu)} \sum_{p \in \Lambda^*_+} \frac{|\widehat{\o}_N (p)|^2}{p^2} a^*_p  a_p - \frac{1}{4} \sum_{p \in \Lambda^*_+} \frac{|\widehat{\o}_N (p)|^2}{p^2}  - C\,. \end{split} \]
To bound the first terms on the r.h.s. of the last equation, we use the term $\widehat{\o}_N (0) \cN_+$, in (\ref{eq:RN-LB1}). To this end, we observe that, with (\ref{eq:omegahat0}), 
\[ \frac{|\widehat{\o}_N (p)|^2}{4 (1-\mu) p^2} \leq  \frac{|\widehat{\o}_N (0)|^2}{4 (1-\mu) p^2} \leq \frac{\widehat{\o}_N (0)}{4 ( 1-\mu) \pi} \left(1 + C \frac{\log N}{N} \right) \leq \frac{\widehat{\o}_N (0)}{2} \]
for every $p \in \L^*_+$ (notice that $|p| \geq 2\pi$, for every $p \in \Lambda^*_+$) and for $N$ large enough (recall the choice $\mu = c / \log N$). From  (\ref{eq:RN-LB1}), we find 
\begin{equation}\label{eq:RN-bd}
\begin{split} 
\cR_{N,\a}^{\text{eff}} \geq &\;  \frac N 2 \,  \widehat \o_N(0)   - \frac{1}{4} \sum_{p \in \Lambda_+^*} \frac{|\widehat{\o}_N (p)|^2}{p^2} +   \frac{\widehat \o_N(0)}{2} \,\cN_+  + \frac{c}{\log N} \cH_N  - C \frac{(\log N)^2}{N} \cN_+^2 - C. 
\end{split}
\end{equation} 
Let us now consider the second term on the r.h.s more carefully. Using that, from (\ref{eq:defomegaN}), $\widehat{\o}_N (p) = g_N \widehat{\chi} (p/N^\alpha)$, we can bound, for any fixed $K > 0$,    
\[ \frac{1}{4} \sum_{p \in \Lambda_+^*} \frac{|\widehat{\o}_N (p)|^2}{p^2} \leq C + \frac{1}{4} \sum_{\substack{p \in \Lambda_+^* : \\ K < |p| \leq N^\alpha}} \frac{|\widehat{\o}_N (p)|^2}{p^2}\,. \]
With $|\widehat{\o}_N (p) - \widehat{\o}_N (0)|\leq C |p|/ N^\alpha$, we obtain 
\begin{equation} \label{eq:bd-riem}  \frac{1}{4} \sum_{p \in \Lambda_+^*} \frac{|\widehat{\o}_N (p)|^2}{p^2} \leq C + \frac{|\widehat{\o}_N (0)|^2}{4} \sum_{\substack{p \in \Lambda_+^* : \\ K < |p| \leq N^\alpha}} \frac{1}{p^2} \leq C + 4 \pi^2 \sum_{\substack{p \in \Lambda^*_+ : \\ K < |p| \leq N^\alpha}} \frac{1}{p^2}\,. \end{equation} 
For $q \in \bR^2$, let us define $h(q) = 1/p^2$, if $q$ is contained in the square of side length $2\pi$ centered at $p \in \Lambda^*_+$ (with an arbitrary choice on the boundary of the squares). 
We can then estimate, for $K$ large enough,  
\[ 4\pi^2 \sum_{\substack{p \in \Lambda_+^* : \\ K < |p| \leq N^\alpha}} \frac{1}{p^2} \leq \int_{K/2 < |q| \leq N^\alpha + K} h(q) dq\,. \]
For $q$ in the square centered at $p \in \Lambda^*_+$, we bound 
\[ \left| h(q) - \frac{1}{q^2} \right| = \frac{|p^2 - q^2|}{p^2  \,q^2} \leq \frac{C}{|q|^3}\,. \]
Hence 
\[  4 \pi^2 \sum_{\substack{p \in \Lambda_+^* : \\ K < |p| \leq N^\alpha}} \frac{1}{p^2} \leq \int_{K/2 < |q| < N^\alpha + K} \frac{1}{q^2} dq + C \leq 2\pi \alpha \log N + C\,. \]
Inserting in (\ref{eq:bd-riem}), we conclude that 
\[ \frac{1}{4} \sum_{p \in \Lambda_+^*} \frac{|\widehat{\o}_N (p)|^2}{p^2} \leq 2\pi \alpha \log N +C\,. \]
Combining the last bound with (\ref{eq:omegahat0}) (and noticing that the contribution proportional to $\log N$ cancels exactly), from (\ref{eq:RN-bd}) we obtain 
\[ \cR_{N,\a}^{\text{eff}} \geq \;  2\pi N  +   \frac{\widehat \o_N(0)}{2} \,\cN_+  + \frac{c}{\log N} \cH_N  - C \frac{(\log N)^2}{N} \cN_+^2 - C \]
which proves (\ref{eq:cRN-LB}). 

Next we prove \eqref{eq:IMS}. From \eqref{eq:RN-LB1}, with the bounds (\ref{eq:RN-LB3}) and since, by (\ref{eq:intV-div-p2}),  
\[ \begin{split} \left| \sum_{p \in \Lambda_+^*} \widehat{\o}_N (p) \langle \xi , b_p^* b_{-p}^* \xi \rangle \right|  &\leq \sum_{p \in \Lambda_+^*} |\widehat{\o}_N (p) \| b_p \xi \| \| (\cN_+ + 1)^{1/2} \xi \| \\ &\leq  \left[ \sum_{p \in \Lambda^*_+} \frac{|\widehat{\o}_N (p)|^2}{p^2} \right]^{1/2} \| (\cN_+ + 1)^{1/2} \xi \| \| \cK^{1/2} \xi \|  \\ &\leq  C (\log N)^{1/2} \| (\cN_+ + 1)^{1/2} \xi \| \| \cK^{1/2} \xi \| \end{split} \]
it follows that 
\be \label{Rtheta2}
\cR_{N,\a}^{\text{eff}} = 2\pi N +\cH_N + \theta_{N,\a}
\ee
where for arbitrary $ \d > 0$, there exists a constant $C>0$ such that 
	\be  \label{eq:thetaNal}
\pm  \theta_{N,\a} \leq  \d \cH_N  + C (\log N) \,(\cN_+ +1) \,. \ee
We now note that for $f: \bR \to \bR$ smooth and bounded and $\theta_{N,\a}$ defined above, there exists a constant $C>0$ such that 
\be \label{eq:fMtheta}
\pm [f(\cN_+/M), [f(\cN_+/M),  \theta_{N,\a}]] \leq C\, \frac{\log N}{M^{2}} \| f' \|_\io^2 (\cH_N +1)
\ee
for all $\a>2$ and $N \in \bN$ large enough. The proof of \eqref{eq:fMtheta} follows analogously to the one for \eqref{eq:thetaNal}, since  the bounds leading to \eqref{eq:thetaNal} remain true if we replace the operators $b_p^\#$, $\#=\{ \cdot, *\}$, and $a_p^* a_q$ with  $[f(\cN_+/M),  [f(\cN_+/M), b^\#_p ]]$ or  $[f(\cN_+/M),  [f(\cN_+/M), a^*_p a_q ]]$ respectively, provided we multiply the r.h.s. by an additional factor $M^{-2}\|f'\|_\io^2$, since, for example
\[ \label{eq:doublecommfb}
 [f(\cN_+/M),  [f(\cN_+/M), b_p ]] = \big( f(\cN_+/M) - f((\cN_++1)/M) \big)^2 b_p
\]
and $\| f(\cN_+/M) - f((\cN_++1)/M) \| \leq C M^{-1} \| f'\|_\io$. With an explicit computation we obtain
	\begin{equation*}
	\begin{split}
	\cR_{N,\a}^{\text{eff}}=f_M \cR_{N,\a}^{\text{eff}}f_M +g_M \cR_{N,\a}^{\text{eff}}g_M+\frac{1}{2}\Big([f_M,[f_M,\cR^{\text{eff}}_{N,\a}]]+[g_M,[g_M,\cR^{\text{eff}}_{N,\a}]]\Big)
	\end{split}
	\end{equation*} 
Writing $\cR_{N,\a}^{\text{eff}}$ as in \eqref{Rtheta2} and using \eqref{eq:fMtheta} we get 
	\begin{equation*}
	\begin{split}
	\pm\Big([f_M,[f_M,\cR_{N,\a}^{\text{eff}}]]+[g_M,[g_M,\cR_{N,\a}^{\text{eff}}]]\Big)\leq \frac{C\log N}{M^2} \big(\|f'\|^2_{\infty} +\|g'\|^2_{\infty}\big) \big( \cH_N +1 \big)\,.
	\end{split}
	\end{equation*}
\end{proof}

\section{Proof of Theorem \ref{thm:main}}
\label{sec:main}

The next proposition combines the results of Prop. \ref{prop:GN}, Prop. \ref{prop:RN} and Prop. \ref{prop:RN-loc}.  Its proof makes use of localization in the number of particle and is an adaptation of the proof of \cite[Proposition 6.1]{BBCS3}. The main difference w.r.t. \cite{BBCS3} is that here we need to localize on sectors of $\cF^{\leq N}$ where the number of particles is $o (N)$, in the limit $N \to \infty$.

\begin{prop}
	Let $V\in L^3(\bR^2)$ be compactly supported, pointwise non-negative and spherically symmetric. Let $\cG_{N,\a}$ be the renormalized excitation Hamiltonian defined as in \eqref{eq:GN}. Then, for every $\a \geq 5/2$, there exist constants $C,c > 0$ such that 
	\begin{equation}\label{eq:cGN-fin} 
\cG_{N,\a} -2\pi N \geq c\, \cN_+ - C  \end{equation}
	for all $N \in \bN$ sufficiently large. 
\end{prop}
\begin{proof} Let $f,g: \bR \to [0;1]$ be smooth, with $f^2 (x) + g^2 (x)= 1$ for all $x \in \bR$. Moreover, assume that $f (x) = 0$ for $x > 1$ and $f (x) = 1$ for $x < 1/2$. For a small $\eps>0$, we fix $M  = N^{1-\e}$ and we set $f_M = f (\cN_+ / M), g_M = g (\cN_+ / M)$. It follows from Prop. \ref{prop:RN-loc} that 
	\begin{equation}\label{eq:cGN-1}\begin{split} \cR_{N,\a}^{\text{eff}}  -2\pi N  \geq\; & f_M \big(\cR^{\text{eff}}_{N,\a} -2\pi N \big) f_M + g_M \big(\cR^{\text{eff}}_{N,\a} -2\pi N \big) g_M \\ & - CN^{2\e-2} (\log N)  (\cH_N + 1) \end{split}\end{equation}
	Let us consider the first term on the r.h.s. of (\ref{eq:cGN-1}). From Prop. \ref{prop:RN-loc}, for all $\a>2$ there exist $c, C>0$ such that 
	\be \label{eq:first-bd} 
\cR_{N,\a}^{\text{eff}} -2\pi N \geq c\, \cN_+ - \fra{C}{N} \,(\log N)^2 \, \cN_+^{\,2}  - C  \,.\ee
On the other hand, with \eqref{Rtheta2} and \eqref{eq:thetaNal} we also find
\begin{equation}\label{eq:second-bd} 
\cR_{N,\a}^{\text{eff}} -2\pi N \geq  c \cH_N  - C (\log N)\,(\cN_++1)  \end{equation}
	for all $\a > 2$ and $N$ large enough. Moreover, due to the choice $M=N^{1-\eps}$, we have
	\[ 
 \fra{(\log N)^2}{N} f_M \cN_+^2 f_M \leq  \frac{(\log N)^{2} }{N^\eps}  f^2_M \cN_+  \,.
 \]
With the last bound, Eq. \eqref{eq:first-bd} implies that 
\be \label{eq:fMGN} 
 f_M \Big(\cR^{\text{eff}}_{N,\a} -2\pi N\Big) f_M  \geq c  f^2_M \cN_+ - C 
\ee
for $N$ large enough.

	Let us next consider the second term on the r.h.s. of (\ref{eq:cGN-1}). We claim that there exists a constant $c > 0$ such that 
	\begin{equation}\label{eq:gMGb} g_M \Big(\cR^{\text{eff}}_{N,\a} -2\pi N  \Big) g_M  \geq c N g_M^2 \end{equation}
		for all $N$ sufficiently large. To prove (\ref{eq:gMGb}) we observe that, since $g(x) = 0$ for all $x \leq 1/2$,   
	\[ \begin{split}g_M \Big( \cR^{\text{eff}}_{N,\a} -2\pi N&   \Big) g_M \geq \left[ \inf_{\xi \in \cF_{\geq M/2}^{\leq N} : \| \xi \| = 1} \frac{1}{N} \langle \xi, \cR^{\text{eff}}_{N,\a} \xi \rangle - 2\pi \right] N g_M^2 \end{split}\]
	where $\cF_{\geq M/2}^{\leq N} = \{ \xi \in \cF_+^{\leq N} : \xi = \chi (\cN_+ \geq M/2) \xi \}$ 
	is the subspace of $\cF_+^{\leq N}$ where states with at least $M/2$ excitations are described (recall that $M = N^{1-\e}$). To prove (\ref{eq:gMGb}) it is enough to show that there exists $C > 0$ with 
	\begin{equation}\label{eq:gMGN2}  \inf_{\xi \in \cF_{\geq M/2}^{\leq N} : \| \xi \| = 1} \frac{1}{N} \langle \xi, \cR^{\text{eff}}_{N,\a} \xi \rangle -2\pi \geq C \end{equation}
	for all $N$ large enough. On the other hand, using the definitions of $\cG_{N,\a}$ in \eqref{eq:GNeff}, $\cR_{N,\a}$ and $\cR^{\text{eff}}_{N,\a}$ in \eqref{eq:cReff}, we obtain that the ground state energy $E_N$ of the system is given by
\[ \label{eq:ENcalGcalR}
\begin{split}
E_N & =  \inf_{\xi \in \cF_+^{\leq N} : \| \xi \|=1}    \bmedia{\xi, e^{-A}\cG_{N,\a}e^A \xi}= \inf_{\xi \in \cF_+^{\leq N} : \| \xi \|=1}    \bmedia{\xi, \big(\cR^{\text{eff}}_{N,\a} + \cE_L \big) \xi} 
\end{split}\]
with $\cE_{L} =  \cE_{\cR} +e^{-A} \cE_{\cG}e^A$. The bounds \eqref{eq:GeffE} and \eqref{eq:ReffE}, together with Prop. \ref{prop:ANgrow} and Prop. \ref{prop:AHNgrow}, imply that for any $\a \geq 5/2$ there exists $C>0$ such that 
\[ \label{eq:cEL}
\begin{split} 
\pm \cE_{L}  & \leq C N^{-1/2}(\log N)^{1/2} \big[ ( \cH_N +1 ) +  e^{-A}\big( N^{-1}(\cH_N+1) + (\cN_++1) \big) e^A \big] + C \\
 & \leq C N^{-1/2}(\log N)^{1/2} ( \cH_N +1 ) + C 
\end{split} \]
With \eqref{eq:second-bd}  we obtain 
\be \label{eq:cEL2}
\pm \cE_L 
\leq  C N^{-1/2} (\log N)^{1/2}  \big(\cR^{\text{eff}}_{N,\a} -2\pi N \big) + C N^{-1/2} (\log N)^{3/2} \cN_+ + C \,,
\ee
and therefore, with $\cN_+\leq N$ 
\[
E_N - 2\pi N \leq   C   \inf_{\xi \in \cF_+^{\leq N} : \| \xi \|=1}    \bmedia{\xi, \big(\cR^{\text{eff}}_{N,\a}  - 2 \pi N \big) \xi}  +C N^{1/2} (\log N)^{3/2} + C\,.
\]
From the result \eqref{eq:en-ti} of \cite{LSeY2d, LS-BEC, LSe-GP-rotating}
	\[ \begin{split} \inf_{\xi \in \cF_{\geq M/2}^{\leq N} : \| \xi \| = 1} \frac{1}{N} \langle \xi,  \cR^{\text{eff}}_{N,\a}  \xi \rangle -2 \pi  & \geq  \inf_{\xi \in \cF_{+}^{\leq N} : \| \xi \| = 1} \frac{1}{N} \langle \xi, \big(\cR^{\text{eff}}_{N,\a} - 2 \pi N \big) \xi \rangle  \\
& \geq  c \left( \frac{E_N}{N} -2\pi \right) - \frac{C}{\sqrt N}\,(\log N)^{3/2} - C N^{-1} \to 0\end{split} \]
	as $N \to \infty$. If we assume by contradiction that (\ref{eq:gMGN2}) does not hold true, then we can find a subsequence $N_j \to \infty$ with
	\[  
\inf_{\xi \in \cF_{\geq M_j/2}^{\leq N_j} : \| \xi \| = 1} \frac{1}{N_j} \langle \xi, \cR_{N_j ,\a}^{\text{eff}} \xi \rangle -2\pi \to 0 
\]
	as $j \to \infty$ (here we used the notation $M_j =  N_j^{1-\e}$). This implies that there exists a sequence $\tl \xi_{N_j} \in \cF^{\leq N_j}_{ \geq M_j /2}$ with $\| \tl \xi_{N_j} \| = 1$ for all $j \in \bN$ such that 
	\begin{equation*}
	\lim_{j \to \infty} \frac{1}{N_j} \langle \tl \xi_{N_j}, \cR^{\text{eff}}_{N_j, \a} \tl \xi_{N_j} \rangle = 2\pi  \, . \end{equation*}
On the other hand, using the relation $\cR^{\text{eff}}_{N_j,\a} = e^{-A}\cG_{N_j,\a} e^A - \cE_{L,j}$ with $\cE_{L,j}$ satisfying the bound \eqref{eq:cEL2} (with $\cN_+ \leq N_j$), we obtain that there exist constants $c_1, c_2, C>0$ such that
\[ \begin{split} \label{eq:relGReff}
& c_1 \langle \tl \xi_{N_j}, \big( \cR^{\text{eff}}_{N,\a}- 2 \pi N_j \big) \tl \xi_{N_j} \rangle  - C N_j^{1/2} (\log N_j)^{3/2}   \\
&\hskip 2.5cm \leq \langle  e^{A} \tl \xi_{N_j}, \big(\cG_{N_j,\a}  -\;  2 \pi N_j \big)  e^A \tl \xi_{N_j} \rangle \\
& \hskip 3.5cm  \leq    c_2 \langle \tl \xi_{N_j}, \big(\cR^{\text{eff}}_{N,\a} - 2 \pi N_j \big)  \tl \xi_{N_j} \rangle  + C N_j^{1/2} (\log N_j)^{3/2}  
\end{split}\]
Hence for $\xi_{N_j} = e^{A} \tl \xi_{N_j}$ we have
\[ \label{eq:limGNa1}
\lim_{N_j \to \infty} \frac{1}{N_j} \langle  \xi_{N_j}, \cG_{N_j,\a}  \xi_{N_j} \rangle = 2\pi \,.
\]
Let now $S:= \{N_j: j\in \bN\} \subset\bN$ and denote by $\xi_N$ a normalized minimizer of $\cG_{N,\a}$ for all $N\in \bN\setminus S$. Setting $\psi_N = U_N^* e^{B} \xi_N$, for all $N \in \bN$, we obtain that $\| \psi_N \| = 1$ and that 
\be \label{eq:limGNa2}
\lim_{N \to \infty} \frac{1}{N} \langle \psi_N, H_N \psi_N \rangle = \lim_{N \to \infty} \frac{1}{N} \langle \xi_N, \cG_{N,\a} \xi_N \rangle = 2\pi  
\ee
Eq. \eqref{eq:limGNa2} shows that the sequence $\psi_N$ is an approximate ground state of $H_N$. From \eqref{eq:BEC1}, we conclude that $\psi_N$ exhibits complete Bose-Einstein condensation in the zero-momentum mode $\ph_0$, and in particular that there exist $\bar \d >0$ such that
	\[ \left|1 - \langle \ph_0, \gamma_N \ph_0 \rangle \right|\leq CN^{-\bar{\d}}\,. \]
	Using Lemma \ref{lm:Ngrow}, Prop. \ref{prop:ANgrow} and the rules (\ref{eq:U-rules}), we observe that 
	\begin{equation}\label{eq:contra1} \begin{split} \frac{1}{N} \langle  \xi_N, \cN_+  \xi_N \rangle &= \frac{1}{N} \langle e^{-B} U_N \psi_N , \cN_+  e^{-B} U_N \psi_N \rangle \\ &\leq \frac{C}{N} \langle \psi_N , U_N^* (\cN_+ +1) U_N \psi_N \rangle \\
&= \frac{C}{N} + C \left[ 1 - \frac{1}{N} \langle \psi_N, a^* (\ph_0) a(\ph_0) \psi_N \rangle \right] \\ &= \frac{C}{N} + C \left[ 1 - \langle \ph_0 , \gamma_N \ph_0 \rangle \right]  \leq CN^{-\bar{\d}}\end{split} \end{equation}
	as $N \to \infty$. 

On the other hand, for $N \in S = \{ N_j : j \in \bN \}$, we have $\xi_N = \chi (\cN_+ \geq M/2) \xi_N$ and therefore
	\[ \frac{1}{N} \langle \xi_N, \cN_+ \xi_N \rangle \geq \frac{M}{2N} = \frac{N^{-\e}}{2}\,. \]
Choosing  $\e < \bar{\d}$ and $N$ large enough we get a contradiction with (\ref{eq:contra1}). This proves (\ref{eq:gMGN2}), (\ref{eq:gMGb}) and therefore also
	\begin{equation}\label{eq:gM-bd}  g_M \Big( \cR^{\text{eff}}_{N,\a} -2\pi N  \Big) g_M \geq c \cN_+ g_M^2\,.
 \end{equation}
	Inserting (\ref{eq:fMGN}) and (\ref{eq:gM-bd}) on the r.h.s. of (\ref{eq:cGN-1}), we obtain that
	\begin{equation}\label{eq:lbd} \cR_{N,\a}^{\text{eff}} -2\pi N \geq c \cN_+  - C (\log N ) N^{2\e-2}  (\cH_N + 1)  - C 
\end{equation}
	for $N$ large enough. With (\ref{eq:second-bd}), (\ref{eq:lbd}) implies 
\[\label{eq:cRNeff-loc} 
\cR^{\text{eff}}_{N,\a} -2\pi N \geq c \cN_+ - C.
\]
To conclude, we use the relation $e^{-A}\cG_{N,\a}e^A=\cR^{\text{eff}}_{N,\a} + \cE_L$ and the bound \eqref{eq:cEL2}. We have that  for $\a\geq 5/2$ there exist $c, C>0$ such that
\[ \begin{split}
\cG_{N,\a} - 2 \pi N &\geq  c e^A \big(\cR^{\text{eff}}_{N,\a} -2\pi N \big) e^{-A} - C N^{-1/2} (\log N)^{3/2} e^A \cN_+ e^A - C   \\
&\geq c \,e^{A} \cN_+ e^{-A} -C  \geq c \cN_+ - C  
\end{split}\]
where we used (\ref{eq:lbd}) and Prop. \ref{prop:ANgrow}. 
\end{proof}

We are now ready to show our main theorem.
\begin{proof}[Proof of Theorem \ref{thm:main}]
	Let $E_N$ be the ground state energy of $H_N$.  
Evaluating (\ref{eq:GNeff}) and \eqref{eq:GeffE} on the vacuum $\O \in \cF^{\leq N}_+$ and using \eqref{eq:defomega0}, we obtain the upper bound 
\[
E_N \leq 2\pi N+ C \log N \,.
\] 
With Eq. (\ref{eq:cGN-fin}) we also find the lower bound $E_N \geq 2\pi N  - C $. This proves (\ref{eq:Enbd}).  
	
\vskip 0.2cm

	Let now $\psi_N \in L^2_s (\Lambda^N)$ with $\| \psi_N \| =1$ and
	\be \label{eq:boundenergy} \langle \psi_N , H_N \psi_N \rangle \leq 2\pi N  + K\,. 
\ee
	We define the excitation vector $\xi_N = e^{-B} U_N \psi_N$. Then $\| \xi_N \| = 1$ and, recalling that $\cG_{N,\a} = e^{-B}  U_N H_N U_N^* e^{B}$ we have, with (\ref{eq:cGN-fin}), 
\be \begin{split}  \label{eq:condensation}
\bmedia{\psi_N, (H_N - 2\pi N) \psi_N} &\geq \bmedia{\xi_N, (\cG_{N,\a} - 2\pi N) \xi_N}  \geq c \centering \bmedia{\xi_N, \cN_+ \xi_N} - C \,.
\end{split}\ee
From Eqs. \eqref{eq:boundenergy} and \eqref{eq:condensation} we conclude that
	\[ \langle \xi_N, \cN_+ \xi_N \rangle  \leq C  (1 +K)\,. \]
	If $\gamma_N$ denotes the one-particle reduced density matrix associated with $\psi_N$, using Lemma  \ref{lm:Ngrow} we obtain 
	\[ \begin{split} 
	1 - \langle \ph_0, \gamma_N \ph_0 \rangle &= 1 - \frac{1}{N} \langle \psi_N, a^* (\ph_0) a (\ph_0) \psi_N \rangle \\ &= 1 - \frac{1}{N} \langle U_N^* e^{B} \xi_N, a^* (\ph_0) a(\ph_0) U_N^* e^{B} \xi_N \rangle \\ &= \frac{1}{N} \langle e^{B} \xi_N, \cN_+ e^{B} \xi_N \rangle \leq \frac{C}{N} \langle \xi_N , \cN_+ \xi_N \rangle \leq \frac{C( 1 +K)}{N} \end{split} \]
	which concludes the proof of (\ref{eq:BEC}). 
\end{proof}

\section{Analysis of the excitation Hamiltonian $\cR_{N} $} \label{sec:RN}

In this section, we show Prop. \ref{prop:RN}, where we establish a lower bound for the operator $\cR_{N,\a} = e^{-A} \cG_{N,\a}^\text{eff} e^A$, with $\cG^\text{eff}_{N,\a}$ as defined in 
(\ref{eq:GNeff}) and with 
\begin{equation}\label{eq:Aell8} A = \frac1{\sqrt N} \sum_{r, v \in \L^*_+} \eta_r \big[b^*_{r+v}a^*_{-r}a_v - \text{h.c.}\big] \, . \end{equation}  
We decompose 
\be \label{eq:GNeff-deco}
 \cG_{N,\a}^\text{eff} =  \cO_{N} + \cK  +\cZ_N+ \cC_{N} + \cV_N \ee
with $\cK$ and $\cV_N$ as in (\ref{eq:KcVN}),  and with 
\begin{equation}\begin{split}\label{eq:GNeff-deco2}
\cO_{N}  =&\;\frac 1 2 \widehat{\o}_N(0) (N-1) \Big(1-\frac{\cN_+}{N}\Big) + \big[ 2 N\widehat V(0)- \frac 1 2 \widehat{\o}_N(0) \big]\cN_+ \Big(1-\frac{\cN_+}{N}\Big) , \\
\cZ_N = &\; \frac 1 2 \sum_{ p\in  \L_+^*}\widehat{\o}_N(p)(b_pb_{-p}+\hc)\\ 
\cC_{N} =&\;  \sqrt{N} \sum_{p,q \in \L^*_+ : p + q \not = 0} \widehat{V} (p/e^N) \left[ b_{p+q}^* a_{-p}^* a_q  + \hc \right] \, .		
\end{split}\end{equation} 
We will analyze the conjugation of all terms on the r.h.s. of \eqref{eq:GNeff-deco} in Subsections 
 \ref{sub:RN-cO}--\ref{sub:Rfin}. The estimates emerging from these subsections will then be combined in Subsection \ref{sub:Rfin} to conclude the proof of Prop. \ref{prop:RN}. Throughout the section, we will need Prop. \ref{prop:AHNgrow} to control the growth of the expectation of the energy $\cH_N = \cK + \cV_N$ under the action of (\ref{eq:Aell8}); the proof of Prop.  \ref{prop:AHNgrow} is contained in Subsection \ref{sec:aprioribnds}. 

In this section, we will always assume that $V \in L^3 (\bR^2)$ is compactly supported, pointwise non-negative and spherically symmetric.

\subsection{A priori bounds on the energy}\label{sec:aprioribnds}

In this section, we show Prop. \ref{prop:AHNgrow}. To this end, we will need the following proposition.  
\begin{prop}\label{prop:commAcVN} Let $\cV_N$ and  $A$  be defined in \eqref{eq:KcVN} and \eqref{eq:defA} respectively. Then, there exists a constant $C > 0$ such that  
	\begin{equation*}\label{eq:commAcVN} 
	[\cV_N,A] = \frac{1}{N^{1/2}}\sum_{\substack{u,r,v \in \Lambda_+^*\\ u\neq-v}} \widehat V((u-r)/e^N) \eta_r\big[b^*_{u+v}a^*_{-u} a_v +\hc\big] + \delta_{\cV_N}  \end{equation*}
	where 
	\begin{equation}\label{eq:RN-cVN1} 
	\begin{split}
	| \langle \xi,  \delta_{\cV_N}  \xi \rangle | \leq&\;  C (\log N)^{1/2} N^{1/2 - \alpha} \|\cH_N^{1/2} \xi\|^2 \end{split}\end{equation}
for any $\alpha > 0$, for all $\xi \in \cF^{\leq N}_+$, and $N \in \bN$ large enough.  
\end{prop}
\begin{proof}
We proceed as in \cite[Prop. 8.1]{BBCS3}, computing $[ a_{p+u}^* a_{q}^* a_{p}a_{q+u}, b^*_{r+v} a^*_{-r}a_{v}]$. We obtain 
\[ [\cV_N, A] = \frac{1}{N^{1/2}}\sum^*_{u\in\Lambda^*, r, v \in \L^*_+} 
	\widehat{V}((u-r)/e^N)\eta_r b^*_{u+v} a_{-u}^*a_v + \Theta_1 + \Theta_2 + \Theta_3 + \hc \]
	with 
	\begin{equation}\label{eq:Theta1to4}
	\begin{split}
	\Theta_1 &:=\;\frac{1}{\sqrt N}\sum^*_{\substack{u\in \Lambda^*\\ r,p,v \in \L^*_+ }}\widehat{V} (u/e^N) \eta_r b_{p+u}^* a_{r+v-u}^*a_{-r}^*a_{p}a_{v}\,,   \\
	\Theta_2 &:=\;\frac{1}{\sqrt N}\sum^*_{\substack{u\in \Lambda^*\\ p,r,v\in \L^*_+}}\widehat{V} (u/e^N) \eta_r b_{r+v}^* a_{p+u}^*a_{-r-u}^*a_{p}a_{v}  \,, \\
	\Theta_3 &:=\;-\frac{1}{\sqrt N}\sum^*_{\substack{u\in \Lambda^*,p,r,v\in \Lambda_+^* }}\widehat{V} (u/e^N) \eta_r  b_{r+v}^* a_{-r}^*a_{p+u}^*a_{p}a_{v+u}  \,.
	\end{split}
	\end{equation}     
and with $\sum^*$ running over all momenta, except choices for which the argument of a creation or annihilation operator vanishes. We conclude that $\delta_{\cV_N} = \Theta_1 + \Theta_2 + \Theta_3 + \hc$. Next, we show that each error term $\Theta_j$, with $j=1,2,3$, satisfies \eqref{eq:RN-cVN1}. To bound $\Theta_1$ we switch to position space and apply Cauchy-Schwarz. We find 
	\begin{equation*}\label{eq:Theta1f} \begin{split}
	|\langle \xi, \Theta_1\xi \rangle |\leq &\; \frac{1}{\sqrt{N}} \int_{\Lambda^2}dxdy\; e^{2N}V(e^N(x-y))  \|\check{a}(\check{ \eta}_y)\check{a}_y \check{a}_x \xi \|\|\check{a}_y \check{a}_x \xi \|\\
	\leq&\; C \|\eta\| \int_{\L^2}dx dy \,e^{2N}V(e^N(x-y))\|\check{a}_y \check{a}_x \xi \|^2\\
	\leq& CN^{-\a}\|\cV_N^{1/2}\xi\|^2\,,
	\end{split}\end{equation*}
for any $\xi\in \cF_+^{\leq N}$
The term 	$\Theta_3$ can be controlled similarly. We find  
	\[\begin{split} \label{eq:Theta3f}
	|\langle \xi, \Theta_3\xi \rangle |=&\; \bigg| \frac{1}{\sqrt N } \int_{\Lambda^2}dxdy \;e^{2N}V(e^N(x-y))\langle \xi, \check{b}_x^*\check{a}^*(\check{ \eta}_x)\check{a}^*_y \check{a}_x \check{a}_y \xi \rangle\bigg|  \\
&\leq  C N^{-\a}\| \cV_N^{1/2}\xi \|^2\,.
	\end{split}\]
	It remains to  bound the term $\Theta_2$ on the r.h.s. of (\ref{eq:Theta1to4}). Passing to position space we obtain, by Cauchy-Schwarz, 
	\[\begin{split} \label{eq:Theta2}
	|\langle \xi,& \Theta_2\xi \rangle |=\; \bigg| \frac{1}{ \sqrt N } \int_{\Lambda^3}dxdydz \;e^{2N}V(e^N(y-z)) \check{ \eta}(x-z)\langle \xi, \check{b}_{x}^*\check{a}_y^*\check{a}^*_z \check{a}_x\check{a}_y\xi \rangle \bigg|  \\
	\leq &\; CN^{-1/2} \int_{\Lambda^3}dxdydz\; e^{2N} V(e^N(y-z)) |\check{\eta}(x-z)|\| \check{a}_x \check{a}_y \check{a}_z  \xi \| \| \check{a}_x \check{a}_y \xi \|\\
	\leq &\; CN^{-1/2}\|\cV_N^{1/2}\cN_+^{1/2}\xi\|\left[\int_{\L^3}dxdydz \, e^{2N}V(e^N(y-z))|\check{\eta}(x-z)|^2\|\check{a}_x \check{a}_y \xi \|^2\right]^{1/2}\,,
	\end{split}\]
To bound the term in the square bracket, we write it in first quantized form and, for any $2 < q < \infty$, 
we apply H\"older inequality and the Sobolev inequality $\| u \|_{q} \leq C \sqrt{q} \, \| u \|_{H^1}$ to estimate (denoting by $1< q' < 2$ the dual index to $q$),  
\begin{equation}\label{eq:first} \begin{split} 
\sum_{n=2}^N &\sum_{i<j}^n \int \left[ e^{2N} V (e^N \cdot ) * |\check{\eta}|^2 \right] (x_i - x_j) \, | \xi^{(n)} (x_1, \dots , x_n)|^2 dx_1 \dots dx_n  \\ 
&  \leq C q \| e^{2N} V (e^N \cdot ) * |\check{\eta}|^2 \|_{q'}  \\ 
& \hspace{1cm} \times \sum_{n=2}^N n \sum_{i=1}^n \int \left[ |\nabla_{x_i} \xi^{(n)} (x_1, \dots , x_n) |^2 + |\xi^{(n)} (x_1, \dots , x_n)|^2 \right] dx_1 \dots dx_n \\
&\leq C q \| \check{\eta} \|_{2q'}^2 \| (\cK + \cN_+)^{1/2} \cN_+^{1/2} \xi \|^2. 
 \end{split} \end{equation} 
With the bounds (\ref{eq:etax}), (\ref{eq:etaHL2}), 
 \[ 
\| \check{\eta} \|_{2q'}^{2} \leq \| \check{\eta} \|_2^{2/q'} \| \check{\eta} \|_\infty^{2(q'-1)/q'} \leq N^{-2\alpha/q'}  N^{2(q'-1)/q'} \] 
we conclude that 
\[ \begin{split} |\langle \xi , \Theta_2 \xi \rangle | &\leq C q^{1/2} N^{-1/2}  N^{-\alpha/q'}  N^{1/q}  \| \cV_N^{1/2} \cN_+^{1/2} \xi \| \|(\cK + \cN_+)^{1/2} \cN_+^{1/2} \xi \| \\ &\leq C q^{1/2} 
N^{1/2}  N^{-\alpha/q'}  N^{1/q} \| \cV_N^{1/2} \xi \| \|\cK^{1/2} \xi \|
 \end{split} \] 
for any $2 < q < \infty$, if $1/q + 1/q' = 1$. Choosing $q = \log N$, we obtain that 
\[  |\langle \xi , \Theta_2 \xi \rangle |  \leq C (\log N)^{1/2} N^{1/2 -\alpha} \| \cH_N^{1/2} \xi \|^2\, .  \]
 \end{proof}

Using Prop. \ref{prop:commAcVN}, we can now show Proposition \ref{prop:AHNgrow}.
\begin{proof}[Proof of Prop. \ref{prop:AHNgrow}]
The proof follows a strategy similar to \cite[Lemma 8.2]{BBCS3}. For fixed $\xi \in \cF_+^{\leq N}$ and $s\in [0; 1]$, we define
\begin{equation*}
f_\xi (s) := \langle \xi, e^{-sA} \cH_N e^{sA} \xi\rangle\,.  \end{equation*}
We compute  
\begin{equation}\label{eq:f'1} f'_\xi  (s) = \langle \xi, e^{-sA} [\cK, A] e^{sA} \xi\rangle + \langle \xi, e^{-sA} [\cV_N, A]  e^{sA} \xi\rangle \,. \end{equation}
With Prop. \ref{prop:commAcVN}, we have  
\[  [\cV_N, A]  = \frac{1}{\sqrt N} \sum_{\substack{ u, v\in \L^*_+, u\neq-v}} (\widehat{V}(\cdot/e^N)*\eta)(u) \left[ b^*_{u+v} a^*_{-u} a_v + \hc \right] +  \delta_{\cV_N} \]
with $\delta_{\cV_N}$ satisfying (\ref{eq:RN-cVN1}). Switching to position space and using Prop. \ref{prop:ANgrow} we find , using (\ref{eq:etax}) to bound $\| \check{\eta} \|_\infty \leq C N$,  
\begin{equation}\label{eq:estVN1}
	\begin{split}
	&\bigg |  \frac{1}{\sqrt N}\sum_{\substack{u,v\in\Lambda_+^*}}(\widehat{V}(\cdot/e^N)*\eta)(u) \langle \xi,e^{-sA}b^*_{u+v}a^*_{-u} a_v e^{sA}\xi \rangle\bigg| \\
	&\hspace{0cm}= \bigg | \fra{1}{\sqrt N}\int_{\Lambda^2} dx dy\; e^{2N} V(e^N(x-y))\check\eta(x-y) \langle \xi, e^{-sA} \check{a}^*_{x} \check{a}^*_{y} \check{a}_{y} e^{sA}\xi \rangle\bigg|\\
	&\leq N^{1/2}  \bigg[ \int_{\Lambda^2}dxdy\; e^{2N}V(e^N(x-y))\|\check{a}_x\check{a}_ye^{sA}\xi\|^2\bigg]^{1/2}\\	
	&\hspace{1.5cm}\times\bigg[ \int_{\Lambda^2}dxdy\; e^{2N}V(e^N(x-y))\|\check{a}_ye^{sA}\xi\|^2\bigg]^{1/2}\\
	&\leq CN^{1/2} \| \cV_N^{1/2} e^{sA} \xi\| \| \cN_+^{1/2} e^{sA} \xi\|
	\end{split}
\end{equation}	
Together with (\ref{eq:RN-cVN1}) we conclude that for any $\a > 1/2$ 
	\be  \label{eq:cVN-commA}
\Big| \langle \xi, e^{-sA} [\cV_N, A]  e^{sA} \xi \rangle \Big| \leq C \langle \xi, e^{-sA} \cH_N e^{sA} \xi \rangle + CN\langle \xi, e^{-sA} (\cN_++1) e^{sA} \xi \rangle 
\ee
	if $N$ is large enough. Next, we analyze the first term on the r.h.s. of (\ref{eq:f'1}). We compute 
	\begin{equation}\label{eq:commAcK}
	\begin{split}
	[\cK,A] &= \frac{1}{\sqrt N}\sum_{r, v\in \L^*_+ } 2r^2\eta_r \big[b^*_{r+v}a^*_{-r} a_v +\text{h.c.}\big]\\
	&\hspace{3cm} +  \frac{2}{\sqrt N}\sum_{r, v\in \L^*_+ }  r\cdot v \;\eta_r\big[b^*_{r+v}a^*_{-r} a_v +\text{h.c.}\big] \\ &=: \text{T}_1 + \text{T}_2 \,.
	\end{split}
	\end{equation} 
	With (\ref{eq:eta-scat}), we write 
	\begin{equation}\label{eq:commAcKterm1}\begin{split}
	\text{T}_1 = \; &-\sqrt N\sum_{\substack{r,v\in\Lambda_+^*\\ r\neq-v} } (\widehat V(\cdot/e^N)\ast \widehat f_{N,\ell})(r) \big[b^*_{r+v}a^*_{-r} a_v +\text{h.c.}\big]\\
	&+ 2\sqrt N\sum_{r, v\in \L^*_+ } e^{2N} \lambda_\ell (\widehat{\chi}_\ell * \widehat{f}_{N,\ell})(r)  \big[b^*_{r+v}a^*_{-r} a_v +\text{h.c.}\big]\\
	=: &\; \text{T}_{11} + \text{T}_{12}\,.
	\end{split}\end{equation}	
The contribution of $\text{T}_{11}$ can be estimated similarly as in (\ref{eq:estVN1}); switching to position space and using \eqref{eq:intpotf}, we obtain 
\begin{equation}\label{eq:T11} \begin{split}
		\big| \langle \xi_1, \text{T}_{11} \, \xi_2 \rangle \big| &\leq 
		C\sqrt{N}\int dxdye^{2N}V(e^N(x-y))f_{\ell}(e^N(x-y))\|\check{a}_x\check{a}_y \xi\|\|  a_y \xi\|\\
		&\leq C\sqrt{N}\, \Big[\int dxdye^{2N}V(e^N(x-y)) \|\check{a}_x\check{a}_y \xi\|^2 \Big]^{1/2}\\& \hskip 1.5cm \times\Big[ \int dxdye^{2N}V(e^N(x-y))f_{\ell}(e^N(x-y))\| a_y \xi\|^2 \Big]^{1/2}\\	
		&\leq C  \|\cV_N^{1/2}  \xi \| \| \cN_+^{1/2} \xi\|\,. \end{split}\end{equation}
for any $\x \in \cF^{\leq N}_+$.
The second term in (\ref{eq:commAcKterm1}) can be controlled using that for any $\x \in \cF^{\leq N}_+$ and $2 \leq q < \io$ we have	
	\begin{equation}\label{eq:estimatechi2}\begin{split}
& N^{2\a} \int_{\L^2} dx dy \, \chi (|x-y|\leq N^{-\a}) \| \check{a}_x \check{a}_y \xi\| \| \check{a}_x \xi \| \\
	& \leq N^{2\alpha} \int_{\L^2} dx \|\check{a}_x\xi\|\left( \int dy \, \chi (|x-y|\leq N^{-\a})\right)^{1-1/q}\left(\int dy \|\check{a}_x\check{a}_y\xi\|^q\right)^{1/q}\\
	&\leq C N^{2\a/q}q^{1/2}\left[\int dx\|\check{a}_x\xi\|^2 \right]^{1/2}\left[\int dxdy \|\check{a}_x\nabla_y\check{a}_y\xi\|^2+\int dxdy\|\check{a}_x\check{a}_y\xi\|^2\right]^{1/2}\\
	&\leq C  N^{2\a/q}q^{1/2}\|(\cN_++1)^{1/2}\xi\|\left[\|\cK^{1/2}(\cN_++1)^{1/2}\xi\|+\|(\cN_++1)\xi\|\right].
	\end{split}\end{equation}
Hence, choosing $q=\log N$, 
	\begin{equation}\label{eq:T12}\begin{split}
	\big |  \langle \xi,  & \text{T}_{12} \xi \rangle\big| \\
&= \Big|\sqrt N e^{2N} \l_\ell \int_{\L^2} dx dy \, \chi (|x-y|\leq N^{-\a})f_{N,\ell}(x-y) \bmedia{\xi, \check{b}^*_x\check{a}^*_y\check{a}_x \xi}\Big|\\	
	& \leq  C N^{2\a -1/2} \int_{\L^2} dx dy \, \chi (|x-y|\leq N^{-\a}) \| \check{a}_x \check{a}_y \xi\| \| \check{a}_x \xi \| \\
	&\leq C (\log N)^{1/2} \|(\cN_++1)^{1/2}\xi\|\left[\|\cK^{1/2}\xi\|+\|(\cN_++1)^{1/2}\xi\|\right] \,,
	\end{split}\end{equation}
With \eqref{eq:T11} and \eqref{eq:T12} we conclude that 
	\begin{equation}\label{eq:T1-bd} |\langle \xi , e^{-A} \text{T}_1 e^A \xi \rangle | \leq  C (\log N)^{1/2}\| (\cH_N + 1)^{1/2} e^{sA}\xi\|  \| (\cN_++1)^{1/2} e^{sA} \xi \|\,. \end{equation}
	for all $\xi \in \cF_+^{\leq N}$. 	As for the second term on the r.h.s. of (\ref{eq:commAcK}) we have
	\begin{equation}\label{eq:boundT2}
	\begin{split}
	&\big| \langle \xi, \text{T}_2 \xi \rangle\big| \\
	&\hspace{.5cm} \leq \frac{C}{\sqrt N}\bigg [\sum_{r\in \L^*_+ }   |r|^2 \| \cN_+^{1/2} a_{-r} \xi \|^2 \bigg]^{1/2}\bigg [ \sum_{r, v\in \L^*_+ } |v|^{2}\eta_r^2\| a_{v} \x \|^2 \bigg]^{1/2}\\	
	&\hspace{.5cm} \leq CN^{-\a} \| \cK^{1/2}\xi\|^2.
	\end{split}
	\end{equation}
for any $\xi \in \cF^{\leq N}_+$. 	With (\ref{eq:T1-bd}) and Prop. \ref{prop:ANgrow}, we conclude that
	\begin{equation*} \label{eq:cK-commA} 
	|\langle \xi, e^{-sA} [\cK, A]  e^{sA} \xi\rangle| \leq C \langle \xi, e^{-sA} \cH_N e^{sA}\xi \rangle + C \log N \langle \xi,  e^{-sA} \cN_+ e^{sA}  \xi \rangle\,.\end{equation*}
Combining with  Eq. \eqref{eq:cVN-commA} we obtain
\begin{equation*}  
	|\langle \xi, e^{-sA} [\cH_N, A]  e^{sA} \xi\rangle| \leq C \langle \xi, e^{-sA} \cH_N e^{sA}\xi \rangle + C N \langle \xi,   e^{-sA} \cN_+ e^{sA}\xi \rangle \,.\end{equation*}
With Prop. \ref{prop:ANgrow}  we obtain the differential inequality 
	\[ | f'_\xi (s) | \leq C f_\xi (s) + CN \langle \xi, (\cN_+ + 1) \xi \rangle \,.\]
	By Gronwall's Lemma, we find (\ref{eq:expAHNexpA}). 
\end{proof}

\subsection{Analysis of $e^{-A} \cO_N e^{A}$}\label{sub:RN-cO}

In this section we study the contribution to $\cR_{N,\a}$ arising from the operator $\cO_N$, 
defined in \eqref{eq:GNeff-deco2}. To this end, it is convenient to use the following lemma.
\begin{lemma} \label{lm:cNops} Let $A$ be defined in \eqref{eq:defA}. Then, there exists a constant $C > 0$ such that 
	\[
	\sum_{p\in \Lambda_+^*} F_p\,  e^{-A} a^*_p a_p  e^{A} =  \sum_{p\in \Lambda_+^*} F_p\, a^*_p a_p  + \cE_F  
\]
where
\[ \label{eq:lmF-1} 
	| \bmedia{\xi_1, \cE_F \xi_2 }|\leq  CN^{-\alpha} \|  F\|_\infty \| (\cN_++1)^{1/2}\xi_1\| \| (\cN_++1)^{1/2}\xi_2\|
	\]
	for all $\alpha > 0$, $\xi_1, \xi_2 \in \cF_+^{\leq N}$, $F \in \ell^{\infty} (\Lambda^*_+)$,  and $N \in \bN$ large enough.  
\end{lemma}

\begin{proof} The lemma is analogous to \cite[Lemma 8.6]{BBCS3}. We estimate
	\[\begin{split}
	\Big| \sum_{p\in \Lambda_+^*} F_p   &( \langle \xi_1, e^{-A} a^*_p a_p  e^{A} \xi_2 \rangle - \langle \xi_1 , a^*_p a_p \xi_2 \rangle ) \Big|  \\ &= \Big|  \int_0^1ds\;\sum_{p\in \Lambda_+^*} F_p \langle \xi_1 , e^{-sA}[a^*_p a_p, A] e^{sA} \xi_2 \rangle \Big| \\ &\leq  \frac1{\sqrt N} \int_0^1 ds \sum_{r, v\in  \L^*_+ } | F_{r+v}+F_{-r} - F_v | |\eta_r|  |\langle e^{sA} \xi_1, b^*_{r+v}a^*_{-r}a_v e^{sA} \xi_2 \rangle | \\
	&\leq  C \| \eta \| \|  F \|_\infty \| (\cN_++1)^{1/2}\xi_1\| \| (\cN_++1)^{1/2}\xi_2\|\,.
	\end{split}\]
where we used Prop. \ref{prop:ANgrow}. 
\end{proof}

We consider now the action of $e^A$ on the operator $\cO_N$, as defined in (\ref{eq:GNeff-deco2}). 
\begin{prop}\label{prop:cON} Let $A$ be defined in \eqref{eq:defA}. Then there exists a constant $C>0$ such that
	\begin{equation*}
	e^{-A} \cO_N e^A = \frac 1 2 \widehat{\o}_N(0)(N-1) \left(1-\frac{\cN_+}{N}\right) + \big[2N\widehat V(0)- \frac 1 2 \widehat{\o}_N(0) \big] \cN_+ (1-\cN_+/N) + \delta_{\cO_N} 
	\end{equation*}
	where  		
	\begin{equation*}\label{eq:RN-cON1} 
	\pm  \delta_{\cO_N}  \leq  CN^{1-\alpha}   (\cN_+ +1)
\end{equation*}
for all $\alpha > 0$,  and $N \in \bN$ large enough. 
\end{prop}

\begin{proof} 
The proof is very similar to \cite[Prop. 8.7]{BBCS3}. First of all, with Lemma \ref{lm:cNops} we can bound 	
\[\begin{split}  \label{eq:RN-cON-p1}
	\pm\bigg\{ e^{-A} &\left[ \frac 1 2 \widehat{\o}_N(0) (N-1)\left(1-\frac{\cN_+}{N}\right) + \big[2N\widehat V(0)- \frac 1 {2} \widehat{\o}_N(0)\big]\cN_+ \right] e^{A} \\ &\hspace{1.5cm} - \left[ \frac 1 2 \widehat{\o}_N(0)  (N-1)\left(1-\frac{\cN_+}{N}\right)  + \big[2 N \widehat V(0)- \frac 1 2 \widehat{\o}_N(0) \big]\cN_+ \right] \bigg\}\\
	&\hspace{1.5cm} \leq C N^{1-\alpha}(\cN_++1)\,. \end{split} \]
Moreover, for the contribution quadratic in $\cN_+$, we can decompose 
\[ \begin{split}& \left\langle \xi, \left[ e^{-A} \cN_+^2 e^A - \cN_+^2 \right] \xi \right\rangle \\ &\hspace{2cm} =   \left\langle \xi_1, \left[ e^{-A} \cN_+ e^A - \cN_+ \right] \xi \right\rangle + \left\langle \xi, \left[e^{-A} \cN_+ e^A - \cN_+ \right] \xi_2 \right\rangle \end{split} \]
	with $\xi_1 = e^{-A} \cN_+ e^A \xi$ and $\xi_2 = \cN_+ \xi$, and estimate, again with Lemma \ref{lm:cNops},  
	\[\begin{split}  &\left| \left\langle \xi, \left[ e^{-A} \cN_+^2 e^A - \cN_+^2 \right] \xi \right\rangle \right| \\ &\hspace{2cm} \leq C N^{-\a} \| (\cN_+ + 1)^{1/2} \xi \| \left[ \| (\cN_+ + 1)^{1/2} \xi_1 \| + \| (\cN_+ + 1)^{1/2} \xi_2 \| \right]\,. \end{split} \]
With Prop. \ref{prop:ANgrow}, we have $\| (\cN_+ + 1)^{1/2} \xi_1 \| \leq C \| (\cN_+ +1)^{3/2} \xi \|$. 
\end{proof}

\subsection{Contributions from $e^{-A} \cK e^{A}$}

In Section \ref{sub:Rfin} we will analyse the contributions to $\cR_{N,\a}$ arising from conjugation of the kinetic energy operator $\cK = \sum_{p \in \L_+^*} p^2 a_p^* a_p$.  To this aim we will exploit  properties of the commutator $[\cK,  A]$, collected in the following proposition.

\begin{prop}\label{prop:RN-K}  Let $A$ be defined as in \eqref{eq:defA} and $\widehat{\o}_N(r)$ be defined in \eqref{eq:defomegaN}. Then there exists a constant $C>0$ such that
	\[ \label{eq:cKAesplicit}
	\begin{split}
	[\cK, A] = &\; -\sqrt N\sum_{p,q\in \Lambda_+^*, p \neq -q } (\widehat V(\cdot/e^N)\ast \widehat f_{N,\ell})(p) ( b^*_{p+q}a^*_{-p} a_q+ \hc)\\
	&\;+ \frac 1 {\sqrt N} \sum_{p, q \in \L^*_+, p \neq -q } \widehat{\o}_N(p) \big[ b^*_{p+q}a^*_{-p} a_q + \hc\big] + \delta_\cK
	\end{split}
	\]
where
	\begin{equation} \label{eq:RN-cK1} 
	\big|  \langle \xi , \d_\cK \xi \rangle \big|    \leq  C N^{-1} (\log N)^{1/2}  \|\cK^{1/2}\xi \| \|\cN_+^{1/2}\xi\| + C N^{-\a}  \|\cK^{1/2}\xi \|^2
\end{equation}
	for all $\a >1$, $\xi \in \cF_+^{\leq N}$, and $N \in \bN$ large enough. Moreover, the operator
\[
\D_\cK = \frac 1 {\sqrt N}  \sum_{p, q \in \L^*_+, p \neq -q } \widehat{\o}_N(p)  \big[ b^*_{p+q}a^*_{-p} a_q , A \big]
\]
satisfies 	\begin{equation} \label{eq:RN-cK3} 
	\big|  \langle \xi , \D_\cK \xi \rangle \big|   \leq  C N^{-\a} (\log N)^{1/2}  \|\cK^{1/2}\xi \|^2 + C N^{-1} \| (\cN_+ + 1)^{1/2} \xi\|^2
	\end{equation}
	for all $\a > 1$, $\xi \in \cF_+^{\leq N}$, and $N \in \bN$ large enough. 

\end{prop}

\begin{proof} To show \eqref{eq:RN-cK1}  we recall from Eqs. (\ref{eq:commAcK}),  \eqref{eq:commAcKterm1} that
\[
	\begin{split}
	[\cK,A] =\;&-\sqrt N\sum_{\substack{r,v\in\Lambda_+^*\\ r\neq-v} } (\widehat V(\cdot/e^N)\ast \widehat f_{N,\ell})(r) \big[b^*_{r+v}a^*_{-r} a_v +\text{h.c.}\big]\\
	&+ 2\sqrt N\sum_{r, v\in \L^*_+ } e^{2N} \lambda_\ell (\widehat{\chi}_\ell * \widehat{f}_{N,\ell})(r)  \big[b^*_{r+v}a^*_{-r} a_v +\text{h.c.}\big]\\
	&\hspace{3cm} +  \frac{2}{\sqrt N}\sum_{r, v\in \L^*_+ }  r\cdot v \;\eta_r\big[b^*_{r+v}a^*_{-r} a_v +\text{h.c.}\big] \\ 
=\; &\text{T}_{11} + \text{T}_{12}+  \text{T}_2 \,. 
	\end{split}
	\]
with $\text{T}_2$ satisfying  (\ref{eq:boundT2}). Using the definition $\widehat{\o}_{N}(p)=2 N e^{2N} \l_\ell \widehat \chi_\ell(p)$  we write
	\[\begin{split} \label{eq:T121-T122}
	\text{T}_{12} 	= \; &\frac 1 {\sqrt N}\sum_{p, q \in \L^*_+, p \neq -q } \widehat{\o}_N(p) \big[ b^*_{p+q}a^*_{-p} a_q + \hc\big] \\
	&+\frac{2}{\sqrt N} \,e^{2N} \l_\ell \sum_{p, q \in \L^*_+, p \neq -q } (\widehat{\chi_\ell}*\eta)(p) \big[ b^*_{p+q}a^*_{-p} a_q + \hc\big] \\
	=\; & \text{T}_{121}+\text{T}_{122}.
	\end{split}\]
	Hence, $\d_K = T_2 + T_{122}$. To bound $T_{122}$ we switch to position space:  
\[ \begin{split} 
&| \langle \xi, \text{T}_{122} \xi \rangle | \\
&\leq CN^{2\a-3/2} \int_{\L^2} \chi_\ell (x-y) \check{\eta} (x-y) \|\check{a}_x\check{a}_y\xi\|\|\check{a}_x\xi\| \\
&\leq CN^{2\a-3/2} \left[ \int_{\L^2} \chi_\ell(x-y)  \|\check{a}_x\check{a}_y\xi\|^2 dx dy \right]^{1/2} \left[ \int_{\L^2} |\check{\eta} (x-y)|^2  \|\check{a}_x \xi\|^2 dx dy \right]^{1/2} \\
&\leq CN^{\a-3/2} \| \cN^{1/2}_+ \xi \| \left[ \int_{\L^2} \chi_\ell(x-y)  \|\check{a}_x\check{a}_y\xi\|^2 dx dy \right]^{1/2}.
\end{split} \]	
To bound the term in the parenthesis, we proceed similarly as in (\ref{eq:first}). We find
\[ \int_{\L^2} \chi_\ell(x-y)  \|\check{a}_x\check{a}_y\xi\|^2 dx dy \leq C q \| \chi_\ell \|_{q'} \| \cK^{1/2} \cN_+^{1/2} \xi \|^2 \leq C q N^{1-2\alpha /q'}  \| \cK^{1/2} \xi \|^2 \]
for any $q > 2$ and $1 < q' < 2$ with $1/q+ 1/q' =1$. Choosing $q = \log N$, we obtain 
\[ |\langle \xi, \text{T}_{122} \xi \rangle | \leq C N^{-1} (\log N)^{1/2}  \|\cN^{1/2}_+ \xi \| \| \cK^{1/2} \xi \| \]
With \eqref{eq:boundT2}, this implies \eqref{eq:RN-cK1}.   \\

Let us now focus on (\ref{eq:RN-cK3}). We have 
	\begin{equation*}\label{eq:commCrenA}
	\begin{split}
	\frac 1 {\sqrt N} & \sum_{p, q \in \L^*_+, p \neq -q } \widehat{\o}_N(p)  \big[ b^*_{p+q}a^*_{-p} a_q , A \big]\\
	&= \frac 1 N  \sum_{\substack{r, p,q, v\in \L^*_+, \\p \neq-q,  r\neq-v}} \widehat{\o}_N(p) \eta_r \big[b^*_{p+q}a^*_{-p} a_q , b^*_{r+v}a^*_{-r} a_v - a_v^* a_{-r}b_{r+v} \big]\,.
	\end{split}
	\end{equation*}
With the commutators from the proof of Prop. 8.8 in \cite{BBCS3}, we arrive at	

\[ \frac 1 {\sqrt N} \sum_{p, q \in \L^*_+, p \neq -q } \widehat{\o}_N(p)  \big[ b^*_{p+q}a^*_{-p} a_q , A \big]+ \text{h.c.} = \sum_{j=1}^{12} \Upsilon_j +\hc \]
	where 
	\begin{equation}\label{eq:commCNrenA2}
	\begin{split}
	\Upsilon_1:= &\;  - \frac 1 N  \sum_{\substack{q,r,v \in \L^*_+,\\ q\neq v, r \neq -v }} \big( \widehat{\o}_N(v-q) + \widehat{\o}_N(v)  \big) \eta_r b^*_{r+v}b^*_{-r} a^*_{q-v} a_q \,,\\
	\Upsilon_2:= &\; \frac 1 N \sum_{\substack{q,r,v \in \L^*_+,\\ r\neq -v, r \neq -q} } \widehat{\o}_N(r+q)\eta_r (1-\cN_+/N)a^*_{v}a^*_{r+q} a_q a_{r+v}\,,\\
	\Upsilon_3:= &\; \frac 1 N \sum_{\substack{r,v \in \L^*_+\,,\\  r \neq -v }} \big( \widehat{\o}_N(r+v)+ \widehat{\o}_N(r) \big)\eta_r(1-\cN_+/N)a^*_{v} a_v,\\
	\Upsilon_4:= &\;\frac 1 N \sum_{\substack{q,r,v \in \L^*_+\,,\\ q\neq v, r \neq -v} } \widehat{\o}_N(r+v-q)\eta_r(1-\cN_+/N)a_v^* a^*_{q-r-v} a_{-r}  a_q,\\
	\Upsilon_5:= &\; - \frac 1  {N^{2}} \sum_{\substack{p,q,r,v \in \L^*_+,\\ p\neq -q, r \neq -v }} \widehat{\o}_N(p) \eta_ra_v^* a^*_{p+q}a^*_{-p}a_{-r} a_{r+v} a_q\,,\\
	\Upsilon_6:= &\; - \frac 1 {N^2}\sum_{\substack{q,r,v \in \L^*_+,\\ q\neq r+v}} \widehat{\o}_N(r+v)\eta_r  a_v^*a^*_{q-r-v}a_{-r}   a_q,  \\
	\Upsilon_7:= &\; - \frac 1 {N^2} \sum_{\substack{q,r,v \in \L^*_+,\\ q\neq -r, r \neq -v}} \widehat{\o}_N(r)\eta_r  a_v^*a^*_{q+r}a_{r+v}a_q\,,\\
	\Upsilon_8:= &\; \frac 1 N \sum_{\substack{r,v,p \in \L^*_+,\\ p \neq -r-v} } \widehat{\o}_N(p)\eta_r b^*_{p+r+v}b^*_{-p} a^*_{-r}a_{v} \,,\\
	\Upsilon_{9}:= &\; \frac 1 N \sum_{\substack{p,r,v \in \L^*_+,\\p\neq r, r\neq -v} } \widehat{\o}_N(p)\eta_r b^*_{p-r}b^*_{r+v} a^*_{-p}a_{v} \,,\\
\end{split}\ee
and
\be \begin{split} \label{eq:commCNrenA3}
	\Upsilon_{10}:= &\; \frac 1 N \sum_{\substack{q,r,v \in \L^*_+,\\ q\neq -r, r\neq -v} } \widehat{\o}_N(r)\eta_r b^*_{q+r}a^*_{v} a_{q}b_{r+v} \,,\\
	\Upsilon_{11}:= &\; -\frac 1 N \sum_{\substack{p,r,v \in \L^*_+,\\p\neq -v,r\neq -v} } \widehat{\o}_N(p)\eta_r b^*_{p+v}a^*_{-p} a_{-r}b_{r+v} \,,\\
	\Upsilon_{12}:= &\; \frac 1 N \sum_{\substack{q, r,v \in \L^*_+\\ r \neq q-v, -v} } \widehat{\o}_N(r+v)\eta_r b^*_{q-r-v}a^*_{v} a_{-r}b_{q} \,.\\
	\end{split}
	\end{equation}	
	
	To conclude the proof of Prop. \ref{prop:RN-K}, we show that all operators in \eqref{eq:commCNrenA2} and \eqref{eq:commCNrenA3} 
satisfy (\ref{eq:RN-cK3}). To study all these terms it is convenient to switch to position space. We recall that  $\widehat \o_N(p)= g_N \widehat \chi(\ell p)$ with $|g_N|\leq C$ and $\ell= N^{-\a}$. Using \eqref{eq:estimatechi2} we find:
	\[\begin{split}
	\big| \langle \xi, \Upsilon_1 \xi \rangle \big| & \leq  C N^{2\a-1} \int_{\L^2}dx dy \, \chi_\ell(x-y)\|\check{b}(\check{ \eta}_x)\check{b}_x\check{a}_y\xi\|\left[ \| \check{a}_x \xi \| +  \|\check{a}_y\xi\| \right] \\
	&\leq CN^{2\a-1}\|\eta\|\int_{\L^2}dx dy \, \chi_\ell(x-y)\|\check{b}_x\check{a}_y(\cN_++1)^{1/2} \xi\| \|\check{a}_x\xi\| \\
& \leq C N^{-\a} (\log N)^{1/2} \| (\cN_++1)^{1/2}\xi \| \| \cK^{1/2}\xi \| \,.
	\end{split}\]
	The expectation of $\Upsilon_2$ is bounded following the same strategy used to show \eqref{eq:estimatechi2}. For any $2\leq q < \io$ we have
	\[\begin{split}
	&\big| \langle \xi, \Upsilon_2 \xi \rangle \big| \\
	&\leq  CN^{2\a-1}\int_{\L^3}dxdydz \chi_\ell(z-y)|\check{\eta}(z-x)| \|\check{a}_x\check{a}_y\xi\|\|\check{a}_z\check{a}_x\xi\|\\
	&\leq CN^{2\a-1}\int_{\L^2}dxdz |\check{\eta}(z-x)|\|\check{a}_z\check{a}_x\xi\|\\
	&\hspace{2.5cm}\times\left(\int_{\L}dy \, \chi(|z-y|\leq N^{-\a})\right)^{1-1/q} \left(\int_{\L}dy\|\check{a}_x\check{a}_y\xi\|^q\right)^{1/q}\\
&\leq Cq^{1/2}N^{2\a/q-1}\|\eta\|\|(\cN_++1)\xi\|\left[\int_{\L^2}dxdy\|\check{a}_x\nabla_y\check{a}_y\xi\|^2+\int_{\L^2}dxdy\|\check{a}_x\check{a}_y\xi\|^2\right]^{1/2}\\
	&\leq CN^{-\a}(\log N)^{1/2}\|(\cN_++1)^{1/2}\xi\| \|\cK^{1/2}\xi\|\,,
	\end{split}\]
where in the last line we chose $q=\log N$.  The term $\Upsilon_3$ is of lower order; using that $\big| \sum_r \widehat{\o}_N (r) \eta_r \big| \leq \| \widehat{\chi} (./N^\alpha) \|_2 \| \eta \|_2 \leq C$ and Cauchy-Schwarz, we easily obtain 
\[ \big| \langle \xi , \Upsilon_{3} \xi \rangle \big|  \leq CN^{-1}\|(\cN_++1)^{1/2}\xi\|^2\,. \]
The term $\Upsilon_4$ can be estimated as $\Upsilon_1$ using \eqref{eq:estimatechi2}:
	\[\begin{split}
	\big| \langle \xi , \Upsilon_4 \xi \rangle \big| &\leq CN^{2\a-1}\int_{\L^2}dxdy \, \chi_\ell(x-y)\|\check{a}_x\check{a}_y\xi\|\|\check{a}(\check{\eta}_y)\check{a}_y\xi\|\\
	&\leq CN^{2\a-1}\|\eta\|\int_{\L^2}dxdy  \,\chi_\ell(x-y)\|\check{a}_x\check{a}_y\xi\|\|\check{a}_y (\cN_++1)^{1/2}\xi\| \\
	&\leq CN^{-\a} (\log N)^{1/2}\|(\cN_++1)^{1/2} \xi \|  \|\cK^{1/2}\xi\|\,.
	\end{split}\]
The term $\Upsilon_5$ is bounded similarly to $\Upsilon_2$; with $q=\log N$ we have	
	\[\begin{split}
	\big|\langle \xi , \Upsilon_5 \xi \rangle \big| & \leq CN^{2\a-2}\|\eta\|\int_{\L^3} dxdydz \, \chi_\ell(y-z)\|\check{a}_x\check{a}_y\check{a}_z\xi\|\|\cN_+^{1/2}\check{a}_x\check{a}_y\xi\|\\
	& \leq CN^{2\a-3/2}\| \eta\| \int_{\L^2} dxdy \, \|\check{a}_x\check{a}_y\xi\|\\
& \hskip 2cm \times \left(\int_{\L} dz \, \chi(|y-z|\leq N^{-\a})\right)^{1-1/q}\left(\int_{\L} dz \, \|\check{a}_x\check{a}_y\check{a}_z\xi\|^q\right)^{1/q}\\
	&\leq CN^{-\a} (\log N)^{1/2}\|(\cN_++1)^{1/2}\xi \| \|\cK^{1/2}\xi\|\,.
	\end{split}\]
	The terms $\Upsilon_6$ and $\Upsilon_7$ are of smaller order and can be bounded with Cauchy-Schwarz; we have 
	\[\begin{split}
	\big|\langle \xi , &\Upsilon_6 \xi \rangle \big|  \leq CN^{2\a-2} \int_{\L^2} dxdydz \,\chi_\ell(x-y)\|\check{a}_x\check{a}_y\xi\|\|\check{a}(\check{\eta}_x)\check{a}_y\xi\|\\
	& \leq CN^{\a-3/2}\left(\int_{\L^2} dxdy \, \|\check{a}_x\check{a}_y\xi\|^2\right)^{1/2}\left(\int_{\L^2} dxdy \, \chi (|x-y|\leq N^{-\a}) \|\check{a}_y\xi\|^2\right)^{1/2}\\
	&\leq CN^{-1}\|(\cN_++1)^{1/2}\xi \|^2\,,
	\end{split}\]
and	
	\[\begin{split}
	\big|\langle \xi , \Upsilon_7 \xi \rangle \big|
	& \leq CN^{2\a-2}\int_{\L^3} dxdydz \, \chi_\ell(y-z)|\check{\eta}(z-x)|\|\check{a}_x\check{a}_y\xi\|^2\\
	& \leq CN^{2\a-2}\left(\int_{\L^3} dxdydz \, \chi_\ell(y-z) \|\check{a}_x\check{a}_y\xi\|^2\right)^{1/2} \\
	& \hskip 2cm \times\left(\int_{\L^3} dxdydz \, |\check{\eta}(z-x)|^2 \|\check{a}_x\check{a}_y\xi\|^2\right)^{1/2}\\
	&\leq CN^{-1}\|(\cN_++1)^{1/2}\xi \|^2\,.
	\end{split}\]
	The terms $\Upsilon_8, \Upsilon_{11}, \Upsilon_{12}$ are again bounded, as $\Upsilon_1$, using \eqref{eq:estimatechi2}. We find
	\[\begin{split}
	\big|\langle \xi , \big(\Upsilon_8 +  \Upsilon_{11} + \Upsilon_{12} \big) \xi \rangle \big| & \leq CN^{2\a-1}\|\eta\|\int_{\L^2} dxdy \, \chi_\ell(x-y)\|\cN_+^{1/2}\check{a}_x\check{a}_y\xi\|\|\check{a}_x\xi\|\\
	&\leq CN^{-\a} (\log N)^{1/2}\|(\cN_++1)^{1/2}\xi \| \|\cK^{1/2}\xi\|\,.
	\end{split}\] 
It remains to bound $\Upsilon_9$ and $\Upsilon_{10}$. The term $\Upsilon_9$ is bounded analogously to $\Upsilon_2$: 
	\[\begin{split}
	&\big|\langle \xi , \Upsilon_{9} \xi \rangle \big| \\
	&\quad \leq CN^{2\a-1}\int_{\L^3} dxdydz \, \chi_\ell(x-z)|\check{\eta}(x-y)|\|\check{a}_x\check{a}_y\check{a}_z\xi\|\|\check{a}_y\xi\|\\
	& \quad\leq CN^{2\a-1}\int_{\L^2} dxdy \, |\check{\eta}(x-y)|\|\check{a}_y\xi\|\left(\int_{\L} dz \, \chi(|y-z|\leq N^{-\a})\right)^{1-1/q}\\
	&\hspace{8cm}\times\left(\int_{\L} dz \, \|\check{a}_x\check{a}_y\check{a}_z\xi\|^q\right)^{1/q}\\
	&\quad\leq Cq^{1/2}N^{2\a/q-1}\left[\int_{\L^2}dxdy \,  |\check{\eta}(x-y)|^2 \|\check{a}_y\xi\|^2\right]^{1/2}\left[\int_{\L^3}dxdy \, \Big\|\|\check{a}_x\check{a}_y\check{a}_z \xi\|\Big\|_{L^q_z}^2\right]^{1/2}\\
		&\quad\leq CN^{-\a}(\log N)^{1/2}\|(\cN_++1)^{1/2}\xi \| \|\cK^{1/2}\xi\|\,.
	\end{split}\]
As for $\Upsilon_{10}$, we find 
\[ \big|\langle \xi , \Upsilon_{10} \xi \rangle \big|  \leq CN^{2\a-1}\int_{\L^3} dx dy dz \, \chi_\ell(y-z) |\check{\eta}(x-z)| \|\check{a}_x\check{a}_y\xi\|^2\]  
Proceeding as in (\ref{eq:first}), we obtain 
\[ \big|\langle \xi , \Upsilon_{10} \xi \rangle \big| \leq C q N^{2\alpha}  \| \chi_\ell * |\check{\eta}| \|_{q'} \| \cK^{1/2} \xi \|^2 \leq C q  \| \check{\eta} \|_{q'} \| \cK^{1/2} \xi \|^2 \]
for any $q > 2$, and $q' < 2$ with $1/q + 1/q' = 1$. Since, for an arbitrary $q' < 2$, $\| \check{\eta} \|_{q'} \leq \| \check{\eta} \|_2 = \| \eta \|_2 \leq N^{-\alpha}$, we obtain 
\[ \big| \langle \xi , \Upsilon_{10} \xi \rangle \big| \leq C N^{-\alpha} \| \cK^{1/2} \xi \|^2 \]
We conclude that for any $\a >1$ 
\[
\big|\langle \xi , \sum_{j=1}^{12}\Upsilon_{i} \xi \rangle \big| \leq  C N^{-\a} (\log N )^{1/2} \, \| (\cK+1)^{1/2}\xi\|^2 + C N^{-1}  \| (\cN_+ + 1)^{1/2}  \xi \|^2 \,. 
\]
\end{proof}

\subsection{Analysis of $e^{-A} \cZ_N e^{A}$}
In this subsection, we consider contributions to $\cR_{N,\a}$ arising from conjugation of $\cZ_{N}$, as defined in  \eqref{eq:GNeff-deco2}. 

\begin{prop}\label{prop:RN-cZ} Let $A$ be defined in \eqref{eq:defA}. Then, there exists a constant $C>0$ such that		
	\begin{equation*}
	e^{A} \cZ_N e^{-A} =\frac 1 2  \sum_{\substack{p \in \L^*_+}}  \widehat{\o}_N(p)\;  \big( b^*_p b^*_{-p} + b_p b_{-p} \big) +\delta_{\cZ_N}
	\end{equation*}
	where 
	\begin{equation*}\label{eq:RN-cZN1}  \pm \delta_{\cZ_{N}} \leq C N^{1-\alpha}(\cH_N+1)
	\end{equation*}
for all $\alpha > 0$, and $N \in \bN$ large enough.
\end{prop}
\begin{proof} We have 
	\begin{equation}
	\label{eq:RNM_n2}
	\begin{split}
	\frac 12 \sum_{\substack{p \in \L^*_+}}  \widehat{\o}_N (p) &\;  \big[e^{-A}\big( b^*_p b^*_{-p} + b_p b_{-p} \big)e^{A}-\big( b^*_p b^*_{-p} + b_p b_{-p} \big) \big]\\
	&= \frac 12 \int_0^1ds\sum_{\substack{p \in \L^*_+ }}  \widehat{\o}_N (p) \;  e^{-sA}\big[ b^*_p b^*_{-p} + b_p b_{-p},A \big]e^{sA}.
\end{split}
\end{equation}
We compute  
\begin{equation} \label{eq:bb-commA}
\begin{split}
\frac 12 \sum_{\substack{p \in \L^*_+}}  &\widehat \o_N(p)  \big[ b^*_p b^*_{-p} , b^*_{r+v} a^*_{-r}a_v - a^*_v a_{-r} b_{r+v} \big] \\
 =\; &  -  \widehat \o_N(v)  b^*_{r+v}b^*_{-v}b^*_{-r} + \widehat \o_N(r)  b^*_{v} \Big( b^*_{r}b_{r+v}-\frac 2 N a^*_r a_{r+v}\Big)\\
&+\widehat \o_N(r+v)  \Bigg(1-\fra{\cN_+}{N} \Bigg)b^*_{-r-v}a^*_va_{-r}  -\fra 1 N \sum_{\substack{p \in \L^*}}  \widehat \o_N(p)   b^*_{p}a^*_{-p}a^*_{v}a_{-r}a_{r+v}.
\end{split}
\end{equation}  
With \eqref{eq:bb-commA} we write
\[
\frac 1 2 \sum_{\substack{p\in \L^*}}  \widehat{\o}_N (p)  \big[ b^*_p b^*_{-p} + b_p b_{-p} ,A \big]= \sum_{j=1}^4 \Pi_j + \hc
\]
with 
\begin{equation*} \label{eq:Pi-terms}
\begin{split}
\Pi_1 =\; & - \frac 1 {\sqrt N}\sum_{\substack{ r,v \in \L^*_+\\ r\neq -v}} \widehat{\o}_N (v)\eta_{r}  b^*_{r+v}b^*_{-v}b^*_{-r}\,,\\
\Pi_2 =\; & \frac 1 {\sqrt N} \sum_{\substack{r,v \in \L^*_+: \\r \neq -v }} \widehat{\o}_N(r)\eta_{r} b^*_{v}\Big( b^*_{r}b_{r+v}-\frac 2 N a^*_r a_{r+v}\Big)\,, \\
\Pi_3 =\; & \frac 1 {\sqrt N}\sum_{\substack{r,v \in \L^*_+ \\r \neq -v}}  \widehat{\o}_N (r+v) \eta_{r} \Bigg(1-\fra{\cN_+}{N} \Bigg)b^*_{-r-v}a^*_va_{-r}\,, \\
\Pi_4 =\; &- \frac 1 {N^{3/2}}\sum_{\substack{r,v,p \in \L^*_+: \\r \neq -v}} \widehat{\o}_N (p)\eta_{r} b^*_{p}a^*_{-p}a^*_{v}a_{-r}a_{r+v} \,.
\end{split}
\end{equation*}
To bound the first term, we observe, with (\ref{eq:intV-div-p2}), 
\[ \begin{split} |\langle \xi, \Pi_1 \xi \rangle | &\leq \frac{\| \eta \|}{\sqrt{N}}  \| \cK^{1/2} \cN_+^{1/2} \xi \| \| (\cN_+ + 1)^{1/2} \xi \|  \left[ \sum_{v \in \L^*_+} \frac{|\widehat{\o}_N (v)|^2}{v^2} \right]^{1/2} \\ &\leq C N^{-\alpha} (\log N)^{1/2} \| \cK^{1/2} \xi \| \| (\cN_+ + 1)^{1/2} \xi \|\,. \end{split} \]
The term $\Pi_3$ can be bounded similarly to $\Pi_1$, with \eqref{eq:intV-div-p2}. We find 
\[
\big| \langle \xi, \Pi_3 \xi\rangle \big| \leq C N^{-\alpha} (\log N)^{1/2} \| (\cN_+ + 1)^{1/2} \xi \| \| \cK^{1/2} \xi \|\,. \]
With $|\widehat{\o}_N (r)| \leq C$, we similarly obtain  
\[ \begin{split} |\langle \xi, \Pi_2 \xi \rangle | &\leq  N^{-1/2} \|  \eta \|   \| \cK^{1/2} \cN_+^{1/2} \xi \| \| (\cN_+ + 1)^{1/2} \xi \|  \\ &\leq C N^{-\alpha} \| \cK^{1/2} \xi \| \| (\cN_+ + 1)^{1/2} \xi \|\,. \end{split} \]
Finally, we estimate, using again \eqref{eq:intV-div-p2}, 
\[ \begin{split} \big| \langle \xi, \Pi_4 \xi\rangle \big| &\leq N^{-3/2} \Big( \sum_{r,v,p \in \L_+^*} p^2 |\eta_r|^2 \| a_{-p} a_v (\cN_+ + 1)^{1/2} \xi \|^2 \Big)^{1/2} \\ &\hspace{4cm} \times \Big( \sum_{r,v,p  \in \L_+^*} \frac{|\widehat{\o}_N (p)|^2}{p^2} \| a_{-r} a_{r+v} \xi \|^2 \Big)^{1/2}  \\  &\leq C N^{-3/2} \| \eta \| (\log N)^{1/2} \| \cK^{1/2} (\cN_+ + 1) \xi \| \| (\cN_+ + 1) \xi \| \\ &\leq C N^{-\alpha} (\log N)^{1/2} \| \cK^{1/2} \xi \| \| (\cN_+ + 1)^{1/2} \xi \|\,. \end{split} \]

With (\ref{eq:RNM_n2}), we conclude that 
\[ \begin{split}
&\bigg| \frac 12 \sum_{\substack{p \in \L^*}}  \widehat{\o}_N (p)\;  \big[ \langle \xi , e^{-A}\big( b^*_p b^*_{-p} + b_p b_{-p} \big)e^{A} \xi \rangle - \langle \xi , \big( b^*_p b^*_{-p} + b_p b_{-p} \big) \xi \rangle \big]\bigg|\\
&\hspace{1.4cm}\leq C N^{-\alpha}  (\log N)^{1/2}   \int_0^1 ds\;  \| \cK^{1/2} e^{sA} \xi \| \| (\cN_++1)^{1/2} e^{sA} \xi \| \,.
\end{split} \]
With Prop. \ref{prop:ANgrow}, Lemma \ref{prop:AHNgrow}, we conclude that 
\[
\begin{split}
\Bigg|\frac 1 2 \sum_{\substack{p \in \L^*}}  \widehat{\o}_N (p)\;  &\big[ \langle \xi, e^{-A}\big( b^*_p b^*_{-p} + b_p b_{-p} \big)e^{A} \xi \rangle - \langle \xi , 
\big( b^*_p b^*_{-p} + b_p b_{-p} \big) \xi \rangle \big]\Bigg|\\
&\leq C N^{-\alpha}  (\log N)^{1/2}   \left[ 
\| \cH_N^{1/2} \xi \| + N^{1/2} \| \cN_+^{1/2} \xi \| \right] 
\| (\cN_++1)^{1/2}  \xi \| \\ &\leq C N^{1-\alpha} \| (\cH_N + 1)^{1/2} \xi \|^2  \,.
\end{split}
\]
\end{proof}
\subsection{Contributions from $e^{-A}\cC_N e^{A}$}\label{sec:RN-cC}

In Section \ref{sub:Rfin} we will analyse the contributions to $\cR_{N,\a}$ arising from conjugation of the cubic operator $\cC_N$ defined in  \eqref{eq:GNeff-deco2}. To this aim we will need some properties of the commutator $[\cC_N,  A]$, as established in the following proposition.  
\begin{prop}\label{prop:commAcCN} Let $A$ be defined in \eqref{eq:defA}. Then, there exists a constant $C>0$ such that
	\begin{equation*}
	\begin{split}
	\big[\cC_N, A \big]  =&\;  2\sum_{r, v\in \L^*_+}  \big[\widehat{V}(r/e^N)\eta_r+\widehat{V}((r+v)/e^N)\eta_r\big]a^*_va_v\Big(1-\fra{\cN_+}N\Big)  +\delta_{\cC_N} \\ 
	\end{split}
	\end{equation*}
	where 
	\begin{equation} \label{eq:RN-cCN1}
	\begin{split} 
	| \langle \xi, \delta_{\cC_N}   \xi \rangle |\leq &\;   C N^{3/2-\a} \|\cH_N^{1/2}\xi\|  \|( \cN_++1)^{1/2}\xi\| 
	\end{split}
	\end{equation}
for all $\alpha > 0$, $\xi \in \cF^{\leq N}_+$, and $N \in \bN$ large enough.
\end{prop}
\begin{proof} 
We consider the commutator 
	\[ \big[\cC_N, A \big] =  \sum_{\substack{p,q \in \L_+^* : p+q \not = 0 \\ r,v \in \L^*_+}} \widehat{V} (p/e^N) \eta_r \big[ b_{p+q}^* a_{-p}^* a_q , b_{r+v}^* a_{-r}^* a_v - a_v^* a_{-r} b_{r+v} \big] + \hc  \,.
	\]
 As in the proof of Prop. \ref{prop:RN-K}, we use the commutators from the proof of Prop. 8.8 in \cite{BBCS3} to conclude that  
\begin{equation*}\label{eq:dec-CA}  \big[\cC_N, A \big] = 2\sum_{r, v\in \L^*_+}  \big[\widehat{V}(r/e^N)\eta_r+\widehat{V}((r+v)/e^N)\eta_r\big]a^*_va_v\frac{N-\cN_+}N + \sum_{j=1}^{12} ( \Xi_j + \hc ) \end{equation*}
	where 
	\begin{equation*}
	\begin{split}
	\Xi_1 :=&\; -\sum_{\substack{ r, v,p\in \L^*_+, \\ 
			p\neq v }}\widehat{V}(p/e^N) \eta_r b^*_{r+v}b^*_{-r} a^*_{-p} a_{v-p},\\ 
	\Xi_2 :=&\;  \sum_{\substack{ r,v,p\in \L^*_+ \\  r \neq -p }}\widehat{V}(p/e^N)\eta_r (1-\cN_+/N)a^*_{v}a^*_{-p} a_{-r-p} a_{r+v} ,\\ 
	\Xi_3 :=&\; \sum_{\substack{ r, v,  p\in\Lambda_+^*: \\r+v\neq p  }} \widehat{V}(p/e^N)\eta_r (1-\cN_+/N)a_v^* a^*_{-p} a_{-r}  a_{r+v-p} ,\\ 
	\end{split}
	\end{equation*}
	as well as 
	\begin{equation*}
	\begin{split}
	\Xi_4 :=&\; -\frac{1}{N}\sum_{\substack{ r,v,p,q\in\Lambda_+^*: p+q \neq 0}}\widehat{V}(p/e^N)\eta_r  a_v^* a^*_{p+q}a^*_{-p}a_{-r} a_{r+v} a_q,\\ 
	\Xi_5 :=&\; -\frac{1}{N}\sum_{\substack{ r, v, q \in\Lambda_+^*:\\ r+v \neq q     }} \widehat{V}((r+v)/e^N) \eta_r  	a_v^*a^*_{q-r-v}a_{-r}   a_q , \\
	\Xi_6 :=&\; -\frac{1}{N}\sum_{\substack{ r,v,q  \in\Lambda_+^*:\\ r \neq -q     }}  \widehat{V}(r/e^N) \eta_r   a_v^*a^*_{q+r}a_{r+v}   a_q	\\  
	\Xi_7 :=&\; \sum_{\substack{ r,v, p \in\Lambda_+^*: \\r+v \neq -p     }}  \widehat{V}(p/e^N) \eta_r   b^*_{p+r+v} b^*_{-p} a^*_{-r} a_v ,\\ 
	\Xi_{8} :=&\;\sum_{\substack{ r, v, p\in\Lambda_+^*:\\ r\neq -p     }}  \widehat{V}(p/e^N) \eta_r  b^*_{p-r}b^*_{r+v}a^*_{-p}  a_v ,\\ 
	\Xi_{9} :=&\; -\sum_{\substack{ r, v, q \in\Lambda_+^*:\\ q\neq v     }} \widehat{V}(v/e^N) \eta_r	 b^*_{q-v}b^*_{r+v}a^*_{-r} a_q ,\\ 
	\Xi_{10} :=&\; \sum_{\substack{ r, v, q\in\Lambda_+^*: \\r\neq -q     }} \widehat{V}(r/e^N) \eta_r	 b^*_{q+r}a_v^*a_qb_{r+v} ,\\ 
	\Xi_{11} :=&\; -\sum_{\substack{ r, v, p\in\Lambda_+^*: \\p\neq -v     }} \widehat{V}(p/e^N) \eta_r	 b^*_{p+v}a^*_{-p} a_{-r}b_{r+v}	 ,\\ 
	\Xi_{12} :=&\; \sum_{\substack{ r, v , q\in\Lambda_+^*: \\q\neq r+v     }} \widehat{V}((r+v)/e^N) \eta_r	 b^*_{q-r-v}a_v^* a_{-r} b_q \,. 
	\end{split}
	\end{equation*}      	
To prove the proposition, we have to show that all terms $\Xi_j$, $j=1,\dots , 12$, satisfy the bound  \eqref{eq:RN-cCN1}. We bound $\Xi_1$ in position space, with Cauchy-Schwarz, by
	\[\begin{split}
	\big|   \langle \xi, \Xi_1 \xi \rangle \big| &\leq  C\int_{\L^3}dxdydz e^{2N}V(e^N(x-y))|\check{\eta}(x-z)|\|\check{a}_x\xi\|\|\check{a}_x \check{a}_y\check{a}_z\xi\|\\
	&\leq C\left[\int_{\L^3}dxdydz \, e^{2N}V(e^N(x-y))\|\check{a}_x \check{a}_y\check{a}_z\xi\|^2\right]^{1/2}\\
	&\hspace{3cm}\times\left[\int_{\L^3}dxdydz \, e^{2N}V(e^N(x-y))|\check{\eta}(x-z)|^2\|\check{a}_x\xi\|^2\right]^{1/2}\\
	& \leq C\|\eta\|\|(\cN_++1)^{1/2}\xi\|\|\cV_N^{1/2}\cN_+^{1/2}\xi\|\\
	& \leq CN^{1/2-\a}\|(\cN_++1)^{1/2}\xi\|\|\cV_N^{1/2}\xi\|.
	\end{split}\]
We can proceed similarly to control $\Xi_9$. We obtain 
	\[ \big|   \langle \xi, \Xi_9 \xi \rangle \big| \leq C N^{1/2-\a} \|(\cN_++1)^{1/2}\xi\| \|\cV_N^{1/2}\xi\|. \]
	The expectations of the terms $\Xi_3$ and $\Xi_{12}$ can be bounded analogously: 
	\[\begin{split}
	\big|   &\langle \xi, \Xi_3 \xi \rangle \big|  +\big| \langle \xi, \Xi_{12}\xi \rangle \big| \\ &\leq  C\int_{\L^3}dxdydz \, e^{2N}V(e^N(x-y))(|\eta(x-z)|+|\eta(y-z)|)\|\check{a}_x\check{a}_y\xi\|\|\check{a}_x \check{a}_z\xi\|\\
	&\leq C\left[\int_{\L^3}dxdydz \, e^{2N}V(e^N(x-y))\|\check{a}_x \check{a}_y\xi\|^2(|\eta(x-z)|^2+|\eta(y-z)|^2)\right]^{1/2}\\
	&\hspace{3cm}\times\left[\int_{\L^3}dxdydz \, e^{2N}V(e^N(x-y))\|\check{a}_x\check{a}_z\xi\|^2\right]^{1/2}\\
	& \leq C\|\eta\|\|(\cN_++1)\xi\|\|\cV_N^{1/2}\xi\|\\
	& \leq CN^{1/2-\a}\|(\cN_++1)^{1/2}\xi\|\|\cV_N^{1/2}\xi\|.
	\end{split}\]
As for $\Xi_4$, we find 
	\[\begin{split}
	|\langle \xi, \Xi_4 \xi\rangle | =&\; \bigg|\frac 1 N \int_{\Lambda^2}dxdydz \, e^{2N} V(e^N(y-z))  \langle \xi , \check{a}^*_{x}\check{a}^*_{y}\check{a}^*_{z}\check{a}(\check{ \eta}_x)\check{a}_{x}\check{a}_{y} \xi\rangle  \bigg|\\ 
	&\leq  C N^{-1}\|\eta\|\int_{\Lambda^2}dxdydz \, e^{2N} V(e^N(y-z))  \| \check{a}_{x}\check{a}_{y} \check{a}_{z}\xi \| \|\cN_+^{1/2} \check{a}_x\check{a}_{y}\xi \| \\ 
	&\leq C N^{-1} \|\eta\|\left[\int_{\Lambda^2}dxdydz\, e^{2N} V(e^N(y-z))  \| \check{a}_{x}\check{a}_{y} \check{a}_{z}\xi \|^2\right]^{1/2}\\&\hspace{1cm}\times \left[\int_{\Lambda^2}dxdydz\, e^{2N} V(e^N(y-z)) \|\cN_+^{1/2} \check{a}_x\check{a}_{y}\xi \|^2\right]^{1/2}\\
	&\leq C N^{1/2-\a}\|\cV_N^{1/2}\xi\| \|\cN_+^{1/2}\xi\|\,.
	\end{split}\]
The  terms $\Xi_5$ and $\Xi_6 $ can be bounded in momentum space, using \eqref{eq:VetaN}.
Hence,
	\[\begin{split}
	&  |  \langle \xi, \Xi_5 \xi \rangle   | + |  \langle \xi, \Xi_6 \xi \rangle   |\\
	&\leq  CN^{-1}\sum_{r, v, q \in\Lambda_+^*}  \!\!\! \Bigg(  \frac{\widehat{V}((v+r)/e^N)}{|v|} |\eta_r| |v|\| a_{v} a_{q-r -v} \xi\|  \| a_{-r}   a_q \xi\| \\ &\hspace{3cm}+  \frac{\widehat{V}(r/e^N)}{|r+v|}|\eta_r| |r+v|  \|a_{r+q}a_v  \xi\| \| a_{q}a_{r+v} \xi\|\Bigg)\\
	 &\leq CN^{1/2-\a} \|(\cN_++1)^{1/2}\xi\| \|\cK^{1/2}\xi\|. \end{split}\]
Similarly we have
	 	\[\begin{split}
	 	|  \langle \xi, \Xi_2 \xi \rangle   | + |  \langle \xi, \Xi_{10} \xi \rangle   | &
	 	\leq  \sum_{r, v, p \in\Lambda_+^*}  \!\!\! \Bigg(  \frac{\widehat{V}(p/e^N)}{|p|}|\eta_r|  |p|\| a_{v}a_{-p}  \xi\|  \|  a_{r+v}a_{-r-p} \xi\| \\ &\hspace{1.5cm} +  \frac{\widehat{V}(r/e^N)}{|r+v|}|\eta_r| |r+v|  \| a_{q}a_{r+v} \xi\| \|a_{r+q}a_v  \xi\|\Bigg)\\
	 	 &\leq CN^{3/2-\a} \|(\cN_++1)^{1/2}\xi\| \|\cK^{1/2}\xi\|. \end{split}\]
	Next, we rewrite $ \Xi_7$, $\Xi_8$ and $\Xi_{11}$ as 	
	\[\begin{split}
	\Xi_7 =&\;  \int_{\Lambda^2}dxdy\; e^{2N} V(e^N(x-y))\check{b}^*_{x}\check{b}^*_{y} a^*(\check{ \eta}_x)\check{a}_{x}\,, \\ 
	\Xi_8 = &\;  \int_{\Lambda^2}dxdydz\; e^{2N} V(e^N(x-y))\check{\eta}(z-x)\check{b}^*_{x}\check{b}^*_{z}\check{a}^*_y\check{a}_{z} \,, \\
	\Xi_{11} =&\;  -  \int_{\Lambda^2}dxdy\; e^{2N} V(e^N(x-y))\check{b}^*_{x}\check{a}^*_{y}\check{a}(\check{ \eta}_x)\check{b}_{x} \,.
	\end{split}\]
	Thus, we obtain 
	\[\begin{split}
	|\langle \xi, \Xi_7 \xi \rangle | &\leq C\|\eta\|\int_{\Lambda^2}dxdy\; e^{2N} V(e^N(x-y))\; \|\cN_+^{1/2} \check{a}_{x}\check{a}_{y} \xi \| \|\check{a}_{x}
	\xi \|\\
	&\leq C\|\eta\|\|\cN_+^{1/2}\cV_N^{1/2}\xi\|\|\cN_+^{1/2}\xi\|\\
	&\leq CN^{1/2-\a}\|\cV_N^{1/2}\xi\|\|\cN_+^{1/2}\xi\|\,,
	\end{split}\]
	as well as  
	\[\begin{split}
	&|\langle \xi, \Xi_8 \xi\rangle |\\
	&\leq  C\int_{\Lambda^2}dxdydz\; e^{2N} V(e^N(x-y))  |\check{\eta}(x-z)|\| \check{a}_{x}\check{a}_{y} \check{a}_{z}\xi \| \| \check{a}_z\xi \| \\
	&\leq C\left[\int_{\Lambda^2}dxdydz\; e^{2N} V(e^N(x-y))  \| \check{a}_{x}\check{a}_{y} \check{a}_{z}\xi \|^2\right]^{1/2}\\&\hspace{3cm}\times \left[\int_{\Lambda^2}dxdydz\; e^{2N} V(e^N(x-y)) |\eta(x-z)|^2\|\check{a}_z\xi \|^2\right]^{1/2}\\
	&\leq C N^{1/2-\a}\|\cV_N^{1/2}\xi\| \|\cN_+^{1/2}\xi\|,
	\end{split}\]
	and
	\[\begin{split}
	|\langle \xi, \Xi_{11} \xi \rangle | &\leq C\|\eta\|\int_{\Lambda^2}dxdy\; e^{2N} V(e^N(x-y))\, \| \check{a}_{x}\check{a}_{y} \xi \| \|\cN_+^{1/2}\check{a}_{x}
	\xi \|\\
	&\leq C\|\eta\|\|\cV_N^{1/2}\xi\|\|\cN_+\xi\| \leq CN^{1/2-\a}\|\cV_N^{1/2}\xi\|\|\cN_+^{1/2}\xi\|.
	\end{split}\]
	Collecting all the bounds above, we arrive at  \eqref{eq:RN-cCN1}. 
\end{proof}

\subsection{Proof of Proposition \ref{prop:RN}}\label{sub:Rfin}

With the results of Sections \ref{sec:aprioribnds}-\ref{sec:RN-cC}, we can now show Proposition \ref{prop:RN}. We assume $\alpha > 2$.  From Eq. \eqref{eq:GNeff-deco},  Prop. \ref{prop:cON} and Prop. \ref{prop:RN-cZ} we obtain that 
\begin{equation*}
\begin{split} 
\cR_{N,\a} =\; & e^{-A}\cG^{\text{eff}}_{N,\a} e^A  \\
= \; & \frac 1 2\, \widehat{\o}_N(0)  (N-1) (1-\cN_+/N)+ \big[ 2 N \widehat V(0)- \frac{1}{2} \widehat{\o}_N(0)  \big] \cN_+ ( 1- \cN_+/N ) \\
& + \frac 12  \sum_{p\in\L^*_+}  \widehat{\o}_N(p)\big[ b^*_p b^*_{-p} + b_p b_{-p} \big] + \cK  + \cC_N+ \cV_N    \\
& + \int_0^1 ds\; e^{-sA}\big [ \cK  + \cC_N+ \cV_N, A \big]e^{sA}  +\cE_{\cR}^{(1)}
\end{split}
\end{equation*}
with 
\[
\pm \cE_{\cR}^{(1)} \leq  C N^{1-\a} (\cH_N + 1)\,.
\]
From Prop.  \ref{prop:commAcVN}, Prop. \ref{prop:RN-K} and Prop.~\ref{prop:commAcCN}, we can write, for $N$ large enough, 
\begin{equation*}
\begin{split}  
&[ \cK  + \cC_N+ \cV_N, A \big] \\
&= \frac 1 {\sqrt N}\sum_{\substack{r,v\in\Lambda_+^*}} \!\!\!\widehat{\o}_N(r) \big[ b^*_{r+v}a^*_{-r} a_v+ \text{h.c.}\big]  -\sqrt N\sum_{\substack{r,v, \in \L^*_+,\\ p \neq-q} }\widehat{V}(r/e^N) \big[ b^*_{r+v}a^*_{-r} a_v + \text{h.c.}\big] \\
&\hspace{0.5cm} + 2\sum_{r,v \in  \L^*_+}  \big[\widehat{V}(r/e^N)\eta_r+\widehat{V}((r+v)/e^N)\eta_r\big] a^*_v a_v (1-\cN_+/N) +\cE_{\cR}^{(2)} 
\end{split}
\end{equation*}
where  
\[ \begin{split}\label{eq:calER-2}
| \langle \xi, \cE_{\cR}^{(2)} \xi \rangle|  \leq  & C N^{1/2-\alpha} (\log N)^{1/2}  \|\cH_N^{1/2} \xi \|^2 + C  N^{3/2-\a}  \|\cH_N^{1/2} \xi \|  \| (\cN_++1)^{1/2}\xi \|  \\
&+ C N^{-1}(\log N)^{1/2}  \|\cH_N^{1/2} \xi \|  \| (\cN_++1)^{1/2}\xi \|   \,.
\end{split}\]
for all $\xi \in \cF^{\leq N}_+$.  From Prop. \ref{prop:ANgrow}, Prop. \ref{prop:AHNgrow} and recalling the definition (\ref{eq:GNeff-deco2}) of the operator $\cC_N$, 
we deduce that 
\begin{equation}\label{eq:prRN4}
\begin{split}  
&\int_0^1 ds\;  e^{-sA}[ \cK  + \cC_N+ \cV_N, A \big] e^{sA} \\
&= \int_0^1 ds \;  e^{-sA} \Big[-\cC_N + \frac 1 {\sqrt N}\sum_{\substack{r, v\in \L^*_+} } \widehat{\o}_N(r)\big[ b^*_{r+v}a^*_{-r} a_v + \text{h.c.}\big] \\
&\hspace{1cm} + 2\sum_{r,v\in \L^*_+}  \big[\widehat{V}(r/e^N)\eta_r+\widehat{V}((r+v)/e^N)\eta_r\big] a^*_v a_v\Big(1-\frac{\cN_+}{N}\Big)\Big]e^{sA}  + \cE_{\cR}^{(3)}
\end{split}
\end{equation}
with 
\[
\pm \cE_{\cR}^{(3)}  \leq  C  [ N^{2-\a} + N^{-1/2} (\log N)^{1/2} ] (\cH_N +1) 
\] for $N \in \bN$ sufficiently large. 

We now rewrite 
\be \label{eq:RNfin-quad} \begin{split}
2\sum_{r,v\in \L^*_+}&  \big[\widehat{V}(r/e^N) \eta_r+\widehat{V}((r+v)/e^N)\eta_r\big] a^*_v a_v\Big(1-\frac{\cN_+}{N}\Big) \\
=\; &  4  \sum_{r,v\in \L^*_+} \widehat{V}(r/e^N)\eta_r a^*_v a_v\Big(1-\frac{\cN_+}{N}\Big) \\
& + 2\sum_{r,v\in \L^*_+} \big[\widehat{V}((r+v)/e^N)-\widehat{V}(r/e^N)\big]  \eta_ra^*_v a_v\Big(1-\frac{\cN_+}{N}\Big) := \text{Q}_1 + \text{Q}_2\,.
\end{split} \ee
With Lemma \ref{lm:propomega}, part iii) we get 
\begin{equation} \label{eq:prRN5}\begin{split}
&\bigg|2 \sum_{r\in \L^*} \widehat{V}(r/e^N)\eta_r -  \big [ 2 \widehat \o_N(0)  - 2N\widehat{V}(0)  \big] \bigg| \leq  \frac{C}{N}\,, \end{split}
\end{equation}
and therefore, using Lemma \ref{lm:cNops} and \eqref{eq:prRN5}  
\begin{equation}\label{eq:prRN6.1}
\begin{split}  
&\pm\bigg[ e^{-sA} \text{Q}_1 e^{sA} - 2\big [ 2 \widehat \o_N(0)  - 2N\widehat{V}(0)  \big]\sum_{v\in \L^*_+}a^*_v a_v\bigg(1-\frac{\cN_+}{N}\bigg)\bigg] \\
& \hskip 6cm \leq  C N^{1-\a} (\cN_++1) + \frac C N \,\cN_+\,.   
\end{split}
\end{equation}
On the other hand it is easy to check that $e^{-sA} \text{Q}_2 e^{sA}$ is an error term; to this aim we notice that
\begin{equation*}
\label{eq:prRN7.2}\begin{split}
& \bigg| \sum_{r\in \L^*}  \big[ \widehat{V}(r/e^N)\eta_r- \widehat{V}((r+v)/e^N)\eta_r\big] \bigg| \leq  C N  |v|  e^{-N}\,. \end{split}
\end{equation*}
Hence  with Props. \ref{prop:ANgrow} and \ref{prop:AHNgrow} we find
\be \label{eq:prRN6.2}
\pm \big[ e^{-sA} \text{Q}_2 e^{sA} \big] \leq C N e^{-N}  e^{-sA} \cN_+^{1/2} \cK^{1/2} e^{sA} \leq C N^{2} e^{-N} (\cH_N +1) \,.
\ee

To handle the second term on the second line of (\ref{eq:prRN4}), we apply Prop.~\ref{prop:RN-K} and then Prop. \ref{prop:ANgrow} and Prop. \ref{prop:AHNgrow}
\begin{equation}\label{eq:prRN7}
\begin{split}
&\pm\bigg(\frac 1 {\sqrt N} \int_0^1 ds\;  \sum_{\substack{r,v \in \L^*_+} }\widehat{\o}_N(r) \Big[e^{-sA} b^*_{r+v}a^*_{-r} a_ve^{sA}- b^*_{r+v}a^*_{-r} a_v \Big]  + \text{h.c.}\bigg)\\
& \hspace{1cm}= \pm \bigg( \frac 1 {\sqrt N}\int_0^1 ds\; \int_0^s dt\;\sum_{\substack{r, v\in \L^*_+}}  \widehat{\o}_N(r) e^{-tA} \Big[  b^*_{r+v}a^*_{-r} a_v , A \Big] e^{tA}\bigg) \\ 
&  \hspace{1cm} \leq  C  \int_0^1 ds\; \int_0^s dt\, e^{-tA} \big(N^{-\a} (\log N)\, \cK + N^{-1}(\cN_++1) \big)  e^{tA} \\
& \hspace{1cm} \leq C N^{1 -\a} \log N (\cH_N +1) \,.
\end{split}
\end{equation}
As for the first term on the second line of (\ref{eq:prRN4}), we use again Prop. \ref{prop:commAcCN}. Using  \eqref{eq:RNfin-quad}, \eqref{eq:prRN6.1} and \eqref{eq:prRN6.2} we have 
\begin{equation}\label{eq:prRN9}
\begin{split} 
\int_0^1 ds\;  e^{-sA} \cC_Ne^{sA } - \cC_N   &=   \int_0^1 ds\;  \int_0^{s}dt \;  e^{-tA } [\cC_N, A] e^{tA} \\
&=   \big [ 2 \widehat \o_N(0)  - 2N\widehat{V}(0)  \big]\sum_{p\in \L^*_+}a^*_pa_p\Big(1-\frac{\cN_+}{N}\Big) + \cE_{\cR}^{(4)}   
\end{split}
\end{equation} 
with $\pm \cE_{\cR}^{(4)}  \leq  C N^{2-\a} (\cH_N+1) + C N^{-1}  (\cN_++1)$. \\

Inserting the bounds (\ref{eq:prRN6.1}),  (\ref{eq:prRN6.2}),  (\ref{eq:prRN7}) and (\ref{eq:prRN9}) into (\ref{eq:prRN4}) we arrive at  
\begin{equation*}\label{eq:prRN-l}
\begin{split} 
\cR_{N,\a}= &\;  \frac 1 2  (N-1)\,  \widehat \o_N(0)   (1-\cN_+/N) + \frac 1 2  \widehat \o_N(0) \,\cN_+ \left( 1 - \cN_+/N \right)   \\
& +  \widehat \o_N(0) \sum_{p\in \L^*_+}a^*_pa_p \Big(1-\frac{\cN_+}{N} \Big)+  \frac 12 \sum_{p\in \L^*_+}  \widehat{\o}_N(p)\big[ b^*_p b^*_{-p} + b_p b_{-p} \big]  \\
& +\frac 1 {\sqrt N} \sum_{\substack{r,v\in \L^*_+:\\ r\neq-v} } \widehat{\o}_N(r)\big[ b^*_{r+v}a^*_{-r} a_v + \text{h.c.}\big] + \cH_N   + \cE_\cR 
\end{split}
\end{equation*} 
with 
\[
\pm \cE_\cR  \leq   C [ N^{2-\a} + N^{-1/2} (\log N)^{1/2} ] (\cH_N +1) 
\]
for $N \in \bN$ sufficiently large. 

\appendix

\section{Analysis of $ \cG_{N,\a}$} \label{sec:GN} 

The aim of this section is to show   Prop. \ref{prop:GN}. 
From (\ref{eq:cLN}) and (\ref{eq:GN}), we can decompose
\[ \cG_{N,\a} = e^{-B} \cL_N e^{B} = \cG^{(0)}_{N,\a} + \cG_{N,\a}^{(2)} + \cG_{N,\a}^{(3)} + \cG_{N,\a}^{(4)} \]
with 
\[ \cG_{N,\a}^{(j)} = e^{-B} \cL_N^{(j)} e^{B}\,. \]
To analyse  $\cG_{N,\a}$ we will need precise informations on the action of the generalized Bogoliubov transformation $e^{B}$ with $B$ the antisymmetric operator defined in \eqref{eq:defB}, which are summarized in subsection \ref{sub:genBog}. Then, in the subsections \ref{sub:G0}--\ref{sub:G4} we prove separate bounds for the operators $\cG_{N,\a}^{(j)}$, $j=0,2,3,4$, which we combine in Subsection \ref{sub:proofGN} to prove Prop. \ref{prop:GN}.  

The analysis in this section follows closely that of \cite[Section 7]{BBCS3} with some slight modifications due to the different scaling of the interaction potential and the fact that the kernel $\eta_p$ of $e^{B}$ is different from zero for all $p \in \L^*_+$ (in \cite{BBCS3} $\eta_p$ is different from zero only for momenta larger than a sufficiently large cutoff of order one). Moreover, while in three dimensions it was sufficient to choose the function  $\eta_p$ appearing in the generalized Bogoliubov transformation with $\| \eta \| $ sufficiently small but of order one, we need here  $\|\eta\|$ to be of order $N^{-\a}$ for some $\a>0$ large enough. As discussed in the introduction this is achieved by considering the Neumann problem for the scattering equation in \eqref{tlf} on a ball of radius $\ell=N^{-\a}$; as a consequence some terms depending on $\ell$ will be large, compared to the analogous terms in \cite{BBCS3}.

\subsection{Generalized Bogoliubov transformations} \label{sub:genBog}

In this subsection we collect important properties about the action of unitary operators of the form $e^{B}$, as defined in \eqref{eq:eBeta}. As shown in \cite[Lemma 2.5 and 2.6]{BBCS1}, we have, if $\| \eta \|$ is sufficiently small,
\begin{equation}\label{eq:conv-serie}
\begin{split} e^{-B} b_p e^{B} &= \sum_{n=0}^\infty \frac{(-1)^n}{n!} \text{ad}_{B}^{(n)} (b_p) \\
e^{-B} b^*_p e^{B} &= \sum_{n=0}^\infty \frac{(-1)^n}{n!} \text{ad}_{B}^{(n)} (b^*_p) \end{split} \end{equation}
where the series converge absolutely. To confirm the expectation that generalized Bogoliubov transformation act similarly to standard Bogoliubov transformations, on states with few excitations, we 
define (for $\| \eta \|$ small enough) the remainder operators 
\begin{equation} \label{eq:defD}
d_q =\sum_{m\geq 0}\frac{1}{m!} \Big [\text{ad}_{-B}^{(m)}(b_q) - \eta_q^m b_{\alpha_m q}^{\sharp_m }  \Big],\hspace{0.5cm} d^*_q =\sum_{m\geq 0}\frac{1}{m!} \Big [\text{ad}_{-B}^{(m)}(b^*_q) - \eta_q^m b_{\alpha_m q}^{\sharp_{m+1}}  \Big]\end{equation}
where $q \in \L^*_+$, $ (\sharp_m, \alpha_m) = (\cdot, +1)$ if $m$ is even and $(\sharp_m, \alpha_m) = (*, -1)$ if $m$ is odd. It follows then from (\ref{eq:conv-serie}) that 
\begin{equation}\label{eq:ebe} 
e^{-B} b_q e^{B} = \gamma_q  b_q +\s_q b^*_{-q} + d_q, \hspace{1cm} e^{-B} b^*_q e^{B} = \g_q b^*_q +\s_q b_{-q} + d^*_q  \end{equation} 
where we introduced the notation $\g_q = \cosh (\eta_q)$ and $\s_q = \sinh (\eta_q)$. It will also be useful to introduce remainder operators in position space. For $x \in \Lambda$, we define the operator valued distributions $\check{d}_x, \check{d}^*_x$ through
\begin{equation}\label{eq:ebex} e^{-B} \check{b}_x e^{B} = b ( \check{\g}_x)  +  b^* (\check{\s}_x) + \check{d}_x, \qquad 
e^{-B} \check{b}^*_x e^{B} = b^* ( \check{\gamma}_x)  +  b (\check{\s}_x) + \check{d}^*_x
\end{equation}
where $\check{\gamma}_x (y) = \sum_{q \in \Lambda^*} \cosh (\eta_q) e^{iq \cdot (x-y)}$ and $\check{\s}_x (y) = \sum_{q \in \Lambda^*} \sinh (\eta_q) e^{iq \cdot (x-y)}$.
The next lemma is taken from \cite[Lemma 3.4]{BBCS3}.
\begin{lemma} \label{lm:dp} 
	Let $\eta \in \ell^2 (\Lambda_+^*)$, $n \in \bZ$. For $p \in \L_+^*$, let $d_p$ be defined as in (\ref{eq:ebe}). If $\| \eta \|$ is small enough, there exists $C > 0$ such that  
	\begin{equation}\label{eq:d-bds}
	\begin{split} 
	\| (\cN_+ + 1)^{n/2} d_p \xi \| &\leq \frac{C}{N} \left[ |\eta_p| \| (\cN_+ + 1)^{(n+3)/2} \xi \| + \| \eta \| \| b_p (\cN_+ + 1)^{(n+2)/2} \xi \| \right], \\ 
	\| (\cN_+ + 1)^{n/2} d_p^* \xi \| &\leq \frac{C}{N} \, \| \eta \| \,\| (\cN_+ +1)^{(n+3)/2} \xi \| \end{split}  \end{equation}
	for all $p \in \L^*_+, \xi \in \cF_+^{\leq N}$. 
	In position space, with $\check{d}_x$ defined as in (\ref{eq:ebex}), we find   
	\begin{equation}\label{eq:dxy-bds} 
	\| (\cN_+ + 1)^{n/2} \check{d}_x \xi \| \leq  \frac{C }{N}\, \| \eta \| \Big[ \,\| (\cN_+ + 1)^{(n+3)/2} \xi \| +  \| b_x (\cN_+ + 1) ^{(n+2)/2}\xi \|\,. \Big] 
	\end{equation}
	Furthermore, letting $\check{\lis{d}}_x = \check{d}_x  + (\cN_+ / N) b^*(\check{\eta}_x)$, we find 
	\be \begin{split} \label{eq:splitdbd}
		\| (\cN_+ &+ 1)^{n/2} \check{a}_y \check{\lis{d}}_x \xi \| \\ &\leq \frac{C}{N} \, \Big[ \, \|\eta \|^2  \| (\cN_+ + 1)^{(n+2)/2} \xi \|  + \| \eta \| |\check{\eta} (x-y)|  \| (\cN +1)^{(n+2)/2}  \xi \| \\
		& \hspace{1cm} + \| \eta \| \| \check{a}_x (\cN_++1)^{(n+1)/2} \xi \| +  \|\eta \|^2 \|\check{a}_y (\cN_+ + 1)^{(n+3)/2} \xi \|\\
		& \hspace{1cm}  + \| \eta \| \| \check{a}_x \check{a}_y (\cN +1)^{(n+2)/2}  \xi \|   \, \Big]
	\end{split}\ee
	and, finally, 
	\begin{equation}\label{eq:ddxy}
	\begin{split} 
	\| (\cN_+ &+ 1)^{n/2} \check{d}_x \check{d}_y \xi \|  \\ &\leq \frac C {N^2} \Big[ \; \|\eta\|^2  \| (\cN_++ 1)^{(n+6)/2} \xi \| + \| \eta \| |\check{\eta} (x-y)|  \| (\cN_+ + 1)^{(n+4)/2}  \xi \| \\ 
	&\hspace{1cm} + \|\eta \|^2 \| {a}_x (\cN_+ + 1)^{(n+5)/2} \xi \|   + \| \eta \|^2 
	\|{a}_y (\cN_+ + 1)^{(n+5)/2} \xi \| \\ &\hspace{1cm}  
	+ \| \eta \|^2\, \|{a}_x {a}_y (\cN_+ +  1)^{(n+4)/2} \xi \| \; \Big] 
	\end{split} \end{equation}
	for all $\xi \in \cF^{\leq n}_+$. 
\end{lemma}

A first simple application of Lemma \ref{lm:dp} is the following bound 
on the growth of the expectation of $\cN_+$. 
\begin{lemma}\label{lm:Ngrow2} 
Assume $B$ is defined as in \eqref{eq:defB}, with $\eta \in \ell^2 (\L^*)$ and $\eta_p = \eta_{-p}$ for all $p \in \L^*_+$. Then, there exists a constant $C > 0$ such that
\[
 \Big| \langle \xi, \big[ e^{-B} \cN_+  e^{B} - \cN_+ \big] \xi \rangle \Big| \leq   \| \eta \| \| (\cN_+ + 1)^{1/2} \xi \|^2
\]
for all $\xi \in \cF_+^{\leq N}$. 
\end{lemma}
\begin{proof} With (\ref{eq:ebe}) we write
	\begin{equation*}
	\begin{split}
	e^{-B} &\cN_+  e^{B} - \cN_+ \\ = \; &\int_0^1 e^{-sB} [\cN_+ ,B] e^{s B}ds \\
	=\; &\int_0^1 \sum_{p \in \L^*_+}  \eta_p \, e^{-s B} ( b_p b_{-p}  + b^*_p b^*_{-p} ) e^{s B} \, ds \\
	= \; &\int_0^1 \sum_{p \in \L^*_+}  \eta_p \, \left[ (\g_p^{(s)} b_p + \s_p^{(s)} b_{-p}^* + d_p^{(s)}) ( \g_p^{(s)} b_{-p} + \s_p^{(s)} b_{-p}^* + d_{-p}^{(s)}) + \hc \right]   ds
	\end{split}
	\end{equation*}
	with $\g_p^{(s)} = \cosh (s \eta_p)$, $\s_p^{(s)} = \sinh (s \eta_p)$. Using $|\g^{(s)}_p| \leq C$ and $|\s_p^{(s)}| \leq C |\eta_p|$, (\ref{eq:d-bds}) in Lemma \ref{lm:dp} we arrive at  
	\[ \begin{split} \Big|  &\langle \xi, \big[ e^{-B} \cN_+  e^{B} - \cN_+ \big] \xi \rangle \Big| \\ 
&\leq C \| (\cN_+ + 1)^{1/2} \xi \|  \sum_{p \in \L^*_+} |\eta_p| \left[ |\eta_p| \| (\cN_+ + 1)^{1/2} \xi \| + \| b_{p} \xi \| \right] \leq C \| \eta \|  \| (\cN_+ + 1)^{1/2} \xi \|^2 \end{split} \]
\end{proof}

\subsection{Analysis of $ \cG_{N,\a}^{(0)}=e^{-B}\cL^{(0)}_N e^{B}$} 
\label{sub:G0}

We define $\cE_{N}^{(0)}$ so that 
\begin{equation*}\label{eq:G0C}
\begin{split}
\cG^{(0)}_{N,\a} &= e^{-B} \cL^{(0)}_N e^{B}= \frac12 \widehat{V} (0) (N+\cN_+-1)(N-\cN_+) + \cE_{N,\a}^{(0)}\,.
\end{split}
\end{equation*}
where we recall from (\ref{eq:cLNj}) that 
\begin{equation*}\label{eq:cLN02} \cL_{N}^{(0)} =\;\frac 12 \widehat{V} (0) (N-1 + \cN_+)(N-\cN_+ )\,.  \end{equation*}

\begin{prop}\label{prop:G0} Under the assumptions of Prop. \ref{prop:GN}, there exists a constant $C > 0$ such that  
	\begin{equation*}\label{eq:E0C}
	\pm \cE_{N, \a}^{(0)} \leq C N^{1-\a} (\cN_+ +1)
	\end{equation*}
	for all $\a > 0$ and $N \in \bN$ large enough. 
\end{prop}

\begin{proof} 
The proof follows \cite[Prop. 7.1]{BBCS3}. 

We write 
\[ \cL_N^{(0)} = \frac{N(N-1)}{2} \widehat{V} (0) + \frac{N}{2}\widehat{V} (0) \Big[ \sum_{q \in \L^*_+} b_q^* b_q - \cN_+ \Big]\,. \]
Hence, 
\begin{equation*}\label{eq:e01}
\begin{split}
\cE_{N}^{(0)} &= \frac{N}{2}\widehat{V} (0) \sum_{q \in \L_+^*} \left[ e^{-B} b_q^* b_q e^{B} - b_q^* b_q \right] - \frac{N}{2}  \widehat{V} (0)\left[ e^{-B} \cN_+ e^{B} - \cN_+ \right]\,. \end{split} \end{equation*}
To bound the first term we use (\ref{eq:ebe}), $|\g_q^2 - 1| \leq C \eta_q^2$, $|\s_q| \leq C |\eta_q|$, the first bound in (\ref{eq:d-bds}), Cauchy-Schwarz and the estimate $\| \eta \| \leq C N^{-\a}$. To bound the second term, we use Lemma \ref{lm:Ngrow2}. We conclude that 
\[ |\langle \xi , \cE_N^{(0)} \xi \rangle | \leq C N^{1-\alpha} \| (\cN_+ + 1)^{1/2} \xi \|^2\,. \]
\end{proof}

\subsection{Analysis of $\cG_{N,\a}^{(2)}=e^{-B}\cL^{(2)}_N e^{B}$}

We consider first conjugation of the kinetic energy operator. 
\begin{prop}\label{prop:K} Under the assumptions of Prop. \ref{prop:GN}, there exists $C > 0$ such that   
	\begin{equation} \label{eq:K-dec} \begin{split} e^{-B}\cK e^{B} = \; &\cK + \sum_{p \in \L^*_+} p^2 \eta_p ( b_p b_{-p} + b^*_p b^*_{-p} ) \\&+ \sum_{p \in  \L^*_+} p^2 \eta_p^2 \Big(\frac{N-\cN_+}{N}\Big) \Big(\frac{N-\cN_+ -1}{N}\Big) +\cE^{(K)}_{N,\a}
	\end{split}
	\end{equation}
	where 
	\begin{equation}\label{eq:errorKc}
	\begin{split}
| \langle \xi, \cE^{(K)}_{N,\a}  \xi \rangle | \leq  C N^{1/2 -\a} \| \cH_N^{1/2}\xi\| \| (\cN_++1)^{1/2}\xi \| + C N^{1-\a}  \| (\cN_++1)^{1/2}\xi \|^2
	\end{split}
	\end{equation}
for any $\a>1$, $\xi \in \cF^{\leq N}_+$ and  $N \in \bN$ large enough.
\end{prop}
\begin{proof} 
We proceed as in the proof of \cite[Prop. 7.2]{BBCS3}. We write 
\begin{equation} \label{eq:cKterms}
	\begin{split} 
	&e^{-B} \cK e^{B}-\cK \\ &= \int_0^1ds \sum_{p \in \L^*_+} p^2 \eta_p \Big [\big(\gamma_p^{(s)} b_p+ \s^{(s)}_p b^*_{-p}\big) \big(\gamma_p^{(s)} b_{-p} + \s^{(s)}_p b^*_{p}\big)\, +\hc\Big ]\\
	&\hspace{.3cm} + \int_0^1ds \sum_{p \in \L^*_+} p^2 \eta_p \big[\big(\gamma_p^{(s)} b_p+ \s^{(s)}_p b^*_{-p}\big) 
	d_{-p}^{(s)}+ d_p^{(s)} \big( \gamma_p^{(s)} b_{-p}+ \s^{(s)}_p b^*_{p}\big)+\hc\big]\\
	&\hspace{.3cm}+ \int_0^1ds \sum_{p \in \L^*_+} p^2 \eta_p\big[ d_p^{(s)} d_{-p}^{(s)} + \hc \big]\\
	&=: \text{G}_1+\text{G}_2+\text{G}_3
	\end{split}  
	\end{equation}
	with $\gamma_p^{(s)} = \cosh(s \eta_p)$, $\s^{(s)}_p =\sinh(s \eta_p)$ and where $d^{(s)}_p$ is defined as in (\ref{eq:defD}), with $\eta_p$ replaced by $s \eta_p$. We find 	
\begin{equation*} \label{eq:first2cK}
	\begin{split} 
	\text{G}_1	&= \sum_{p \in \L^*_+} p^2 \eta_p \big(b_p b_{-p} + b^*_{-p} b^*_{p} \big)+ \sum_{p \in \L^*_+} p^2 \eta_p^2 \left(1-\frac{\cN_+}{N}\right)+\cE^K_{1}
	\end{split}  
	\end{equation*}
	with
	\begin{equation} \nonumber
	\begin{split} 
	\cE^K_{1}=&\, 2 \int_0^1ds \sum_{p \in \L^*_+} p^2 \eta_p  (\s^{(s)}_p)^2 \big(b_p b_{-p} + b^*_{-p} b^*_{p} \big)\\
	&+\int_0^1ds \sum_{p \in \L^*_+} p^2 \eta_p \gamma_p^{(s)} \s^{(s)}_p (4b_p^*b_{p}-2N^{-1}a^*_pa_p)\\
	&+2\int_0^1ds \sum_{p \in \L^*_+} p^2 \eta_p  \left[ (\gamma_p^{(s)}-1) \s^{(s)}_p + (\s^{(s)}_p-s \eta_p) \right] \Big(1-\frac{\cN_+}{N}\Big) \,.
	\end{split}  
	\end{equation}
Since  $|\big((\gamma_p^{(s)})^2-1\big)|\leq C \eta_p^2$, $(\s^{(s)}_p)^2\leq C \eta_p^2$, $p^2 |\eta_p| \leq C $, $\| \eta \|_\io \leq N^{-\a }$, we can estimate 
\begin{equation} \label{eq:cE0K}
\begin{split} 
	|\langle\xi, &\cE^K_{1} \xi\rangle| \\ \leq \; & C \sum_{p \in \L^*_+} p^2 |\eta_p|^3 \|b_p\xi\|\|(\cN_++1)^{1/2} \xi\|+C  \sum_{p \in \L^*_+} p^2 \eta_p^2 \|a_p\xi\|^2+C\sum_{p \in \L^*_+} p^2 \eta_p^4 \| \xi\|^2 \\
\leq \; & C\| \eta \|  \|(\cN_++1)^{1/2} \xi\|^2 \leq C N^{-\a} \|(\cN_++1)^{1/2} \xi\|^2,
	\end{split}  
	\end{equation}
for any $\xi\in\cF_+^{\leq N}$. To bound the term $\text{G}_3$ in \eqref{eq:cKterms}, we switch to position space: 
	\[ \begin{split}  | \langle \xi , \text{G}_3 \xi \rangle | \leq \; &CN  \int_0^1 ds \int_{\L^2}  dx dy \left[ e^{2N} V(e^N(x-y)) + N^{2\a -1} \chi(|x-y|\leq N^{-\a}) \right] \\ &\hspace{3.5cm} \times \| (\cN_+ + 1)^{-1/2} \check{d}^{(s)}_x \check{d}^{(s)}_y \xi \| \| (\cN_+ + 1)^{1/2} \xi \| \end{split} \]
	With (\ref{eq:ddxy}), we obtain 
	\begin{equation}\label{eq:G3-bd} \begin{split} 
&| \langle \xi , \text{G}_3 \xi \rangle | \\ 
&\leq C N^{1-\a}  \int_{\L^2}  dx dy \left[ e^{2N} V(e^N(x-y)) + N^{2\a-1}\chi (|x-y|\leq N^{-\a}) \right]  \| (\cN_+ + 1)^{1/2} \xi \|^2 \\
& \,+ C  N^{-2\a}\int_{\L^2}  dx dy \left[ e^{2N} V(e^N(x-y)) + N^{2\a-1}\chi (|x-y|\leq N^{-\a}) \right]  \| (\cN_+ + 1)^{1/2} \xi \| \\ 
&\hspace{3.5cm} \times \Big[  \| \check{a}_x  (\cN_++1)\xi \| + \| \check{a}_y  (\cN_++1) \xi \| + \| \check{a}_x \check{a}_y (\cN_++1)^{1/2}  \xi \| \Big] \\ 
&\leq C N^{1-\a} \| (\cN_+ + 1)^{1/2} \xi \|^2 + C N^{1/2-\a} \| (\cN_+ + 1)^{1/2} \xi \| \| \cV_N^{1/2} \xi \| \,.
	\end{split} \end{equation}
Finally, we consider $\text{G}_2$ in \eqref{eq:cKterms}. We split it as $\text{G}_2 = \text{G}_{21} + \text{G}_{22} + \text{G}_{23} + \text{G}_{24}$, with 
	\begin{equation} \label{eq:cKtermsG2}
	\begin{split} 
	\text{G}_{21} &= \int_0^1ds \sum_{p \in \L^*_+} p^2 \eta_p \left( \gamma_p^{(s)} b_p d_{-p}^{(s)} + \hc \right), \\ \text{G}_{22} &=  \int_0^1ds \sum_{p \in \L^*_+} p^2 \eta_p \left( \s^{(s)}_p b^*_{-p} d_{-p}^{(s)} + \hc \right) \\
	\text{G}_{23} &= \int_0^1ds \sum_{p \in \L^*_+} p^2 \eta_p \left( \gamma_p^{(s)} d_p^{(s)} b_{-p}+ \hc \right), \\ \text{G}_{24} &=  \int_0^1ds \sum_{p \in \L^*_+} p^2 \eta_p \left( \s^{(s)}_p d_p^{(s)} b^*_{p}+ \hc \right) \,.
	\end{split} \end{equation}
	We consider $\text{G}_{21}$ first. We write
	\begin{equation*}\label{eq:G21}
	\text{G}_{21} = \; - \sum_{p \in \L^*_+} p^2 \eta_p^2 \, \frac{\cN_+ + 1}{N} \frac{N-\cN_+}{N} + \left[  \cE_{2}^K + \hc \right] \end{equation*}
	where $\cE_{2}^K = \sum_{j=1}^3 \cE_{2j}^K$, with  
	\begin{equation}\label{eq:cE2} \begin{split} 
	\cE_{21}^K = \; &\frac{1}{2N} \sum_{p \in \L^*_+} p^2 \eta_p^2 (\cN_++1) \big( b^*_p b_p - \frac 1 N a_p^* a_p\big) , \\ 
	\cE_{22}^K = \; &\int_0^1 ds \sum_{p \in \L^*_+} p^2 \eta_p (\gamma_p^{(s)} - 1)  b_p d_{-p}^{(s)}
	\,, \\  \cE_{23}^K = \; & \int_0^1 ds \sum_{p \in \Lambda_+^*} p^2 \eta_p b_p \lis{d}_{-p}^{(s)} \,.
	  \end{split} \end{equation}
and where we introduced the notation $\lis{d}^{(s)}_{-p} = d_{-p}^{(s)} + s \eta_p (\cN_+ / N) b_p^*$. 
With \eqref{eq:etapio}, we find 
	\begin{equation}\label{eq:cEK21} 
|\langle \xi , \cE_{21}^K \xi \rangle | \leq C \sum_{p \in \L^*_+} \eta_p \| a_p \xi \|^2 \leq C N^{-\a}\|  \cN_+^{1/2} \xi \|^2 
\end{equation}
Using $|\gamma_p^{(s)} - 1| \leq C \eta_p^2$ and (\ref{eq:d-bds}), we obtain  
\begin{equation}\label{eq:cEK22} \begin{split} |\langle \xi , \cE_{22}^K \xi \rangle | &\leq \sum_{p \in \L^*_+} p^2 |\eta_p|^3 \| \cN_+^{1/2} \xi \| \| d^{(s)}_{-p} \xi \| 
\leq  C N^{-3\a} \| (\cN_+ + 1)^{1/2} \xi \|^2. \end{split} \end{equation}
	To control the third term in (\ref{eq:cE2}), we use (\ref{eq:eta-scat0}) and we switch to position space. We find 
\begin{equation} \label{eq:star1} \begin{split} \cE_{23}^K = \; & -N  \int_0^1 ds \int_{\Lambda^2}  dx dy \, e^{2N} V(e^N(x-y)) f_{N,\ell} (x-y) \check{b}_x \check{\lis{d}}^{\,(s)}_y \\ &+ N  \int_0^1 ds e^{2N} \lambda_\ell \int_{\Lambda^2} dx dy \, \chi_\ell (x-y) f_{N,\ell} (x-y) \check{b}_x \check{\lis{d}}^{\,(s)}_y \\
=\; & \cE_{231}^K+\cE_{232}^K \,.
\end{split} \end{equation} 
With (\ref{eq:splitdbd}) and $|\check\eta(x-y)| \leq C N$, we obtain 
\be \begin{split}  \label{eq:cE231K}
|\langle \xi , \cE_{231}^K \xi \rangle |  \leq \;  & N \int_0^1 ds \int_{\Lambda^2}  dx dy \, e^{2N} V(e^N(x-y)) \\ &\hspace{3cm} \times  \| (\cN_+ + 1)^{1/2} \xi \| \| (\cN_+ + 1)^{-1/2}  \check{a}_x \check{\lis{d}}^{(s)}_y \xi \| \\	
\leq \; &C N^{1-\a} \| (\cN_+ + 1)^{1/2} \xi \|^2 + C N^{1/2-\a} \|(\cN_+ + 1)^{1/2} \xi \| \| \cV_N^{1/2} \xi \|. \end{split} \ee
As for $\cE_{232}^K$, with (\ref{eq:splitdbd}) and Lemma \ref{lm:propomega} (recalling $\ell= N^{-\a})$, we find
\be \begin{split} \label{eq:EK232}
|\langle \xi , \cE_{232}^K \xi \rangle |  \leq \; & C N^{-\a} \| (\cN_+ + 1)^{1/2} \xi \|^2  \\ &+ \int_{\L^2}  dx dy \, \c(|x-y| \leq  N^{-\a})\| (\cN_+ + 1)^{1/2} \xi \| \| \check{a}_x \check{a}_y \cN_+^{1/2} \xi \| 
\end{split} \end{equation} 
To bound the last term on the r.h.s. of \eqref{eq:EK232} we use H\"older's and Sobolev inequality $\| u \|_q \leq C q^{1/2} \|u \|_{H^1}$, valid for any $2 \leq q <\io$. We find 
	\[
		\begin{split}
		 \int_{\L^2} & dx dy \, \c(|x-y| \leq  N^{-\a})\| (\cN_+ + 1)^{1/2} \xi \| \| \check{a}_x \check{a}_y \cN_+^{1/2} \xi \| \\
		&\leq  C \| (\cN_+ + 1)^{1/2} \xi \|   \int_{\L} dx \left(\int_{\L} dy \, \c(|x-y| \leq  N^{-\a})\right)^{1-1/q}\left(\int_{\L} dy \, \| \check{a}_x \check{a}_y \cN_+^{1/2} \xi \|^q\right)^{1/q}\\
		&\leq  CN^{2\a/q -2\a} \| (\cN_+ + 1)^{1/2} \xi \|  \int_{\L} dx \left(\int_{\Lambda}dy \, \| \check{a}_x \check{a}_y \cN_+^{1/2} \xi\|^q\right)^{1/q}\\ 
		& \leq C q^{1/2}N^{2\a/q -2\a} \| (\cN_+ + 1)^{1/2} \xi \|   \\
& \hskip 3cm \times  \left[\int_{\L^2} dxdy \, \|\check{a}_x\nabla_y\check{a}_y\cN_+^{1/2}\xi \|^2+\int_{\L^2} dxdy \, \| \check{a}_x \check{a}_y \cN_+^{1/2} \xi \|^2\right]^{1/2}  \\
		& \leq  Cq^{1/2}N^{2\a/q-2\a} \| (\cN_+ + 1)^{1/2} \xi \|  \left[\ \| \cK^{1/2}\cN_+ \xi\|+\| \cN_+^{3/2} \xi \|\right]\,.
		\end{split}\]
Choosing $q=\log N$, we get    
	\begin{equation}
		\label{eq:estimatechi}
	\begin{split}
 \int_{\L^2} &dx dy \, \c(|x-y| \leq  N^{-\a})\| (\cN_+ + 1)^{1/2} \xi \| \| \check{a}_x \check{a}_y (\cN_++1)^{1/2} \xi \| \\
 \leq &\; CN^{1-2\a}(\log N)^{1/2} \| (\cN_+ + 1)^{1/2} \xi \| \|\cK^{1/2} \xi \| .
	\end{split}
	\end{equation}
Therefore, for any $\xi \in \cF^{\leq N}_+$,   
	\[
	|\langle \xi ,  \cE_{232}^K \xi \rangle |   \leq \;  N^{1-2\a} (\log N)^{1/2}  \|\cK^{1/2}\xi\| \| (\cN_+ + 1)^{1/2} \xi \|  + N^{-\alpha} \| (\cN_+ + 1)^{1/2} \xi \|^2  \,.
\]
Combining the last bound with (\ref{eq:cEK21}), (\ref{eq:cEK22}) and \eqref{eq:cE231K}, 
we conclude that 
\begin{equation}\label{eq:cE2f} \begin{split}
| \langle \xi,   \cE_2^K  \xi \rangle | \leq   C N^{1-\a} \| (\cN_+ + 1)^{1/2} \xi \|^2 + C N^{1/2-\a}\| \cH_N^{1/2} \xi \| \|(\cN_+ + 1)^{1/2} \xi \| \,.
\end{split} \end{equation}
for any $\a>1$, $N \in \bN$ large enough, $\xi \in \cF^{\leq N}_+$. 	
	
The term $\text{G}_{22}$ in \eqref{eq:cKtermsG2} can be bounded using (\ref{eq:d-bds}). We find 
	\begin{equation}\label{eq:G22}
	\begin{split} 
	|\langle\xi,\text{G}_{22}\xi\rangle| &\leq CN^{-2\a}\|(\cN_++1)^{1/2}\xi\|^2\,.
	\end{split}  
	\end{equation}
	We split $\text{G}_{23} =  \cE_{31}^K + \cE_{32}^K + \hc$, with 
	\begin{equation*} 
	\begin{split} 
	\cE_{31}^K = \; &\int_0^1ds \sum_{p \in \L^*_+} p^2 \eta_p \big(\gamma_p^{(s)}-1\big) d_p^{(s)}b_{-p} \, , \hspace{.8cm} \cE_{32}^K =  \int_0^1ds \sum_{p \in \L_+^*} p^2 \eta_p d_p^{(s)} b_{-p} 
	\end{split}  
	\end{equation*}
	 With (\ref{eq:d-bds}), we find 
	\[ \begin{split} 
	|\langle \xi, \cE_{31}^K \xi \rangle | &\leq C \int_0^1ds\; \sum_{p \in \L^*_+} p^2 |\eta_p |^3 \| (d_p^{(s)})^* \xi \| \|  b_{-p} \xi \| ds \leq C N^{-3\a} \| (\cN_+ + 1)^{1/2} \xi \|^2 \end{split} \]
To estimate $\cE_{32}^K$, we use (\ref{eq:eta-scat0}) and we switch to position space. Proceeding as we did in (\ref{eq:star1}), (\ref{eq:cE231K}), (\ref{eq:EK232}), we obtain  
	\[\begin{split} |\langle \xi , \cE_{32}^K \xi \rangle | \leq CN\int_0^1 ds \int_{\Lambda^2} dx dy &\left[ e^{2N} V(e^N(x-y)) + N^{2\a-1} \chi (|x-y| \leq N^{-\a}) \right] \\ &\hspace{.5cm} \times \| (\cN_+ + 1)^{1/2} \xi \| \| (\cN_+ + 1)^{-1/2} \check{d}_x^{(s)} \check{b}_y \xi \|\,. \end{split} \]
	With (\ref{eq:dxy-bds}) and \eqref{eq:estimatechi} we find 
	\[\begin{split} 
|\langle \xi , \cE_{32}^K \xi \rangle | &\leq C N^{-\a}  \int_{\Lambda^2} dx dy \left[ e^{2N} V(e^N(x-y)) + N^{2\a -1} \chi(|x-y|\leq N^{-\a}) \right] \\ &\hspace{2cm} \times  \| (\cN_++1)^{1/2} \xi \|  \left[  \| \check{a}_y  (\cN_++1) \xi \|  + \| \check{a}_x \check{a}_y  (\cN_++1)^{1/2} \xi \| \right]  \\ 
	&\leq C N^{1-\a} \| (\cN_+ + 1)^{1/2} \xi \|^2 + C N^{1/2-\a} \| (\cN_+ + 1)^{1/2} \xi \| \| \cV_N^{1/2} \xi \| \\
	& \qquad \hskip 3.5cm + C N^{1-2\a} (\log N)^{1/2} \|(\cN_++1)^{1/2}\xi\|   \|\cK^{1/2}\xi\| \,.
	\end{split} \]
	Combining the bounds for $\cE_{31}^K $ and $\cE_{32}^K $ , we conclude that, if $\alpha > 1$, 
	\begin{equation}\label{eq:G23f} \begin{split}
| \langle \xi,  \text{G}_{23}  \xi \rangle | \leq  C N^{1/2-\a}   \| (\cN_+ + 1)^{1/2} \xi \| \| \cH_N^{1/2} \xi \| +CN^{1-\a}\| (\cN_+ + 1)^{1/2} \xi \|^2 
\end{split} \end{equation} 
To bound $\text{G}_{24}$ in \eqref{eq:cKtermsG2}, we use (\ref{eq:d-bds}), the bounds \eqref{eq:modetap} and $\| \eta \|^2_{H_1} \leq C N$, and the commutator \eqref{eq:comm-bp}:
\begin{equation*} 
	\begin{split} 
	|\langle &\xi,\text{G}_{24}\xi\rangle| \\ &\leq C \int_0^1 ds  \sum_{p \in \L^*_+} p^2 \eta_p^2 \|(\cN_++1)^{1/2}\xi\|\|(\cN_++1)^{-1/2} d_p^{(s)}b^*_{p}\xi\| \\
	&\leq C \|(\cN_++1)^{1/2}\xi\|  \sum_{p \in \L^*_+} p^2 \eta_p^2 \left[ |\eta_p| \| (\cN_+ + 1)^{1/2} \xi \| +N^{-1} \| \eta\| \| b_p b_p^* (\cN_+ +1)^{1/2} \xi \| \right] \\
	&\leq C N^{-\a} \| (\cN_+ + 1)^{1/2} \xi \|^2 \,.
\end{split} \end{equation*}
Together with \eqref{eq:cKtermsG2}, \eqref{eq:cE2f}, \eqref{eq:G22} and \eqref{eq:G23f}, this implies that
\[
 \text{G}_2 = - \sum_{p \in \L^*_+} p^2 \eta_p^2 \, \frac{\cN_+ + 1}{N} \frac{N-\cN_+}{N}+  \cE_4^K 
\]
with 
\be \label{eq:err4K}
| \langle \xi,   \cE_4^K \xi \rangle | \leq  C N^{1/2 -\a} \| \cH_N^{1/2}\xi\| \| (\cN_++1)^{1/2}\xi \| + C N^{1-\a}  \| (\cN_++1)^{1/2}\xi \|^2\,.
\ee
Combining (\ref{eq:cE0K}), \eqref{eq:G3-bd} and \eqref{eq:err4K}, we obtain (\ref{eq:K-dec}) and (\ref{eq:errorKc}).
\end{proof}

In the next proposition, we consider the conjugation of the operator
\[ \begin{split} \cL^{(2,V)}_{N} =\; &N\sum_{p \in \Lambda_+^*} \widehat{V} (p/e^N) \left[ b_p^* b_p - \frac{1}{N} a_p^* a_p \right] + \frac N2 \sum_{p \in \Lambda^*_+} \widehat{V} (p/e^N) \left[ b_p^* b_{-p}^* + b_p b_{-p} \right] \end{split} \]
\begin{prop}\label{prop:G2V} Under the assumptions of Prop. \ref{prop:GN}, there is a constant $C > 0$ such that 
	\begin{equation}\label{eq:cEV-defc} 
	\begin{split}
	e^{-B} \cL^{(2,V)}_N  e^{B} =\; &  N\sum_{p \in \L^*_+} \widehat{V} (p/e^N)\eta_p \Big(\frac{N-\cN_+}{N}\Big)\Big(\frac{N-\cN_+-1}{N}\Big)\\
& + N\sum_{p \in \Lambda_+^*} \widehat{V} (p/e^N) a^*_pa_p\left(1-\fra{\cN_+}{N}\right) \\ 
&+ \frac{N}{2}\sum_{p \in \Lambda^*_+} \widehat{V} (p/e^N) \big( b_p b_{-p}+ b_{-p}^* b_p^*\big) +\cE_{N}^{(V)}
	\end{split}
	\end{equation}
	where 
	\begin{equation}\label{eq:errorVc}
	\begin{split}
| \langle \xi, \cE_{N}^{(V)} \xi \rangle | \leq  C N^{1/2 -\a} \| \cH_N^{1/2}\xi\| \| (\cN_++1)^{1/2}\xi \| + C N^{1-\a}  \| (\cN_++1)^{1/2}\xi \|^2\,.
	\end{split}
	\end{equation}
for any $\a > 1$, $\xi \in \cF^{\leq N}_+$ and $N \in \bN$ large enough.
\end{prop}

\begin{proof}
We write 
\begin{equation}\label{eq:G2-deco} \begin{split} 
e^{-B} \cL^{(2,V)}_{N}  e^{B}  =\; & N\sum_{p \in \Lambda^*_+}  \widehat{V} (p/e^N) e^{-B} b_p^* b_p e^{B} - \sum_{p \in \Lambda^*_+} \widehat{V} (p/e^N) e^{-B} a_p^* a_p e^{B} \\ 
	&+ \frac{N}{2} \sum_{p \in \Lambda^*_+} \widehat{V} (p/e^N) e^{-B} \big[ b_p b_{-p} + b_p^* b_{-p}^* \big] e^{B} \\ =: \; &\text{F}_1 + \text{F}_2 +\text{F}_3\,. 
	\end{split} \end{equation} 
	With (\ref{eq:ebe}), we find 
	\begin{equation*} 
	\begin{split} 
	\text{F}_1 =\; & N\sum_{p \in  \L_+^*}  \widehat{V} (p/e^N) \big[ \gamma_p b_p^* + \sigma_p b_{-p} \big] \big[ \gamma_p b_p + \sigma_p b_{-p}^*] \\
	&+N\sum_{ p\in \L_+^*} \widehat{V} (p/e^N) \big[ (\gamma_p b_p^* + \sigma_p b_{-p})  d_p+  d_p^*(\gamma_p b_p + \sigma_p b_{-p}^*)+ d_p^* d_p\big]	
\end{split} \end{equation*}
where $\gamma_p = \cosh \eta_p$, $\s_p = \sinh \eta_p$ and the operators $d_p$ are defined in (\ref{eq:defD}). Using $|1-\gamma_p | \leq \eta_p^2$, $|\sigma_p| \leq C |\eta_p|$ and using Lemma  \ref{lm:dp} for the terms on the second line, we find 
\begin{equation}\label{eq:F1} \text{F}_1 = N\sum_{p \in \L_+^*} \widehat{V} (p/e^N) b_p^* b_p + \cE_1^V \end{equation}
with $\pm \cE_1^V \leq C N^{1-\a} (\cN_+ + 1)$. 
	
Let us now consider the second contribution on the r.h.s. of (\ref{eq:G2-deco}). We find 	
\begin{equation} \label{eq:F2}
		-\text{F}_2 = \sum_{p \in \Lambda^*_+} \widehat{V} (p/e^N)  a_p^* a_p + \cE_2^V
	\end{equation}
	with 
	\begin{equation*}
	 \cE_2^V =  \sum_{p \in \Lambda^*_+} \widehat{V} (p/e^N)\int_0^1 e^{-sB}(\eta_pb_{-p}b_p + \hc) e^{s B}ds.
	\end{equation*}
With Lemma \ref{lm:Ngrow}, we easily find $\pm \cE_2^V \leq C N^{-\a} (\cN_+ + 1)$. 

Finally, we consider the last term on the r.h.s. of \eqref{eq:G2-deco}.  With (\ref{eq:ebe}), we obtain 
	\begin{equation}\label{eq:G23split}
	\begin{split}
	\text{F}_3 =\; &\frac{N}{2} \sum_{p \in \Lambda^*_+} \widehat{V} (p/e^N) \left[ \gamma_p b_p + \sigma_p b_{-p}^* \right] \left[ \gamma_p b_{-p} + \sigma_p b_p^* \right]  +\hc \\ 
	&+ \frac{N}{2}  \sum_{p \in \Lambda^*_+} \widehat{V} (p/e^N) \, \left[ (\g_p b_p+ \s_p b^*_{-p}) \, d_{-p} + 
	d_p\, (\g_p b_{-p} + \s_p b^*_{p}) \right]  + \hc \\ 
	&+\frac{N}{2}  \sum_{p \in \Lambda^*_+} \widehat{V} (p/e^N) d_p d_{-p} + \hc \\
	=:  &\,\text{F}_{31} + \text{F}_{32}+\text{F}_{33}\,.
	\end{split} \end{equation}
Using $|1-\gamma_p| \leq C \eta_p^2$, $|\sigma_p| \leq C |\eta_p|$, we obtain 
	\begin{equation}\label{eq:fin-G23}
	\begin{split}
	\text{F}_{31}	= \; & \frac{N}{2}\sum_{p \in \Lambda^*_+} \widehat{V} (p/e^N) \big( b_p b_{-p}+ b_{-p}^* b_p^*\big) +N\sum_{p \in \L^*_+} \widehat{V} (p/e^N) \eta_p \frac{N-\cN_+}{N} + \cE^V_3 
\end{split} \end{equation}
with $\pm\cE_3^V \leq C N^{1-\a}(\cN_++1)$. 
As for $\text{F}_{32}$ in \eqref{eq:G23split}, we divide it into four parts
	\begin{equation}\label{eq:splitF32}
	\begin{split}
	\text{F}_{32} = \; & \frac{N}{2}  \sum_{p \in \Lambda^*_+} \widehat{V} (p/e^N) \, \left[ ( \g_p b_p+ \s_p b^*_{-p}) \, d_{-p} + d_p\, (\g_p b_{-p} + \s_p b^*_{p}) \right] +\hc \\ 
	=: \; & \text{F}_{321}+\text{F}_{322}+\text{F}_{323}+\text{F}_{324} \,.
	\end{split} \end{equation}
	We start with $\text{F}_{321}$, which we write as
	\begin{equation*} 
	\text{F}_{321} = - N\sum_{p \in \L^*_+} \widehat{V} (p/e^N) \eta_p \, \left(\frac{N-\cN_+}{N}\right)\left(\frac{\cN_+ +1}{N}\right) +\cE_4^V \end{equation*} 
	where $\cE_4^V = \cE_{41}^V + \cE_{42}^V + \cE_{43}^V + \hc$, with 
	\begin{equation*} 
	\begin{split}
	\cE^V_{41} = \; &\frac{N}{2}  \sum_{p \in \Lambda^*_+} \widehat{V} (p/e^N) \, (\g_p - 1) b_p d_{-p} \, , \qquad 
	\cE_{42}^V = \frac{N}{2} \sum_{p \in \L_+^*} \widehat{V} (p/e^N) b_p \lis{d}_{-p} \\
	\cE_{43}^V = \; &- \frac{N}{2} \sum_{p \in \L^*_+} \widehat{V} (p/e^N) \eta_p \frac{\cN_+ + 1}{N} (b_p^* b_p - N^{-1} a_p^* a_p ) 
	\end{split} \end{equation*}
	and with the notation $\lis{d}_{-p} = d_{-p} + N^{-1}  \eta_p \, \cN_+ b_p^*$. Since $|\g_p - 1| \leq C \eta_p^2$, $\| \eta\|_\io \leq C N^{-\a}$, we find easily with (\ref{eq:d-bds})  that 
	\[ \begin{split} |\langle \xi, \cE_{41}^V \xi \rangle |  &\leq C  N^{1-3\a}\| (\cN_+ + 1)^{1/2} \xi \|^2 \,.\end{split} \]
Moreover 
	\[ |\langle \xi , \cE_{43}^V \xi \rangle | \leq C N \sum_{p \in \L^*_+} \eta_p \| a_p \xi \|^2 \leq C N^{1-\a} \| \cN_+^{1/2} \xi \|^2\,. \]
As for $\cE_{42}^V$, we switch to position space and we use (\ref{eq:splitdbd}). We obtain 
	\[ \begin{split} 
	|\langle \xi , \cE_{42}^V \xi \rangle | &\leq CN \int_{\L^2}  dx dy \, e^{2N} V(e^N(x-y)) \| (\cN_+ + 1)^{1/2} \xi \| \| (\cN_+ + 1)^{-1/2} \check{a}_x \check{\lis{d}}_y \xi \| \\ &
\leq C N^{1-\a}   \int_{\L^2}  dx dy \, e^{2N} V(e^N(x-y)) \| (\cN_+ + 1)^{1/2} \xi \| \\ &\hspace{2cm} \times \Big[  \| (\cN_+ + 1)^{1/2} \xi \| +  \| \check{a}_x  \xi \| +\| \check{a}_y  \xi \| + N^{-1/2}\| \check{a}_x \check{a}_y \xi \| \Big] \\
	&\leq C N^{1-\a} \|(\cN_+ + 1)^{1/2} \xi \|^2 + C N^{1/2-\a} \| (\cN_+ + 1)^{1/2} \xi \| \| \cV_N^{1/2} \xi \|\,.
	\end{split} \]
	We conclude that \begin{equation*}\begin{split}
	|\langle \xi , \cE^V_4 \xi \rangle |  \leq C N^{1/2-\a} \| (\cN_+ + 1)^{1/2} \xi \| \| \cV_N^{1/2} \xi \| +   C N^{1-\a} \|(\cN_+ + 1)^{1/2} \xi \|^2.
	\end{split}\end{equation*}	
To bound the term $\text{F}_{322}$ in (\ref{eq:splitF32}), we use (\ref{eq:d-bds})  and $|\s_p|\leq C|\eta_p|$; we obtain 
	\begin{equation*} 
	\begin{split}
	|\langle\xi,\text{F}_{322} \xi\rangle| \leq \; &C N  \sum_{p \in \L^*_+} |\eta_p| \| b_{-p} \xi\| \left[ |\eta_p| \| (\cN_+ + 1)^{1/2} \xi \| + \| \eta \| \| b_{-p} \xi \| \right] \\ \leq \; &C N^{1-2\a} \| (\cN_+ + 1)^{1/2} \xi \|^2 \,. \end{split} \end{equation*}
	Let us now consider the term $\text{F}_{323}$ on the r.h.s. of (\ref{eq:splitF32}). We write  $\text{F}_{323} = \cE_{51}^V + \cE_{52}^V + \hc$, with
	\begin{equation*} 
	\begin{split}
	\cE_{51}^V = \frac{N}{2}  \sum_{p \in \Lambda^*_+} \widehat{V} (p/e^N) \, (\g_p-1) \, d_p  b_{-p} \, , \qquad \cE_{52}^V = \frac{N}{2}  \sum_{p \in \Lambda^*_+} \widehat{V} (p/e^N) \,  d_p  b_{-p} \,.
	\end{split} \end{equation*}
	With $|\g_p - 1| \leq C \eta_p^2 $ and  (\ref{eq:d-bds})  we obtain 
	\[ |\langle \xi , \cE_{51}^V \xi \rangle | \leq C N\sum_{p \in \L^*_+} \eta_p^2 \,\| d^*_p \xi \| \| a_p \xi \| \leq C N^{1-3\a} \| (\cN_+ + 1)^{1/2} \xi \|^2 \,. \]
	We find, switching to position space and using (\ref{eq:dxy-bds}),
	\[ \begin{split} &|\langle \xi, \cE_{52}^V \xi \rangle |\\
	&\leq CN \int_{\L^2}  dx dy \, e^{2N} V(e^N(x-y)) \| (\cN_+ + 1)^{1/2} \xi \| \| (\cN_+ +1)^{-1/2} \check{d}_x \check{a}_y \xi \| \\ 
	&\leq CN^{1-\a}   \| (\cN_+ + 1)^{1/2} \xi \| \int_{\L^2}  dx dy \, e^{2N} V(e^N(x-y)) \left[ \| \check{a}_y   \xi \| + N^{-1/2}\| \check{a}_x \check{a}_y  \xi \| \right] \\
&\leq C N^{1-\a} \| (\cN_+ + 1)^{1/2} \xi \|^2 + C N^{1/2-\a}\| (\cN_+ + 1)^{1/2} \xi \| \| \cV_N^{1/2} \xi \|\,. 
	\end{split} \]
	Hence, \[ | \langle \xi, \text{F}_{323} \xi \rangle| \leq C N^{1-\a} \| (\cN_+ + 1)^{1/2} \xi \|^2 + C N^{1/2-\a}\| (\cN_+ + 1)^{1/2} \xi \| \| \cV_N^{1/2} \xi \| \] 	
	To estimate the term $\text{F}_{324}$ in \eqref{eq:splitF32} we use (\ref{eq:d-bds}) and  the bound 
	\[ \begin{split} \label{eq:sumVetap}
\sum_{ p\in  \L_+^*}|\widehat{V}(p/e^N)||\eta_p| & \leq C \sum_{ \substack{p\in  \L_+^*,\, |p|\leq e^N}} \frac {1}{p^2} + C\sum_{ \substack{p\in  \L_+^*,\, |p|> e^N}} \frac{|\widehat{V}(p/e^N)|}{p^2} \\
& \leq C N  + C \Bigg( \sum_{ \substack{p\in  \L_+^*}} |\widehat{V}(p/e^N)|^2 \Bigg)^{1/2} \Bigg( \sum_{ \substack{p\in  \L_+^*,\, |p|> e^N}} \frac{1}{p^4} \Bigg)^{1/2}  \\
&  \leq C N 
\end{split}\]
We find 
	\begin{equation}\nonumber
	\begin{split}
	|\langle\xi, \text{F}_{324} \xi \rangle|  &\leq \; CN\sum_{p \in \L^*_+} \big|\widehat{V} (p/e^N) \big||\eta_p|\|(\cN_++1)^{1/2}\xi\|\|(\cN_++1)^{-1/2} d_p\, b^*_{p}\xi\|
	\\ 
	&\leq CN\sum_{p \in \L^*_+} \big|\widehat{V} (p/e^N) \big||\eta_p|\|(\cN_++1)^{1/2}\xi \|  \\ &\hspace{1cm} \times 
	\left[ |\eta_p| \| (\cN_+ + 1)^{1/2} \xi \| + N^{-1}\| \eta \| \| b_p b^*_p (\cN_+ + 1)^{1/2} \xi \| \right] \\ 
	&\leq CN\sum_{p \in \L^*_+} \big|\widehat{V} (p/e^N) \big||\eta_p|\|(\cN_++1)^{1/2}\xi \| \\
	&\hspace{1cm} \times  
	\left[ |\eta_p| \| (\cN_+ + 1)^{1/2} \xi \| + N^{-1}\| \eta\| \| (\cN_+ + 1)^{1/2} \xi \| + \| \eta \| \| a_p \xi \| \right] \\ 
&\leq \; C N^{1-\a}\|(\cN_++1)^{1/2}\xi\|^{2}\,. 
	\end{split} \end{equation}
Combining the last bounds, we arrive at   
	\begin{equation*}\label{eq:F32-fin} \text{F}_{32} = N\sum_{p \in \L^*_+} \widehat{V} (p/e^N) \eta_p \left( \frac{N-\cN_+}{N} \right) \left( \frac{-\cN_+ -1}{N} \right) + \cE_6^V \end{equation*} 
	with 
	\begin{equation}\label{eq:F32} \begin{split}
| \langle \xi, \cE^V_6  \xi \rangle| \leq C N^{1-\a} \| (\cN_+ + 1)^{1/2} \xi \|^2 + C N^{1/2-\a}\| (\cN_+ + 1)^{1/2} \xi \| \| \cV_N^{1/2} \xi \|\,.
	\end{split}
	 \end{equation}
	
	To control the last contribution $\text{F}_{33}$ in \eqref{eq:G23split}, we switch to position space. With (\ref{eq:ddxy}) and (\ref{eq:etax}) we obtain 
	\[ \begin{split}  |\langle \xi , \text{F}_{33} \xi \rangle | &\leq CN \| (\cN_+ + 1)^{1/2} \xi \| \int_{\L^2}  dx dy \, e^{2N} V(e^N(x-y)) \| (\cN_+ + 1)^{-1/2} \check{d}_x \check{d}_y \xi \| \\
	&\leq C N^{1-\a} \| (\cN_+ + 1)^{1/2} \xi \|^2 + C N^{1/2-2\a} \| (\cN_++1)^{1/2} \xi \| \| \cV_N^{1/2} \xi \|\,. \end{split} \]    
	The last equation, combined with (\ref{eq:G23split}), (\ref{eq:fin-G23}) and 
	(\ref{eq:F32}), implies that 
	\[ \begin{split} 
	\text{F}_3 = \; &\frac{N}{2} \sum_{p \in \L_+^*} \widehat{V} (p/e^N) (b_p b_{-p} + b^*_{-p} b^*_p) \\ &+ N\sum_{p \in \L^*_+} \widehat{V} (p/e^N) \eta_p \left( \frac{N-\cN_+}{N} \right) \left( \frac{N-\cN_+ - 1}{N} \right) + \cE_7^V \end{split} \]
	with 
	\begin{equation*}\begin{split}
| \langle \xi, \cE^V_7  \xi \rangle| \leq C N^{1-\a} \| (\cN_+ + 1)^{1/2} \xi \|^2 + C N^{1/2-\a}\| (\cN_+ + 1)^{1/2} \xi \| \| \cV_N^{1/2} \xi \|\,.
	\end{split}\end{equation*}
	
	Together with (\ref{eq:F1}) and with (\ref{eq:F2}),  and recalling that $b_p^* b_p -N^{-1} a_p^* a_p = a_p^* a_p (1- \cN_+ / N)$, we obtain (\ref{eq:cEV-defc}) with (\ref{eq:errorVc}). 
		\end{proof}

	\subsection{Analysis of $ \cG_{N,\a}^{(3)}=e^{-B}\cL^{(3)}_N e^{B}$}

We consider here the conjugation of the cubic term $\cL_N^{(3)}$, defined in (\ref{eq:cLNj}). 
\begin{prop}\label{prop:GN-3}  Under the assumptions of Prop. \ref{prop:GN}, there exists a constant $C > 0$ such that 
		\begin{equation*}\label{eq:def-E3}
		\cG_{N,\a}^{(3)}  = e^{-B}\cL^{(3)}_N e^{B}  = \sqrt{N} \sum_{p,q \in \L^*_+ : p + q \not = 0} \widehat{V} (p/e^N) \left[ b_{p+q}^* a_{-p}^* a_q  + \hc \right] + \cE^{(3)}_{N} \end{equation*}
		where  
		\begin{equation}\label{eq:lm-GN31}\begin{split} 
| \langle \xi, \cE_{N}^{(3)}  \xi \rangle| \leq C N^{1/2-\a}\| (\cN_+ + 1)^{1/2} \xi \| \| \cV_N^{1/2} \xi \|+ C N^{1-\a} \| (\cN_+ + 1)^{1/2} \xi \|^2 
\end{split}\end{equation}
for any $\a > 1$ and $N \in \bN$ large enough.
\end{prop} 
\begin{proof} This proof is similar to the proof of \cite[Prop. 7.5]{BBCS3}. Expanding $e^{-B} a_{-p}^* a_q e^B$, we arrive at 
\begin{equation} \label{eq:deco-cE3} \begin{split} 
		\cE^{(3)}_{N} &=  \sqrt{N} \sum_{p,q \in \Lambda^*_+ : p+q \not = 0} \widehat{V} (p/e^N) \big((\g_{p+q}-1) b^*_{p+q} + \s_{p+q} b_{-p-q} + d_{p+q}^* \big) \, a_{-p}^* a_q  \\
		&\hspace{.3cm}+ \sqrt{N} \sum_{p,q \in \Lambda_+^* , p+q \not = 0} \widehat{V} (p/e^N) \eta_p \, e^{-B}b^*_{p+q}e^{B} \int_0^1 ds\, e^{-sB} b_{p} b_{q} e^{sB}\\
		&\hspace{.3cm}+ \sqrt{N} \sum_{p,q \in \Lambda_+^* , p+q \not = 0} \widehat{V} (p/e^N) \eta_q \, e^{-B} b^*_{p+q}e^{B} \int_0^1 ds\,  e^{-sB}b_{-p}^*b^*_{-q} e^{sB}  \\
		&\hspace{.3cm}+ \hc \\ &=:  \; \cE^{(3)}_1 + \cE_2^{(3)} + \cE_3^{(3)} + \hc 
		\end{split} \end{equation}
where, as usual, $\gamma_p = \cosh \eta (p)$, $\s_p = \sinh \eta (p)$ and $d_p$ is as in (\ref{eq:defD}).  We consider $\cE_1^{(3)}$. To this end, we write 
		\[ \begin{split}
		\cE^{(3)}_1 =\; &  \sqrt{N} \sum_{p,q \in \Lambda^*_+ : p+q \not = 0} \widehat{V} (p/e^N) \big((\g_{p+q}-1) b^*_{p+q} + \s_{p+q} b_{-p-q} + d_{p+q}^* \big) \, a_{-p}^* a_q \\
		=: & \; \cE^{(3)}_{11} + \cE^{(3)}_{12} +\cE^{(3)}_{13} \,.
		\end{split}\]
		Since $|\g_{p+q}-1|\leq |\eta_{p+q}|^2$ and $\| \eta \| \leq C N^{-\a}$, we find 
		\be \begin{split} \label{eq:GN3-P11}
			|\langle \xi, \cE^{(3)}_{11} \xi \rangle| 
			& \leq  CN \| \eta\|^2 \| (\cN_++1)^{1/2} \xi \|^2 \leq C N^{1-2\a} \| (\cN_+ + 1)^{1/2} \xi \|^2 \,.
		\end{split}\ee
As for  $\cE^{(3)}_{12}$, we commute $a^*_{-p}$ through $b_{-p-q}$ (recall $q \not = 0$). With $|\s_{p+q}| \leq C |\eta_{p+q}|$, we obtain 
		\be \begin{split} \label{eq:GN3-P12}
			|\langle \xi, \cE^{(3)}_{12} \xi \rangle| & \leq  C N^{1-\a} \| (\cN_++1)^{1/2} \xi \|^2 \,.
		\end{split}\ee
We decompose now $\cE^{(3)}_{13} = \cE^{(3)}_{131} + \cE^{(3)}_{132}$, with 
		\[ \begin{split} 
		\cE^{(3)}_{131} =\; & \sqrt{N} \sum_{p,q \in \Lambda^*_+ : p+q \not = 0} \widehat{V} (p/e^N) \, \bar d^*_{p+q}   a^*_{-p} a_q \\
		\cE^{(3)}_{132} = \; & -\frac{(\cN_++1)}{N}\sqrt{N} \sum_{p,q \in \Lambda^*_+ : p+q \not = 0} \widehat{V} (p/e^N) \eta_{p+q} \, b_{-p-q} a^*_{-p} a_q \,.
\end{split}\]
where we defined $d^*_{p+q}= \lis{d}^*_{p+q} - \frac{(\cN_++1)}{N}\, \eta_{p+q} b_{-p-q}$.
The term $\cE^{(3)}_{132}$ is estimated similarly to $\cE_{12}^{(3)}$, moving $a_{-p}^*$ to the left of $b_{-p-q}$; we find $\pm \cE^{(3)}_{132} \leq C N^{1-\a} (\cN_+ + 1)$. We bound $\cE^{(3)}_{131}$ in  position space. We find
		\[\begin{split} 
		&|\langle \xi, \cE^{(3)}_{131} \xi \rangle|\\
		&\quad \leq 
		N^{1/2}\int_{\L^2}  dx dy \,e^{2N}V(e^N(x-y))  \| \check{a}_x \xi \| \|\check{a}_y \check{\bar d}_x  \xi \|  \\
		&\quad \leq  CN^{1/2-\a}  \int_{\L^2}  dx dy \,e^{2N} V(e^N(x-y))\| \check{a}_x \xi \| \\ & \hspace{0.8cm} \times  \big[ \,\|(\cN_++1)\xi \| + N^{-1} \| \check{a}_x (\cN_+ + 1)^{1/2} \xi \| +\| \eta \| \| \check{a}_y (\cN_+ + 1)^{1/2} \xi \| + \| \check{a}_x \check{a}_y \xi \| \big]  \\
		&\quad \leq C N^{1-\a} \| (\cN_+ + 1)^{1/2} \xi \|^2+CN^{1/2-\a}\| (\cN_+ + 1)^{1/2} \xi \| \| \cV_N^{1/2} \xi \|  \,.
		\end{split}\]
		With \eqref{eq:GN3-P11} and \eqref{eq:GN3-P12} we obtain 
		\be \label{eq:G3N-P1end}
		\begin{split}
 |\langle \xi , \cE^{(3)}_1\xi \rangle | \leq &  C N^{1/2-\a}   \| \cV_N^{1/2}\xi \| \| (\cN_+ +1)^{1/2} \xi \|+ C N^{1-\a}  \| (\cN_+ +1)^{1/2} \xi \|^2 \,.
		\end{split}
		\ee
		
Next, we focus on $\cE^{(3)}_2$, defined in (\ref{eq:deco-cE3}). With Eq. \eqref{eq:ebe}, we find  
		\begin{equation}\label{eq:cE32-deco} \begin{split}
		\cE^{(3)}_2 =\; & \sqrt{N}\sum_{p,q \in \Lambda_+^* , p+q \not = 0} \widehat{V} (p/e^N) \eta_p \, e^{-B}b^*_{p+q}e^{B}\\
		& \hspace{0.8cm} \times \int_0^1 ds\, \big( \g^{(s)}_p  \g^{(s)}_{q} b_{p} b_q + \s^{(s)}_p  \s^{(s)}_{q}b^*_{-p} b^*_{-q} +\g^{(s)}_p   \s^{(s)}_{q} b^*_{-q} b_{p} +  \s^{(s)}_p  \g^{(s)}_{q}  b^*_{-p} b_{q}  \big)  \\
		& +\sqrt{N} \sum_{p,q \in \Lambda_+^* , p+q \not = 0} \widehat{V} (p/e^N) \eta_p \, e^{-B}b^*_{p+q}e^{B} \int_0^1 ds\,  \g^{(s)}_p  \s^{(s)}_{q} [b_{p},  b^*_{-q}]  \\
		& + \sqrt{N} \sum_{p,q \in \Lambda_+^* , p+q \not = 0} \widehat{V} (p/e^N)  \eta_p \, e^{-B}b^*_{p+q}e^{B}\\
		& \hskip 1cm\times \int_0^1 ds\, \Big[ d^{(s)}_p \big( \g^{(s)}_{q} b_{q} + \s^{(s)}_{q} b^*_{-q} \big) + \big( \g^{(s)}_p b_{p} + \s^{(s)}_p b^*_{-p} \big)   d^{(s)}_{q} +   d^{(s)}_{p}  d^{(s)}_{q}   \Big]\\
		=: &\; \cE^{(3)}_{21} + \cE^{(3)}_{22} + \cE^{(3)}_{23}
		\end{split}\end{equation}
		with $\gamma^{(s)}_p = \cosh (s \eta_p)$, $\s^{(s)}_p = \sinh (s \eta_p)$ and $d^{(s)}_p$ defined as in (\ref{eq:defD}), with $\eta$ replaced by $s \eta$. With Lemma \ref{lm:Ngrow}, we get 
		\be \begin{split} \label{eq:G3N-P31} 
			|\langle \xi , \cE^{(3)}_{21} \xi \rangle | \leq \; &C  N^{1-\a} \| (\cN_++1)^{1/2}\xi \|^2 \,.
		\end{split}\ee
		Since $[b_{p},b^*_{-q}] = - a^*_{-q} a_{p} /N $ for $p \neq -q$, we find 
		\be \begin{split}\label{eq:G3N-P32} 
			|\langle \xi , \cE^{(3)}_{22} \xi \rangle | &\leq CN^{-2\a}  \| (\cN_+ + 1)^{1/2} \xi \|^2  \,.
		\end{split}\ee
As for the third term on the r.h.s. of (\ref{eq:cE32-deco}), we switch to position space. We find  
		\[ \begin{split}
		\cE^{(3)}_{23} =\;& \sqrt N  \int_{\L^3} dx dy dz\, e^{2N} V(e^N(x-z))  \check{\eta} (y-z) \,e^{-B} \check{b}^*_x e^{B}\\
		& \hskip 1cm\times \int_0^1 ds\, \Big[  \check{d}^{(s)}_y \big( b(\check{\g}^{(s)}_{x}) +b^*(\check{\s}^{(s)}_{x}) \big) +  \big( b(\check{ \g}^{(s)}_{y}) +b^*(\check{\s}^{(s)}_{y}) \big) \check{ d}^{(s)}_x    +  \check{ d}^{(s)}_y \check{ d}^{(s)}_x  \Big]\,.
		\end{split}\]
		Using the bounds \eqref{eq:dxy-bds}, \eqref{eq:splitdbd}, \eqref{eq:ddxy} and Lemma \ref{lm:Ngrow} we arrive at 
		\begin{equation*}\begin{split} 
		 &|\langle \xi , \cE^{(3)}_{23} \xi \rangle | \\
		 & \leq   C \sqrt N  \int_{\Lambda^3} dx dy dz \, e^{2N} V(e^N(x-z)) | \check{\eta} (y-z)| \| \check{b}_x e^{B}  \xi \| \int_0^1 ds  \\
&  \; \times  \Big[ \| \check{d}^{(s)}_y \big( \check{b}_{x} + b(\check{r}^{(s)}_x) + b^*(\check{\s}^{(s)}_x)\big)\xi \|  +\| \big( \check{b}_{y} + b(\check{r}^{(s)}_y) + b^*(\check{\s}^{(s)}_y)\big) \check{d}^{(s)}_x  \xi \| + \|  \check{d}^{(s)}_x  \check{d}^{(s)}_y \xi \|   \Big] \\
& \leq C \sqrt N  \int_{\Lambda^3} dx dy dz \, e^{2N} V(e^N(x-z)) | \check{\eta} (y-z)| \| \check{b}_x e^{B}  \xi \| \Big[ N^{-1} |\check{\eta} (x-y)|  \| (\cN_+ +1) \xi\|  \\ 
& \quad + \| \eta \| \| \check{b}_{x} \check{b}_{y} \xi \| + \| \eta \| \| (\cN_++1) \xi \| + \| \eta \| \| \check{b}_x (\cN_++1)^{1/2} \xi \| + \| \eta \|\| \check{b}_y (\cN_++1)^{1/2} \xi \|  \Big] \\
		& \leq  C N^{1-\a} \|  \cN_+^{1/2} e^{B} \xi \| \| (\cN_++ 1) \xi \| \\
		& \leq  C N^{1-\a} \|  (\cN_++1)^{1/2} \xi \|^2 \,
		\end{split} \end{equation*}
where $\check{r}$ indicates the function in $L^2 (\L)$ with Fourier coefficients $r_p =1-\g_p$, and the fact that $ \| \check{\eta}\|, \|\check{r}\|, \|\check{\s}\|\leq C N^{-\a}$.
		Combined with (\ref{eq:G3N-P31}) and (\ref{eq:G3N-P32}), the last bound implies that 
		\begin{equation}\label{eq:cE32f}\begin{split} 
\pm \cE_2^{(3)} \leq & C N^{1-\a}(\cN_+ + 1) \,.
\end{split}\end{equation}
To bound the last contribution on the r.h.s. of (\ref{eq:deco-cE3}), it is convenient to bound 
(in absolute value) the expectation of its adjoint  
		\begin{equation*} \begin{split} 
		\cE^{(3)*}_3 =\; & \sqrt{N} \sum_{p,q \in \Lambda_+^* , p+q \not = 0} \widehat{V} (p/e^N)  \eta_q \int_0^1 ds\, e^{-sB} b_{-q}e^{sB}\\
		& \hskip 1cm \times  \big( \g^{(s)}_p b_{-p} + \s^{(s)}_p b^*_{p} +  d^{(s)}_{-p}\big) \big( \g_{p+q} b_{p+q} + \s_{p+q} b^*_{-p-q} +  d_{p+q}\big) \\
		=\; &\sqrt{N} \sum_{p,q \in \Lambda_+^* , p+q \not = 0} \widehat{V} (p/e^N)  \eta_q \, \int_0^1 ds\,e^{-sB} b_{-q}e^{sB}\\
		&  \hskip -0.2cm\times \Big[ \, \g^{(s)}_p  \g_{p+q} b_{-p} b_{p+q} + \s^{(s)}_p  \s_{p+q} b^*_{p} b^*_{-p-q}  + \g^{(s)}_p  \s_{p+q} b^*_{-p-q} b_{-p}+  \g_{p+q}  \s^{(s)}_p b^*_{p} b_{p+q} \\
		& \hskip 1.4cm +  d^{(s)}_{-p} \big( \g_{p+q} b_{p+q} + \s_{p+q} b^*_{-p-q}\big) +  \big( \g^{(s)}_p b_{-p} + \s^{(s)}_p b^*_{p}\big)  d_{p+q} +   d^{(s)}_{-p}  d_{p+q}\Big]  \\
		& +\sqrt{N} \sum_{p,q \in \Lambda_+^* , p+q \not = 0} \widehat{V} (p/e^N)  \eta_q \, \int_0^1 ds\,e^{-sB} b_{-q}e^{sB}   \g^{(s)}_p  \s_{p+q} [b_{-p},b^*_{-p-q}] 
		\\ =: & \, \cE_{31}^{(3)} + \cE_{32}^{(3)} \,.
		\end{split} \end{equation*}
Since $q \neq 0$, $[b_{-p},b^*_{-p-q}] = - a^*_{-p-q} a_{-p} /N $. Thus, we can estimate 
		\begin{equation}\label{eq:cE332f}\begin{split}
		| \langle \xi,  &\cE_{32}^{(3)}  \xi\rangle| \\ & \leq  CN^{-1/2}  \int_0^1 ds \sum_{p,q \in \Lambda_+^* , p+q \not = 0} | \eta_q | | \eta_{p+q}| \, \| a_{-p-q}\,e^{-sB} b^*_{-q}e^{sB} \xi \| \|  a_{-p} \xi \| \\
		& \leq C  \| \eta\|^2  \|(\cN_++1)^{1/2} \xi \|^2 \leq CN^{-2\a} \| (\cN_+ + 1)^{1/2} \xi \|^2 \,.
		\end{split}\end{equation}
To bound the expectation of $\cE^{(3)}_{31}$, we switch to position space. We find  
\begin{equation*} \begin{split}  \label{eq:GN3-P2} 
& |\langle \xi ,  \cE^{(3)}_{31} \xi \rangle | \\
& \leq   N^{1/2}\int_0^1  ds\,  \int_{\Lambda^2} dx dy \, e^{2N} V(e^N(x-y))\|  b^*(\check{\eta}_{x}) e^{sB}  \xi \| \Big[ \| \check{b}_{x}\check{b}_{y} \xi \| \\ &    + \| \eta\|  \| \check{b}_x (\cN_++1)^{1/2} \xi \| +  \| \eta\| \| \check{b}_y (\cN_++1)^{1/2} \xi \| + N^{-1}|\check \eta(x-y)|   \| (\cN_++1) \xi \| \Big] \,.
		\end{split} \end{equation*}
With Lemma \ref{lm:Ngrow}, we conclude that 
\begin{equation} \begin{split}  \label{eq:GN3-cE313} 
		 |\langle \xi , \cE_{31}^{(3)} \xi \rangle | \leq &  C N^{1/2-\a}   \| \cV_N^{1/2}\xi \| \| (\cN_+ +1)^{1/2} \xi \|+ C N^{1-\a}  \| (\cN_+ +1)^{1/2} \xi \|^2 \,.
\end{split} \end{equation}
From (\ref{eq:cE332f}) and \eqref{eq:GN3-cE313} we obtain 
		\[ \begin{split}
	 |\langle \xi , \cE^{(3)}_3\xi \rangle | \leq &  C N^{1/2-\a}   \| \cV_N^{1/2}\xi \| \| (\cN_+ +1)^{1/2} \xi \|+ C N^{1-\a}  \| (\cN_+ +1)^{1/2} \xi \|^2 \, . 
\end{split}\]
Together with (\ref{eq:deco-cE3}), (\ref{eq:G3N-P1end}) and (\ref{eq:cE32f}), we arrive at  (\ref{eq:lm-GN31}).
	\end{proof}

	\subsection{Analysis of $ \cG_{N,\a}^{(4)}=e^{-B}\cL^{(4)}_N e^{B}$}
	\label{sub:G4}
	
Finally, we consider the conjugation of the quartic term $\cL_N^{(4)}$. We define the error operator 
$\cE_N^{(4)}$ through 
 \begin{equation*} \begin{split} 
	\cG_{N,\a}^{(4)} = e^{-B} \cL^{(4)}_{N} e^{B}  = \; &\cV_N + \frac 12 \sum_{\substack{q \in \Lambda^*_+, r\in \Lambda^* \\  r\neq -q}} \widehat{V} (r/e^N) \eta_{q+r} \eta_q \left( 1-\frac{\cN_+ }{N} \right) \left( 1 - \frac{\cN_+ +1}{N} \right)  \\ &+  \frac 12\sum_{\substack{q \in \Lambda^*_+, r \in \L^*: \\ r\neq -q}} \widehat{V} (r/e^N) \, \eta_{q+r} \left(  b_q b_{-q} + b^*_q b^*_{-q} \right)  + \cE^{(4)}_{N} \end{split} \end{equation*}
	\begin{prop}\label{prop:GN-4} Under the assumptions of Prop.\ref{prop:GN} there exists a constant $C > 0$ such that 
		\begin{equation}\label{eq:E4bound2}
		\begin{split}
|\langle \xi , \cE^{(4)}_N\xi \rangle | \leq &  C N^{1/2-\a}   \| \cV_N^{1/2}\xi \| \| (\cN_+ +1)^{1/2} \xi \|+ C N^{1-\a}  \| (\cN_+ +1)^{1/2} \xi \|^2 
		\end{split}
		\end{equation}
for any $\a > 1$, $\xi \in \cF^{\leq N}_+$ and $N \in \bN$ large enough.  
	\end{prop}
	
To show Prop. \ref{prop:GN-4}, we use the following lemma, whose proof can be  obtained as in \cite[Lemma 7.7]{BBCS3}.  
	\begin{lemma}\label{lm:prel4}
		Let $\eta \in \ell^2 (\L^*)$ as defined in \eqref{eq:defeta}. Then there 
		exists a constant $C > 0$ such that 
		\[ \begin{split} \label{eq:prel4-2}
			\| (\cN_++1)^{n/2} e^{-B} &\check{b}_x \check{b}_y  e^{B}\xi \| \\
			&\leq  C  \Big [\, N \| (\cN_+ +1)^{n/2} \xi \| + \| \check{a}_y (\cN_+ +1)^{(n+1)/2} \xi \|  \\
			&\hskip 1cm + \| \check{a}_x (\cN_+ +1)^{(n+1)/2} \xi \|   +  \| \check{a}_x \check{a}_y (\cN_++1)^{n/2} \xi \| \Big ]
		\end{split}\]
		for all $\xi \in \cF_+^{\leq N}$, $n \in \bZ$. \end{lemma}
	
\begin{proof}[Proof of Prop. \ref{prop:GN-4}] We follow the proof of \cite[Prop. 7.6]{BBCS3}.
We write 
\[  \cG^{(4)}_{N,\a}  = \cV_N + W_1 + W_2 + W_3 + W_4 \]
with 
\begin{equation}\label{eq:defW} 
		\begin{split} 
		\text{W}_1 = \; & \frac{1}{2} \sum_{q \in \Lambda_+^* , r \in \Lambda^* : r \not = -q} \widehat{V} (r/e^N) \eta_{q+r} \int_0^1 ds \,  \big( e^{-s B} b_{q} b_{-q}\,e^{sB}   + \text{h.c.} \big) \\ 
		\text{W}_2 = \; &  \sum_{p,q \in \Lambda_+^* , r \in \Lambda^* : r \not = p,-q} \widehat{V} (r/e^N)\,  \eta_{q+r} \int_0^1 ds \, 
		 \big(  e^{-s B} b^*_{p+r} b^*_{q} e^{s B}  a^*_{-q-r} a_p + \text{h.c.} \big)\\
		\text{W}_3 = \; &    \sum_{p,q\in \Lambda^*_+, r \in \Lambda^* : r \not = -p -q} \widehat{V} (r/e^N) \eta_{q+r} \eta_p \, \\ &\hspace{1.5cm} \times  \int_0^1 ds\,\int_0^s d\t \, \big(e^{-sB} b^*_{p+r} b^*_q e^{sB} e^{-\t B} b^*_{-p} b^*_{-q-r} e^{\t B}+ \hc \big) \\ 
		\text{W}_4 = \; &  \sum_{p,q\in \Lambda^*_+, r \in \Lambda^* : r \not = -p -q} \widehat{V} (r/e^N) \, \eta^2_{q+r} \\ &\hspace{1.5cm} \times \int_0^1 ds\,\int_0^s d\t \,  \big( e^{-sB} b^*_{p+r} b^*_q e^{sB} e^{-\t B} b_{p} b_{q+r} e^{\t B} + \hc \big) \,.
		\end{split} \end{equation}
		
Let us first consider the term $\text{W}_1$. With (\ref{eq:ebe}), we find
\begin{equation}\label{eq:W1-dec} 
\begin{split} 
\text{W}_1 = \; & \frac{1}{2} \sum_{q \in \Lambda_+^* , r \in \Lambda^* : r \not = -q} \widehat{V} (r/e^N) \eta_{q+r} \int_0^1 ds (\g^{(s)}_q)^2  (b_q b_{-q} + \hc ) \\ &+ \frac{1}{2} \sum_{q \in \Lambda_+^* , r \in \Lambda^* : r \not = -q} \widehat{V} (r/e^N) \eta_{q+r} \int_0^1 ds \, \g_q^{(s)} \s^{(s)}_q \big( [b_q , b_q^*] + \hc \big) \\
		&+ \frac{1}{2} \sum_{q \in \Lambda_+^* , r \in \Lambda^* : r \not = -q} \widehat{V} (r/e^N) \eta_{q+r} \int_0^1 ds  \, \g_q^{(s)} \big( b_q d_{-q}^{(s)}+\hc \big) + \cE_{10}^{(4)} \\
		=: & \; \text{W}_{11} + \text{W}_{12} + \text{W}_{13} + \cE^{(4)}_{10}  
		\end{split} 
		\end{equation}
		where  
		\begin{equation}\label{eq:cE410}
		\cE_{10}^{(4)} =  \; \cE_{101}^{(4)} +  \cE_{102}^{(4)} +  \cE_{103}^{(4)} +  \cE_{104}^{(4)} +  \cE_{105}^{(4)} \end{equation} 
		with 
		\begin{equation}\label{eq:cE410j} \begin{split} \cE^{(4)}_{101} &= \frac{1}{2} \sum_{q \in \Lambda_+^* , r \in \Lambda^* : r \not = -q} \widehat{V} (r/e^N) \eta_{q+r} \int_0^1 ds \Big[ 2 \gamma_q^{(s)} \s_q^{(s)} b_q^* b_q  + (\s_q^{(s)})^2 b_{-q}^* b_q^* +\hc \Big] \\
		\cE^{(4)}_{102} &= \frac{1}{2} \sum_{q \in \Lambda_+^* , r \in \Lambda^* : r \not = -q} \widehat{V} (r/e^N) \eta_{q+r} \int_0^1 ds \, \sigma_q^{(s)} \big( b_{-q}^* d_{-q}^{(s)} + \hc \big)  \\
		\cE^{(4)}_{103} &= \frac{1}{2} \sum_{q \in \Lambda_+^* , r \in \Lambda^* : r \not = -q} \widehat{V} (r/e^N) \eta_{q+r} \int_0^1 ds \, \sigma_q^{(s)} \big( d^{(s)}_q b_q^* + \hc \big) \\
		\cE^{(4)}_{104} &= \frac{1}{2} \sum_{q \in \Lambda_+^* , r \in \Lambda^* : r \not = -q} \widehat{V} (r/e^N) \eta_{q+r} \int_0^1 ds \, \g_q^{(s)} \big( d^{(s)}_q b_{-q} + \hc \big) \\
		\cE^{(4)}_{105} &= \frac{1}{2} \sum_{q \in \Lambda_+^* , r \in \Lambda^* : r \not = -q} \widehat{V} (r/e^N) \eta_{q+r} \int_0^1 ds \big( d^{(s)}_q d^{(s)}_{-q} + \hc \big) \,.
		\end{split} \end{equation}
With 
\begin{equation}\label{eq:VetaN} 
\frac 1N\sup_{q \in \Lambda_+^*} \sum_{r \in \L_+^*} |\widehat{V} (r/e^N)| |\eta_{q+r}|  \leq C \, < \infty \end{equation}
uniformly in $N \in \bN$, we can estimate the first term in (\ref{eq:cE410j}) by 
\[ \begin{split}|\langle \xi, \cE^{(4)}_{101} \xi \rangle | & \leq C N^{1-\a}\| (\cN_+ + 1)^{1/2} \xi \|^2\,. \end{split} \]
		Using (\ref{eq:VetaN}) and \eqref{eq:d-bds} we also find 
		\[ \begin{split} |\langle \xi, \cE_{102}^{(4)} \xi \rangle | &\leq C N^{1-2\a} \| (\cN_+ + 1)^{1/2} \xi \|^2 \,.\end{split} \]
For the third term in (\ref{eq:cE410j}) we switch to position space and use \eqref{eq:dxy-bds}:
\[ \begin{split} 
		|\langle \xi,  \cE_{103}^{(4)} \xi \rangle | 
		&\leq \frac 12 \int dx dy e^{2N} V(e^N (x-y)) |\check \eta(x-y)|  \\
		& \hskip 0.5cm \times\int_0^1 ds \, \| (\cN +1)^{-1/2} \check{d}_y b^*(\check{\s}_x^{(s)})\xi\| \| (\cN+1)^{1/2}\xi\| \\
		&  \leq C \| \check \eta \|_\io \| \eta\| \int dx dy e^{2N} V(e^N (x-y))  \| (\cN_++1)^{1/2}\xi\|\int_0^1ds  \\
		& \hskip 0.5cm \times \Big[ \| b^*(\check{\s}^{(s)}_x) \xi\| + \frac 1 N |\check{\eta}^{(s)}(x-y)| \|(\cN+1)^{1/2}\xi\| + \frac{1}{\sqrt N} \| b^*(\check{\s}^{(s)}_x) \check{b}_y \xi\| \Big]  \\
			&\leq C N^{1-\a}\| (\cN_++1)^{1/2} \xi \|^2 \,. \end{split} \]
		Consider now the fourth term in (\ref{eq:cE410j}). We write $\cE_{104}^{(4)} = \cE_{1041}^{(4)} + \cE_{1042}^{(4)}$, with
		\[ \begin{split} \cE_{1041}^{(4)} &= \frac{1}{2} \sum_{q \in \Lambda_+^* , r \in \Lambda^* : r \not = -q} \widehat{V} (r/e^N) \eta_{q+r} \int_0^1 ds \, (\g_q^{(s)} -1) d^{(s)}_q b_{-q} \\
		\cE_{1042}^{(4)} &= \frac{1}{2} \sum_{q \in \Lambda_+^* , r \in \Lambda^* : r \not = -q} \widehat{V} (r/e^N) \eta_{q+r} \int_0^1 ds  \, d^{(s)}_q b_{-q} 
		\end{split} \]
		With $|\g_q^{(s)} - 1| \leq C |\eta_q|^2$, (\ref{eq:VetaN}) and $\|d^*_q \xi \| \leq C \| \eta \| \|(\cN_++1)^{1/2}\xi \|$, we  find 
		\[|  \langle \xi , \cE_{1041}^{(4)} \xi \rangle|  \leq C N^{1-3\a}\| (\cN_+ + 1)^{1/2} \xi \|^2 \]  
		As for $\cE_{1042}^{(4)}$, we switch to position space. Using (\ref{eq:etax}) and (\ref{eq:dxy-bds}), we obtain 
		\begin{equation*} 
		\begin{split}
		&| \langle \xi , \cE_{1042}^{(4)} \xi \rangle | \\
		&\quad= \Big|\frac{1}{2} \int_0^1 ds \int_{\Lambda^2} dx dy \,e^{2N} V(e^N(x-y)) \check{\eta} (x-y) \langle \xi, \check{d}^{(s)}_x \check{b}_y \xi \rangle \Big| \\
		&\quad\leq CN \int_0^1 ds \int_{\Lambda^2} dx dy \,e^{2N} V(e^N(x-y)) \| (\cN_+ + 1)^{1/2} \xi \| \| (\cN_+ + 1)^{-1/2} \check{d}^{(s)}_x \check{b}_y \xi \| \\ 
		&\quad\leq C N\| \eta \|  \int_0^1 ds \int_{\Lambda^2} dx dy \,e^{2N} V(e^N(x-y)) \| (\cN_+ + 1)^{1/2} \xi \|  \\ &\hspace{7cm} \times N^{-1}\left[ \| \check{a}_y \cN_+ \xi \| + \| \check{a}_x \check{a}_y \cN_+^{1/2} \xi \| \right] \\ 
		&\quad\leq C N^{1-\a} \| (\cN_+ + 1)^{1/2} \xi \|^2 + CN^{1/2-\a} \| (\cN_+ + 1)^{1/2} \xi \| \| \cV_N^{1/2} \xi \| 
		\end{split} \end{equation*} 
Let us consider the last term in (\ref{eq:cE410j}). Switching to position space and using (\ref{eq:ddxy}) in Lemma \ref{lm:dp} and again (\ref{eq:etax}), we arrive at 		
\[ \begin{split} 
		&|\langle \xi , \cE_{105}^{(4)} \xi \rangle | \\
		&\leq C N\int_{\L^2} dx dy \, e^{2N} V(e^N(x-y))  \| (\cN_+ + 1)^{1/2} \xi \| \int_0^1 ds  \| (\cN_+ + 1)^{-1/2} \check{d}^{(s)}_x \check{d}^{(s)}_y \xi \| \\ &\leq CN \| \eta \|  \| (\cN_+ + 1)^{1/2} \xi \| \int_{\L^2} dx dy \, e^{2N}V(e^N(x-y)) \\ &\hspace{1cm} \times \left[  \| (\cN_+ + 1)^{1/2} \xi \| + \| \eta\| \| \check{a}_x  \xi \| + \| \eta\| \| \check{a}_y  \xi \| + N^{-1/2}\|\eta\| \| \check{a}_x \check{a}_y \xi \| \right] \\ &\leq C N^{1-\a} \| (\cN_++ 1)^{1/2} \xi \|^2  + CN^{1/2-2\a}\| (\cN_+ + 1)^{1/2} \xi \| \| \cV_N^{1/2} \xi \|\,. \end{split} \]
Summarizing, we have shown that (\ref{eq:cE410}) can be bounded by 
		\be \begin{split} \label{eq:E10-fin}
|\langle \xi , \cE_{10}^{(4)}\xi \rangle | \leq &  C N^{1/2-\a}   \| \cV_N^{1/2}\xi \| \| (\cN_+ +1)^{1/2} \xi \|+ C N^{1-\a}  \| (\cN_+ +1)^{1/2} \xi \|^2  
\end{split}\ee
for any $\a > 1$, $\xi \in \cF^{\leq N}_+$. 		
		Next, we come back to the terms $\text{W}_{11}, \text{W}_{12}, \text{W}_{13}$ introduced in (\ref{eq:W1-dec}). Using (\ref{eq:VetaN}) and $|\gamma_q^{(s)} -1| \leq C \eta_q^2$, we can write  
		\begin{equation}\label{eq:W11} \text{W}_{11} =  \frac{1}{2} \sum_{q \in \Lambda_+^* , r \in \Lambda^* : r \not = -q} \widehat{V} (r/e^N) \eta_{q+r} (b_q b_{-q} + \hc ) + \cE_{11}^{(4)}\,, \end{equation}
		where $\cE_{11}^{(4)}$ is such that 
		\[ \begin{split} | \langle \xi , \cE_{11}^{(4)} \xi \rangle | &\leq C N^{1-2\a} \| (\cN_+ + 1) \xi \|^2\,. \end{split} \]
Next, we can decompose the second term in (\ref{eq:W1-dec}) as 
\begin{equation}\label{eq:W12}  \text{W}_{12} = \frac{1}{2} \sum_{q \in \Lambda_+^* , r \in \Lambda^* : r \not = -q} \widehat{V} (r/e^N) \eta_{q+r} \eta_q  \left(1- \frac{\cN_+}{N} \right) + \cE_{12}^{(4)} \end{equation}
where $\pm \cE^{(4)}_{12}  \leq C N^{-\a} \cN_+ + N^{1-3\a}$. 

The third term on the r.h.s. of (\ref{eq:W1-dec}) can be written as  
		\begin{equation}\label{eq:W13} \text{W}_{13} = - \frac{1}{2} \sum_{q \in \Lambda_+^* , r \in \Lambda^* : r \not = -q} \widehat{V} (r/e^N) \eta_{q+r} \eta_q \left(1- \frac{\cN_+}{N} \right) \frac{\cN_+ +1}{N}  + \cE^{(4)}_{13} \end{equation}
		where $\cE^{(4)}_{13} = \cE_{131}^{(4)} +  \cE_{132}^{(4)} + \cE_{133}^{(4)} + \cE_{134}^{(4)}$, with 
		\begin{equation*} 
		\begin{split} 
		\cE^{(4)}_{131} = \; &\frac{1}{2} \sum_{q \in \Lambda_+^* , r \in \Lambda^* : r \not = -q} \widehat{V} (r/e^N) \eta_{q+r}  \int_0^1 ds \, (\g_q^{(s)} -1) b_q d_{-q}^{(s)} +\hc  \\
		\cE^{(4)}_{132} = \; &\frac{1}{2} \sum_{q \in \Lambda_+^* , r \in \Lambda^* : r \not = -q} \widehat{V} (r/e^N) \eta_{q+r} \int_0^1 ds \, b_q \left[ d_{-q}^{(s)} + s \eta_q \frac{\cN_+}{N} b_{q}^* \right] + \hc \\
		\cE^{(4)}_{133} = \;&- \frac{1}{2} \sum_{q \in \Lambda_+^* , r \in \Lambda^* : r \not = -q} \widehat{V} (r/e^N) \eta_{q+r} \eta_q b^*_q  b_{q} \frac{\cN_+ +1}{N} \\
		\cE^{(4)}_{134} = \; & \frac{1}{2N} \sum_{q \in \Lambda_+^* , r \in \Lambda^* : r \not = -q} \widehat{V} (r/e^N) \eta_{q+r} \eta_q a_q^* a_q \frac{\cN_+ +1}{N} \,.
		\end{split} \end{equation*}
With (\ref{eq:VetaN}), we immediately find 
		\[  \pm \cE_{133}^{(4)} \leq C N^{1-\a} (\cN_+ + 1) , \qquad \pm \cE_{134}^{(4)} \leq C N^{-\a}		(\cN_+ + 1)\,.  \]
With $|\g_q^{(s)} -1| \leq C \eta_q^2$, Lemma \ref{lm:dp} and, again, (\ref{eq:VetaN}), we also obtain  
		\[  \begin{split} |\langle \xi , \cE_{131}^{(4)} \xi \rangle | &\leq C N^{1-3\a} \| (\cN_+ + 1)^{1/2} \xi \|^2\,. \end{split} \]
		Let us now consider $\cE_{132}^{(4)}$. In position space, with $\check{\lis{d}}^{(s)}_y = d^{(s)}_y +  (\cN_+ / N) b^* (\check{\eta}_{y})$ and using (\ref{eq:splitdbd}), we obtain 
		\[ \begin{split} 
		|\langle \xi , \cE_{132}^{(4)} \xi \rangle |  &= \Big| \frac 1 2 \int_0^1 ds \int_{\L^2} dx dy\, e^{2N} V(e^N(x-y)) \check{\eta} (x-y) \langle \xi, \check{b}_x \check{\lis{d}}^{(s)}_y \xi \rangle \Big| \\ &\leq C N^{1-\a} 
		\int_{\L^2}  dx dy \,e^{2N} V(e^N(x-y)) \| (\cN_+ + 1)^{1/2} \xi \| \\ &\hspace{1cm} \times \left[  \| (\cN_+ + 1)^{1/2} \xi \|  + \| \check{a}_y\xi \|  + \| \check{a}_x \xi \| +N^{-1} \| \check{a}_x \check{a}_y \cN_+^{1/2} \xi \| \right] \\
		&\leq C N^{1-\a}\| (\cN_+ + 1)^{1/2} \xi \|^2 + C N^{1/2-\a} \| (\cN_+ + 1)^{1/2} \xi \| \| \cV_N^{1/2} \xi \|\,. \end{split} \]
It follows that \[ | \langle \xi, \cE_{13}^{(4)} \rangle | \leq C  N^{1/2 -\a} \|\cV_N^{1/2} \xi \| \| (\cN_++1)^{1/2}\xi \| +CN^{1-\a}\| (\cN_++1)^{1/2}\xi \|^2 . \] With \eqref{eq:E10-fin}, (\ref{eq:W11}), (\ref{eq:W12}), (\ref{eq:W13}), we obtain 
		\begin{equation}\label{eq:W1f} \begin{split}  \text{W}_1 = \; &\frac{1}{2} \sum_{q \in \L_+^*, r \in \L^* : r \not = -q} \widehat{V} (r/e^N) \eta_{q+r}  \big(b_q b_{-q} + \hc \big) \\ &+\frac{1}{2} \sum_{q \in \Lambda_+^* , r \in \Lambda^* : r \not = -q} \widehat{V} (r/e^N) \eta_{q+r} \eta_q  \left(1- \frac{\cN_+}{N} \right) \left(1- \frac{\cN_+ + 1}{ N} \right) + \cE^{(4)}_1 \end{split} \end{equation}
where
		\[ 
| \langle \xi, \cE_{1}^{(4)}\xi\rangle | \leq C  N^{1/2 -\a} \|\cV_N^{1/2} \xi \| \| (\cN_++1)^{1/2}\xi \| +CN^{1-\a}\| (\cN_++1)^{1/2}\xi \|^2\,,
\]
		
Next, we control the term $\text{W}_2$, from (\ref{eq:defW}). In position space, we find
		\begin{equation*} \begin{split} 
		\text{W}_2 = \int_{\Lambda^2} dx dy\, e^{2N} V(e^N(x-y))  \int_0^1  ds \big( e^{-sB} \check{b}^*_x \check{b}^*_y  e^{s B} 
		a^* (\check{\eta}_{x}) \check{a}_y  + \text{h.c.}   \big)
		\end{split} 
		\end{equation*}
		with $\check{\eta}_{x} (z) = \check{\eta} (x-z)$. By Cauchy-Schwarz, we have
		\[ \begin{split} 
		|\langle  \xi, \text{W}_2 \xi \rangle | &\leq \int_{\Lambda^2} dx dy \, e^{2N}V(e^N(x-y))  \int_0^1  ds  \\ &\hspace{1cm} \times  \| (\cN_+ +1)^{1/2} e^{-sB} \check{b}_x \check{b}_y  e^{s B}  \xi \| \| (\cN_+ +1)^{-1/2} a^* (\check{\eta}_{x}) \check{a}_y \xi \| \,.
		\end{split} \]
		With 
		\[ \| (\cN_+ +1)^{-1/2} a^* (\check{\eta}_{x}) \check{a}_y \xi \| \leq C \|  \eta \| \| \check{a}_y \xi \| \leq C N^{-\a} \| \check{a}_y \xi \|  \]
		and using Lemma \ref{lm:prel4}, we obtain 
		\be \begin{split}  \label{eq:W2end}
			|\langle \xi, \text{W}_2 \xi \rangle |&\leq  C N^{-\a} \int_{\Lambda^2}  dx dy \, e^{2N} V(e^N(x-y))  \| \check{a}_y \xi \| \\ &\hspace{1cm} \times \Big\{ N \| (\cN_+ +1)^{1/2} \xi \| + N  \| \check{a}_x \xi \| + N \| \check{a}_y \xi \| + N^{1/2}  \| \check{a}_x \check{a}_y \xi \| \Big\}  \\ 
			& \leq C N^{1-\a} \,  \| (\cN_+ +1)^{1/2} \xi \| ^2+CN^{1/2-\a}\| ( \cN_+  + 1 )^{1/2} \xi \| \|\cV_N^{1/2}\xi\|\,. \end{split} \ee
		Also for the term $\text{W}_3$ in (\ref{eq:defW}), we switch to position space. We find
		\begin{equation*}
		\begin{split}  \text{W}_3 = &\;  \int_{\L^2}  dx dy \, e^{2N} V(e^N(x-y))  \\
		&  \times \int_0^1 ds\, \int_0^s d \t \, \big( e^{-sB} \check{b}^*_x \check{b}^*_y e^{sB} \, e^{- \t B} b^*(\check{ \eta}_{x}) b^* (\check{\eta}_{y}) e^{\t B} + \text{h.c.} \big) \,.\end{split} \end{equation*}
		and thus 
		\[ \begin{split} & \left|\langle \xi, \text{W}_3 \xi \rangle \right|  \leq  \int_{\L^2} dx dy \, e^{2N} V(e^N(x-y))  \int_0^1 ds\, \int_0^s d \t \, \| (\cN_+ +1)^{1/2} e^{-sB} \check{b}_x \check{b}_y e^{sB} \xi \| \\ &\hspace{5cm} \times  \|  (\cN_+ +1)^{-1/2}  e^{- \t B}  b^*(\check{ \eta}_{x})\, b^* (\check{ \eta}_{y}) e^{ \t B} \xi \|\,. 
		\end{split} \]
		With Lemma \ref{lm:Ngrow}, we find 
		\begin{equation*}
		\begin{split} 
		\| (\cN_+ +1)^{-1/2}  e^{- \t B}  b^*(\check{ \eta}_{x})\, b^*(\check{ \eta}_{y}) e^{ \t B} \xi \| &  \leq C \| \eta\|^2 \| (\cN_+ +1)^{1/2} \xi \|  \,.
		\end{split} \end{equation*}
		Using Lemma \ref{lm:prel4}, we conclude that 
		\be \begin{split}  \label{eq:W3end}
			|\langle \xi, \text{W}_3 \xi \rangle |&\leq  C \|\eta\|^2\, \int_{\Lambda^2}  dx dy \, e^{2N} V(e^N(x-y))  \| (\cN_+ +1)^{1/2}\xi \| \\ &\hspace{1cm} \times \Big\{ N \| (\cN_+ +1)^{1/2} \xi \| + N  \| \check{a}_x \xi \| + N \| \check{a}_y \xi \| + N^{1/2}  \| \check{a}_x \check{a}_y \xi \| \Big\}  \\ 
			& \leq C N^{1-2\a} \,  \| (\cN_+ +1)^{1/2} \xi \|^2 +CN^{1/2-2\a}\| \cV_N^{1/2} \xi \| \| (\cN_+ +1)^{1/2} \xi \|.\end{split} \ee
		The term $\text{W}_4$ in (\ref{eq:defW}) can be bounded similarly. In position space, we find 
		\begin{equation*} 
		\begin{split} \text{W}_4 = \; & \int dxdy \, e^{2N} V(e^N(x-y)) \\
		& \times\int_0^1 ds \int_0^s d \t \,  \big( e^{-sB} \check{b}^*_x \check{b}^*_y  \, e^{sB} \,  e^{-\t B} b(\check{\eta^2_x}) \check{b}_y e^{\t B} + \text{h.c.} \big) 
		\end{split} \end{equation*}
		with $\check{\eta^2}$ the function with Fourier coefficients $\eta^2_q$, for $q \in \L^*$, and where $\check{\eta^2_x}(y) := \check{\eta^2}(x-y)$. Clearly  $\| \check{\eta^2_{x}} \|  \leq C \| \check{\eta}\|^2 \leq C N^{-2\a}$. With Cauchy-Schwarz and Lemma \ref{lm:Ngrow}, we obtain 
		\[ \begin{split} |\langle \xi , \text{W}_4 \xi \rangle | \leq \; & C N^{-2\a} \int_0^1 ds \int_0^s d\tau \int dx dy \, e^{2N} V(e^N(x-y)) \\ &\hspace{2cm} \times  \| (\cN_+ + 1)^{1/2} \check{b}_y \check{b}_x e^{sB} \xi \|  \| \check{b}_y e^{\tau B} \xi \|\,. \end{split} \] 
		Applying Lemma \ref{lm:prel4} and then Lemma \ref{lm:Ngrow}, we obtain 
		\[ \begin{split} 
		|\langle \xi , \text{W}_4 \xi \rangle | \leq &\; C N^{-2\a} \int_0^1 ds \int_0^s d\tau \int dx dy\, e^{2N} V(e^N(x-y))
		\| \check{b}_y e^{\tau B} \xi \| \\ &\hspace{1cm} \times \left\{  N \| (\cN_+ +1)^{1/2} \xi \| + N \| \check{a}_x \xi \| + N \| \check{a}_y \xi \|  + N^{1/2} \| \check{a}_x \check{a}_y \xi \| \right\} 
		\\ \leq \; &C  N^{1-2\a } \| (\cN_+ + 1)^{1/2} \xi \|^2  +CN^{1/2-2\a} \| \cV_N^{1/2} \xi \| \| (\cN_+ + 1)^{1/2} \xi \|\,.
		\end{split} \]
From (\ref{eq:W1f}), (\ref{eq:W2end}), (\ref{eq:W3end}) and the last bound, we conclude that 
		\[ \begin{split} \cG^{(4)}_{N,\a} = \; &\cV_N + \frac{1}{2} \sum_{q \in \L_+^*, r \in \L^* : r \not = -q} \widehat{V} (r/e^N) \eta_{q+r}  \big(b_q b_{-q} + \hc \big) \\ &+\frac{1}{2} \sum_{q \in \Lambda_+^* , r \in \Lambda^* : r \not = -q} \widehat{V} (r/e^N) \eta_{q+r} \eta_q  \left(1- \frac{\cN_+}{N} \right) \left(1- \frac{\cN_+ + 1}{ N} \right) + \cE^{(4)}_{N,\a} \end{split} \] 
		where $\cE_{N,\a}^{(4)}$ satisfies \eqref{eq:E4bound2}. 

\end{proof}\subsection{Proof of Proposition \ref{prop:GN}}
\label{sub:proofGN}

With the results established in Subsections \ref{sub:G0} - \ref{sub:G4}, we cam now show Prop. \ref{prop:GN}. Propositions \ref{prop:G0}, \ref{prop:K}, \ref{prop:G2V}, \ref{prop:GN-3}, \ref{prop:GN-4}, imply that 
\begin{equation} \begin{split} \label{eq:proofGNell-1}
& \cG_{N,\a} =    \frac{\widehat{V} (0)}{2}\, (N +\cN_+ -1) \, (N-\cN_+)\\
& \hskip 0.2cm+ \sum_{p \in \L^*_+} \eta_p \Big[p^2 \eta_p + N\widehat{V} (p/e^N) + \frac 12\sum_{\substack{r \in \L^*\\ p+r \neq 0}} \widehat{V} (r/e^N)  \eta_{p+r}\Big]\Big(\frac{N-\cN_+}{N}\Big) \Big(\frac{N-\cN_+ -1}{N}\Big)\\
&  \hskip 0.2cm+\cK +N\sum_{p \in \Lambda^*_+}  \widehat{V} (p/e^N) a^*_pa_p\Big(1-\fra{\cN_+}{N}\Big) \\
&  \hskip 0.2cm+ \sum_{p \in \L^*_+} \Big[\; p^2 \eta_p + \frac N2 \widehat{V} (p/e^N) + \frac 12 \sum_{r \in \L^*:\; p+r \neq 0}\hskip -0.5cm \widehat{V} (r/e^N)  \eta_{p+r} \; \Big]  \big( b^*_p b^*_{-p} + b_p b_{-p} \big) \\
&  \hskip 0.2cm+ \sqrt{N} \sum_{p,q \in \L^*_+ :\, p + q \not = 0} \widehat{V} (p/e^N) \left[ b_{p+q}^* a_{-p}^* a_q  + \hc \right]  +\cV_N   + \cE_{1}
\end{split} \end{equation}
where 
\begin{equation*}
| \langle \xi, \cE_1  \xi \rangle | \leq  C N^{1/2 -\a} \| \cH_N^{1/2}\xi\| \| (\cN_++1)^{1/2}\xi \| + C N^{1-\a}  \| (\cN_++1)^{1/2}\xi \|^2 
\end{equation*}
for any $\a >1$ and  $\xi \in \cF^{\leq N}_+$. With \eqref{eq:eta-scat}, we find 
\begin{equation*} 
\begin{split} 
&\sum_{p \in \L^*_+}  \eta_p \Big[p^2 \eta_p + N\widehat{V} (p/e^N) + \frac 1 {2} \sum_{r \in \L^*:\; p+r \neq 0}\hskip -0.5cm \widehat{V} (r/e^N)  \eta_{p+r}\Big] \\
& = \sum_{p \in \L^*_+} \eta_p \Big[ \;\frac N 2 \widehat{V} (p/e^N)  +Ne^{2N} \l_\ell \widehat \chi_\ell(p) +e^{2N} \l_\ell \sum_{q \in \L^*} \widehat \chi_\ell(p-q) \eta_q - \frac 12 \widehat V(p/e^N) \eta_0\;\Big] 
\end{split} \end{equation*}
From Lemma  \ref{lm:propomega} and estimating  $\| \widehat \chi_\ell\|= \| \chi_\ell\|  \leq C N^{-\a}$, $\|\eta\| \leq C N^{-\a}$ and $\| \widehat{\chi}_\ell * \eta \| = \| \chi_\ell \check{\eta} \| \leq \| \check{\eta} \| 
\leq C N^{-\a}$, we have 
\[
\begin{split}
\Big | Ne^{2N} \l_\ell \sum_{p \in \L^*_+} \eta_p \widehat \chi_\ell(p) \Big | 
\leq CN^{2\a}\|\widehat{\chi}_\ell\|\|\eta\| \leq C,
\end{split}\]
and
\[
\begin{split}
\Big| e^{2N} \l_\ell \sum_{\substack{p \in \L^*_+,\, q \in \L^*}} \widehat \chi_\ell(p-q) \eta_q\eta_p \Big| 
& \leq C N^{2\a-1} \| \hat \chi_\ell \ast \eta \|\| \eta\| \leq CN^{-1}.
\end{split}\]
Moreover, using \eqref{eq:VetaN} and the bound  \eqref{eq:wteta0} we find 
\[
\Big| \frac 12 \sum_{p \in \L^*_+} \widehat V(p/e^N) \eta_p \eta_0 \Big| \leq C N^{1-2\a}\,.
\]
We obtain 
\begin{equation*} \begin{split} 
 & \sum_{p \in \L^*_+} \eta_p  \Big[p^2 \eta_p +  N\widehat{V} (p/e^N) + \frac 1 {2} \sum_{\substack{r \in \L^*\\ p+r \in \L^*_+}} \widehat{V} (r/e^N)  \eta_{p+r}\Big]\Big(\frac{N-\cN_+}{N}\Big) \Big(\frac{N-\cN_+ -1}{N}\Big) \\ 
& \qquad  = 
\frac{N}{2} \sum_{p \in \L^*_+} \widehat{V} (p/e^N) \eta_p \left( \frac{N-\cN_+}{N} \right) \left( \frac{N-\cN_+ - 1}{N} \right) +\cE_2 \end{split} \end{equation*}
with $\pm \cE_2  \leq C $ for all $\a \geq 1/2$. 
On the other hand, using \eqref{eq:wteta0} we have
\[ \begin{split}
\frac N 2 \sum_{p \in \L^*_+}  \widehat{V} (p/e^N) \eta_p =\; & \frac N 2  \big( \widehat V(\cdot / e^N) \ast \eta \big)(0) - \frac N 2  \widehat{V} (0) \eta_0 \\
=\; & \frac{N^2}2 \Big( \int dx V(x) f_{\ell}(x) -  \widehat V(0) \Big) + \tl \cE_2
\end{split}\]
with $\pm \tl \cE_2 \leq C N^{1-2\a}$. 
With the first bound in \eqref{eq:omegahat0} we conclude that
\begin{equation}\label{eq:1lineG} \begin{split} 
 & \sum_{p \in \L^*_+} \eta_p  \Big[p^2 \eta_p + N\widehat{V} (p/e^N) + \frac 1 {2} \sum_{\substack{r \in \L^*\\ p+r \in \L^*_+}} \widehat{V} (r/e^N)  \eta_{p+r}\Big]\Big(\frac{N-\cN_+}{N}\Big) \Big(\frac{N-\cN_+ -1}{N}\Big)  \\ 
& = \frac 1 {2N}  \left[ \widehat{\o}_N(0) - N\widehat{V} (0) \right] (N-\cN_+ -1) \left( N-\cN_+ \right)   + \cE_3 \end{split} \end{equation}
where $\pm \cE_3 \leq C$, if $\alpha \geq 1/2$. Using \eqref{eq:eta-scat}, we can also handle the fourth line of \eqref{eq:proofGNell-1}; we find 
\begin{equation}\begin{split} \label{eq:proofGNellQ}
& \sum_{p \in \L^*_+} \Big[\, p^2 \eta_p + \frac N2 \widehat{V} (p/e^N) + \frac 12 \sum_{r \in \L^*:\; p+r \in \L^*_+}\hskip -0.5cm \widehat{V} (r/e^N)  \eta_{p+r} \, \Big]  \big( b^*_p b^*_{-p} + b_p b_{-p} \big) \\
& \quad = \sum_{p \in \L^*_+} \Big[ Ne^{2N} \lambda_\ell \widehat{\chi}_\ell (p) +e^{2N} \lambda_\ell \sum_{q \in \Lambda^*} \widehat{\chi}_\ell (p-q) \eta_q - \frac 1 2 \widehat{V} (p/e^N)  \eta_{0}  \Big] \big( b^*_p b^*_{-p} + b_p b_{-p} \big)\,.
\end{split} \end{equation}
The last two terms on the right hand side of \eqref{eq:proofGNellQ} are error terms. With \eqref{eq:wteta0} and \eqref{eq:VetaN} we have 
\[ \begin{split}
\Big| \sum_{p \in \L^*_+}&\widehat{V} (p/e^N)  \eta_{0} \big( b^*_p b^*_{-p} + b_p b_{-p} \big) \Big| \\
& \leq  C N^{-2\a}  \bigg[  \sum_{p \in \L^*_+} \frac{|\widehat{V} (p/e^N)|^2}{p^2} \bigg]^{1/2}  \bigg[  \sum_{p \in \L^*_+} p^2 \| a_p \xi \|^2 \bigg]^{1/2}  \|(\cN_++1)^{1/2} \xi \| \\
& \leq C N^{1/2 - 2\a} \| \cK^{1/2}\xi\|  \|(\cN_++1)^{1/2} \xi \|\,.
\end{split}\]
The second term on the right hand side of \eqref{eq:proofGNellQ} can be bounded in position space: 
 \[\begin{split}
    & \Big| \langle \xi,\; e^{2N} \lambda_\ell \sum_{p\in \L^*_+} (\widehat{\chi}_\ell * \eta)(p) (b_p^* b_{-p}^* + b_p b_{-p}) \xi \rangle \Big|  \\
     &\quad\leq CN^{2\a-1} \| (\cN_+ + 1)^{1/2} \xi \| \int_{\L^2} dxdy \,\chi_\ell(x-y) |\check{\eta}(x-y)|\|(\cN_+ + 1)^{-1/2}\check{b}_x\check{b}_y\xi\|\\ 
     &\quad\leq CN^{\a-1}\|(\cN_+ + 1)^{1/2} \xi \| \left[\int_{\L^2} dxdy \, \chi_\ell(x-y) \|(\cN_+ + 1)^{-1/2}\check{a}_x\check{a}_y\xi\|^2\right]^{1/2}.
 \end{split}\]
 The term in parenthesis can be bounded similarly as in \eqref{eq:first}. Namely,
 \[
 \int_{\L^2} dxdy \, \chi_\ell(x-y) \|(\cN_+ + 1)^{-1/2}\check{a}_x\check{a}_y\xi\|^2 \leq Cq N^{-2\a/q'}\|\cK^{1/2}\xi\|^2
 \]
 for any $q >2$ and $1 < q' < 2$ with $1/q +1/q' =1$. Choosing $q =\log N$, we get
    \[ \begin{split}
 \Big| \langle \xi,e^{2N} \lambda_\ell &\sum_{p\in \L^*_+} (\widehat{\chi}_\ell * \eta)(p) (b_p^* b_{-p}^* + b_p b_{-p}) \xi \rangle \Big|   \\
&\leq C N^{-1} (\log N)^{1/2}  \| (\cN_+ + 1)^{1/2} \xi \| \|\cK^{1/2}\xi\| \,,
\end{split}\]

and, from (\ref{eq:proofGNellQ}), we conclude that 
\begin{equation} \label{eq:2lineG}
\begin{split}
\sum_{p \in \L^*_+} &\Big[p^2 \eta_p + \frac N2 \widehat{V} (p/e^N) + \frac 1 {2} \sum_{\substack{r \in \L^* : \\ p+r \in \L^*_+}}  \widehat{V} (r/e^N)  \eta_{p+r} \Big]  \big( b^*_p b^*_{-p} + b_p b_{-p} \big) \\ &= \sum_{p \in \L^*_+} Ne^{2N} \lambda_\ell \widehat{\chi}_\ell (p)\big( b^*_p b^*_{-p} + b_p b_{-p} \big) + \cE_4\,,
\end{split}\end{equation}
with \[ | \langle \xi, \cE_4  \xi \rangle| \leq C  N^{-1} (\log N)^{1/2}  \|(\cN_++1)^{1/2}\xi\| \| \cK^{1/2}\xi \|.\] 
if $\alpha > 1$. 
Combining (\ref{eq:proofGNell-1}) with \eqref{eq:1lineG} and (\ref{eq:2lineG}), and using the definition \eqref{eq:defomegaN} we conclude that 
\begin{equation} \begin{split}  \label{eq:proofGNell-2}
\cG_{N,\a} = \; &   \frac 1 2 \,\widehat{\o}_N(0) (N-1)\Big(1-\frac{\cN_+}{N}\Big) + \Big[N \widehat{V} (0)- \frac 1 2 \,\widehat{\o}_N(0) \Big]\, \cN_+\Big(1-\frac{\cN_+}{N}\Big)\\
&+ N \sum_{p \in \Lambda^*_+}  \widehat{V} (p/e^N) a^*_pa_p \left(1- \frac{\cN_+}{N} \right) + \frac 1 2  \sum_{ p\in  \L_+^*}\widehat{\o}_N(p)(b_pb_{-p}+\hc)\\
& + \sqrt{N} \sum_{p,q \in \L^*_+ : p + q \not = 0} \widehat{V} (p/e^N) \left[ b_{p+q}^* a_{-p}^* a_q  + \hc \right]  \\
& +\cK  +\cV_N  + \cE_5 \,,
\end{split} \end{equation} 
with
\[\begin{split}
| \langle \xi, \cE_5  \xi \rangle|   \leq \; &C N^{1/2 -\a} \| \cH_N^{1/2}\xi\| \| (\cN_++1)^{1/2}\xi \| + C N^{1-\a}  \| (\cN_++1)^{1/2}\xi \|^2 \\
&+ C  N^{-1} (\log N)^{1/2}  \| \cK^{1/2}\xi \| \|(\cN_++1)^{1/2}\xi\|  + C \| \xi \|^2  \,,
\end{split}
\]
for any $\a> 1$. Observing that $|\widehat{V} (p/e^N) - \widehat{V} (0)| \leq C |p| e^{-N}$ in the second line on the r.h.s. of (\ref{eq:proofGNell-2}), we arrive at $\cG_{N,\a} = \cG^\text{eff}_{N,\a} + \cE_{\cG}$, with $\cG^\text{eff}_{N,\a}$ defined as in (\ref{eq:GNeff}) and with $\cE_{\cG}$ that  satisfies \eqref{eq:GeffE}.

\section{Properties of the Scattering Function} \label{App:omega}

Let $V$ be a potential with finite range $R_0>0$ and scattering length $\aa$. For a fixed $R>R_0$, we study properties of the ground state $f_R$ of the Neumann problem
\be \label{tlf2}
\Big( -\D + \frac {1}{2} V(x) \Big) f_{R}(x) =  \l_{R}\,  f_{R}(x)
\ee
on the ball $|x|\leq R$, normalized so that $f_{R}(x)=1$ for $|x|=R$.  Lemma \ref{lm:propomega}, parts i)-iv), follows by setting $R=e^N \ell$ in the following lemma. 

\begin{lemma} \label{lm:appA}
	Let $V\in L^3(\bR^2)$ be non-negative, compactly supported and spherically symmetric, and denote its scattering length by $\aa$. Fix $R>0$ sufficiently large and denote by $f_{R}$ the Neumann ground state of  \eqref{tlf2}. Set  $w_{R}=1 -  f_{R}$. Then we have 
\[  \label{eq:bounds-fR}
0\leq f_R(x)\leq 1
\]
Moreover, for $R$ large enough there is a constant $C>0$ independent of $R$ such that
		\be \label{eq:eigenvalue2}
		\left | \,\l_{R} - \frac{2}{R^2 \log(R/\aa)}  \left( 1 + \frac{3}{4}\fra{1}{\log(R/\aa)} \right) \,\right | \leq  \frac{C}{R^2} \frac{1}{\log^3(R/\aa)}  \,.
		\ee
and
		\be \label{eq:intpotf2}  
		\left| 
		\int \di x\, V(x) f_{R}(x) - \frac{4\pi}{\log(R/\aa)}  \right| \leq  \frac{C}{\log^2 (R/\aa)}   \,.
		\ee
Finally, there exists a constant $C>0$ such that 
		\be \begin{split} \label{eq:w-bounds}
		|w_{R}(x)| &\leq   \chi(|x|\leq R_0)+ C\, \frac{\log(|x|/R) }{\log (\aa/R)} \,\chi(R_0 \leq  |x| \leq R  )  \\[0.2cm]
|\nabla w_{R}(x)| & \leq \frac{C}{\log (R /\aa)} \frac { \chi(|x|\leq R)} {|x| + 1}
		\end{split} \ee
for $R$ large enough.
\end{lemma}

To show Lemma \ref{lm:appA} we adapt to the two dimensional case the strategy used in \cite[Lemma A.1]{ESY0} for the three dimensional problem. We will use some well known properties of the zero energy scattering equation in two dimensions, summarized in the following lemma.

\begin{lemma} 
Let $V \in L^3 (\bR^2)$ non-negative, with $\text{supp } V \subset B_{R_0} (0)$ for an $R_0 > 0$. Let $\aa \leq R_0$ denote the scattering length of $V$. For $R > R_0$, let $\phi_R : \bR^2 \to \bR$ be the radial solution of the zero energy scattering equation 
\be\label{eq:fell}
\left[ -\D + \frac 12 V \right] \phi_R= 0
\ee
normalized such that $\phi_R (x) = 1$ for $|x| = R$. Then 
\be \label{A.3}
\phi_R(x) 
= \frac{\log(|x| / \aa)}{\log(R/\aa)} 
\ee
for all $|x| > R_0$. Moreover, $|x| \to \phi_R (x)$ is monotonically increasing and there exists a constant $C > 0$ (depending only on $V$)  such that 
\begin{equation}\label{eq:low-phi}
\phi_R (x) \geq \phi_R (0) \geq \frac{C}{\log (R/ \aa)} 
\end{equation}
for all $x \in \bR^2$. Furthermore, there exists a constant $C > 0$ such that 
\begin{equation}\label{eq:nablaphi}
 |\nabla \phi_R(x)| \leq \frac {C}{|\log (R/\aa)|} \frac 1 {|x|+1} \end{equation}
for all $x \in \bR^2$. 
\end{lemma}

\begin{proof} 
The existence of the solution of (\ref{eq:fell}), the expression (\ref{A.3}), the fact that $\phi_R (x) \geq 0$ and the monotonicity are standard (see, for example, Theorem C.1 and Lemma C.2 in \cite{LSSY}). The bound (\ref{eq:low-phi}) for $\phi_R (0)$ follows from (\ref{A.3}), comparing $\phi_R (0)$ with $\phi_R (x)$ at $|x| = R_0$, with Harnack's inequality (see \cite[Theorem C.1.3]{Simon}). Finally, (\ref{eq:nablaphi}) follows by rewriting (\ref{eq:fell}) in integral form 
\[ \label{scatt2d}
\phi_R (x) = 1 - \frac{1}{4\pi} \int_{\bR^2} \log \Big( R/|x-y| \Big) V(y) \phi_R (y) \di y \,.
\]
For $|x| \leq R_0$, this leads (using that $\phi_R (y) \leq \log (R_0/\aa) / \log (R/\aa)$ for all $|y| \leq R_0$ and the local integrability of $|.|^{-3/2}$) to 
\[ 
|\nabla \phi_R (x) | \leq C \int \frac{V(y) \phi_R(y)}{|x-y|} dy  \leq \frac{C \| V \|_3}{\log (R/\aa)} 
\] 
Combining with the bound for $|x| > R_0$ obtained differentiating (\ref{A.3}), we obtain the desired estimate. 
\end{proof}

\begin{proof}[Proof of Lemma \ref{lm:appA}]   By standard arguments (see for example \cite[proof of theorem C1]{LSSY}), $f_R(x)$ is spherically symmetric and non-negative.  We now start by proving an upper bound for $\l_R$, consistent with \eqref{eq:eigenvalue2}. To this end, we calculate the energy of a suitable trial function. For $k \in \bR$ we define
\[ \label{eq:PsiOut}
\ps_k (x)= J_0(k |x|) - \frac{J_0(k \aa)}{Y_0(k\aa)} Y_0(k|x|)\,.
\]
with $J_0$ and $Y_0$ the zero Bessel functions of first and second type, respectively. Note that 
\[ 
- \D \ps_k (x) = k^2 \ps_k (x)\,.
\]
and  $\ps_k (x)=0$ if $|x| =a$. We define $k= k(R)$ to be the smallest positive real number satisfying $\dpr_r \psi_R (x)=0$ for $|x| = R$. One can check that 
\be \label{eq:chosenk}
\left| \,k^2 - \frac{2}{R^2 \log(R/\aa)} \left( 1 + \frac{3}{4}\fra{1}{\log(R/\aa)} \right)  \,\right| \leq \frac C {R^2} \frac{1}{\log^3 (R/\aa)}
\ee
in the limit $R\to \io$. To prove (\ref{eq:chosenk}), we observe that 
\be \label{eq:der-R}
\dpr_r \psi_k (x)\Big|_{|x| = R} = - k J_1 (kR) + k \frac{J_0(k \aa)}{Y_0(k \aa)} Y_1(kR) \,,
\ee
and we expand for $kR, k\aa \ll 1$ using (with $\g$ the Euler constant)
\be \begin{split} \label{eq:expansion}
 & \left| J_0(r) -1 +\frac {r^2} 4   \right| \leq C r^4 \,, \hskip 0.5cm \left| J_1(r) - \frac r 2  \Big(1-\fra {r^2}8\Big)  \right| \leq C r^5\,,  \\
&   \left| Y_0(r) - \frac 2 \pi \log(r e^\g /2) \right| \leq C r^2 \log(r) \,, \\
&  \left| Y_1(r) + \frac 2 \pi  \frac 1 r \left(1-\frac{r^2}{2} \Big(1-\frac{r^2}8\Big) \log (re^\g/2) +\fra{r^2}{4}\right)\right| \leq C r^3.
\end{split}\ee
With \eqref{eq:expansion} one finds that \eqref{eq:der-R} 
\begin{equation}\label{eq:neumann} \begin{split}  \dpr_r \psi_R &(x)\Big|_{|x| = R} \\ = \; &-\frac{1}{2 kR \log (k\aa e^\gamma / 2)} \\ & \cdot \left\{ \frac{(kR)^4}{8} \log (R/\aa) - (kR)^2 \left[ \log (R/\aa)  - \frac{1}{2} \right] + 2 + \cO ( (kR)^4 + (k\aa)^2)  \right\} \end{split} \end{equation}
The smallest solution of 
\[ \frac{(kR)^4}{8} \log (R/\aa) - (kR)^2 \left[ \log (R/\aa)  - \frac{1}{2} \right] + 2  = 0 \]
is such that 
\begin{equation}\label{eq:kvalue} (kR)^2 = \frac{2}{\log (R/\aa)} \left[ 1 + \frac{3}{4 \log (R/\aa)} \right] + \cO (\log^{-3} (R/\aa)) \end{equation}
in the limit of large $R$. Inserting in (\ref{eq:neumann}), we find that the r.h.s. changes sign around the value of $k$ defined in (\ref{eq:kvalue}). By the intermediate value theorem, we conclude that there is a $k = k (R) > 0$ satisfying (\ref{eq:chosenk}), such that $\partial_r \psi_{k(R)} (x) = 0$ if $|x| = R$.

Now, let $\phi_R (x)$ be the solution of the zero energy scattering equation \eqref{eq:fell}, with $\phi_R(x)=1$ for $|x| = R$. We set
\be \label{eq:defPsi}
\Psi_R (x) := \psi_{k} (m_R(x)) = J_0 (k m_R (x)) - \frac{J_0(k\aa)}{Y_0(k\aa)} Y_0(k m_R(x))\,,
\ee
with $k = k(R)$ satisfying \eqref{eq:chosenk} and
\[ \label{eq:defmR}
m_R(x):= \aa \exp\big( \log(R/\aa) \phi_{R}(x)\big)\,.
\]
With this choice we have $m_R(x)=|x|$ outside the range of the potential;  
hence $\Psi_R (x)=\psi_{k} (x)$  for  $R_0 \leq |x| \leq R$. In particular, $\Psi_R$ satisfies Neumann boundary conditions at $|x|=R$. 

From \eqref{A.3}, \eqref{eq:low-phi} and the monotonicity of $\phi_R$, we get 
\be \label{eq:mR-bnd}
C \aa  \leq m_R(x) \leq R_0  \qquad \text{for all} \quad 0\leq |x| \leq R_0 \, 
\ee
and for a constant $C > 1$, independent of $R$. From \eqref{eq:nablaphi} we also get 
\begin{equation}\label{eq:nablam} |\nabla m_R(x)|  \leq C  \qquad \text{for all } \quad 0\leq |x| \leq R \,.
\end{equation}
With the notation  ${\mathfrak h}= -\D + \frac 12 V $, we now evaluate $\bmedia{\Psi_R,  {\mathfrak h} \Psi_R}$. To this end we note that
\be \label{eq:splitting}
\bmedia{\Psi_R, \mathfrak{h} \Psi_R} = \int_{|x| < R_0} \lis{\Psi_R(x)} ( \mathfrak{h} \Psi_R(x)) dx  + k^2 \int_{|x| \geq R_0}  |\Psi_R(x)|^2 \, dx \,.
\ee
Let us consider the region $|x| < R_0$. From \eqref{eq:defPsi} and \eqref{eq:expansion} we find, first of all, 
\be \label{eq:trialPsi}
\left| \Psi_R(x) + \frac{\log(m_R(x)/\aa)}{\log (k \aa e^\g /2)}\right| \leq C  ( k m_R (x))^2\,, 
\ee
Next, we compute $-\D \Psi_R (x)$. With 
\begin{align*}
J'_0(r) &=- J_1(r)   &&  J'_1(r)  =\frac 1 2 \big( J_0(r)- J_2(r)\big)  \\
Y'_0(r)& = - Y_1(r)  && Y'_1(r)  =\frac 1 2 \big( Y_0(r)- Y_2(r)\big) \,.
\end{align*}
we obtain (here, we use the notation $m'_R$ and $m''_R$ for the radial derivatives of the radial function $m_R$) 
\[ \begin{split}
- \D \Psi_R (x) =\; & - \dpr^2_r \Psi_R (x)  - \frac{1}{|x|} \dpr_r \Psi_R (x) \\
=\; &  - k m_R'' (x) \big[ -J_1(k m_R (x) ) + \frac{J_0(k\aa)}{Y_0(k \aa)} Y_1(km_R(x)) \big] \\
& - \frac 1 2 k^2 \big(m_R'(x)\big)^2 \big[ J_2(k m_R(x) ) - \frac{J_0(k \aa)}{Y_0(k \aa)} Y_2(km_R(x)) \big] \\
& - \frac 1 2 k^2 \big(m_R'(x)\big)^2 \big[ -J_0(k m_R(x) ) + \frac{J_0(k\aa)}{Y_0(k\aa)} Y_0(km_R (x)) \big] \\
& - \frac{k m_R'(x)}{|x|} \big[ -J_1(k m_R(x) ) + \frac{J_0(k\aa)}{Y_0(k\aa)} Y_1(km_R(x)) \big] \,.
\end{split}\]
We note that, using the scattering equation \eqref{eq:fell},  
\be \label{eq:mell}
m_{R}'' - \frac{(m_R')^2}{m_R} + \frac{1}{|x|} m_R' = \frac 1 2 V m_R \, \phi_R \log(R/\aa)= \frac 12 V m_R \log(m_R/\aa)\,.
\ee
Now we write
\be \begin{split} \label{eq:D1}
& - \D \Psi_R (x) \\
&=  \Big[ -k \Big( m''_R (x) + \frac{m'_R(x)}{|x|} \Big)Y_1(km_R (x)) + \frac {k^2} 2 (m'_R(x))^2 Y_2(km_R(x))  \Big]    \frac{J_0(k\aa)}{Y_0(k \aa)} + g_R(x)
\end{split}\ee
where $g_R(x)= \sum_{i=1}^3 g_R^{(i)}(x)$ with
\[ \begin{split}
g_R^{(1)}(x) & = k \,\Big(  m''_R(x) +\frac{m'_R(x)}{|x|} \Big) J_1(k m_R(x)) \\
g_R^{(2)}(x) & = - \frac 1 2  k^2 (m'_R(x))^2 J_2(k m_R(x)) \\
g_R^{(3)}(x) & = -\frac 1 2 k^2 (m'_R(x))^2 \Big( -J_0(k m_R(x)  +\frac{J_0(k\aa)}{Y_0(k\aa)} Y_0(km_R (x))  \Big)= \frac{k^2} 2  (m'_R(x))^2  \Psi_R (x) \,.
\end{split}\]
With \eqref{eq:mell}, \eqref{eq:expansion} and (\ref{eq:mR-bnd}), (\ref{eq:nablam}), we find
\[
|g_R^{(1)}(x) | \leq C k^2 \Big( (m'_R(x))^2 + \frac 1 2 V(x) m_R^2 (x) \log(m_R(x)/\aa)\Big) \leq C k^2 (1+V(x)) \,.
\]
 Next, with 
$ | J_2(r) - r^2/8 | \leq C r^4  $
we get
\[
|g_R^{(2)}(x) |\leq  C k^4  (m'_R(x))^2 (m_R(x))^2 \leq C k^4\,.
\]
With (\ref{eq:trialPsi}), we can also bound 
\[
|g_R^{(3)}(x)| \leq  C k^2 (m'_R(x))^2 \frac{\log(m_R(x)/\aa)}{\log (k \aa)} \leq C k^2 \log^{-1}(k \aa)\,.
\]
We conclude that $ |g_R(r)| \leq C (1+V(x)) k^2 $ for all $r \leq R_0$ and $R$ large enough.   Finally, using Eq. \eqref{eq:mell}, the expansion for $Y_1(r)$ in Eq. \eqref{eq:expansion}, and the bound
\[
\Big|Y_2(r) + \frac 4 \pi \frac 1 {r^2} \Big| \leq C \,,
\]
 we can rewrite the first term on the r.h.s. of \eqref{eq:D1} as
\be \begin{split} \label{eq:D2}
&\Big[ -k \Big( m''_R(x) + \frac{m'_R(x)}{|x|} \Big)Y_1(km_R(x)) + \frac {k^2} 2 (m'_R(x))^2 Y_2(km_R(x))  \Big]    \frac{J_0(k\aa)}{Y_0(k \aa)} \\
& \qquad  = \frac 1 \pi V (x) \log(m_R(x)/\aa) \frac{J_0(k\aa)}{Y_0(k \aa)} + h_R(x)
\end{split}\ee
with $|h_R(x)|\leq C (1+V(x)) k^2$ for all $r \leq R_0$, $R$ large enough. With the identities  \eqref{eq:D1} and \eqref{eq:D2} we obtain 
\[ \begin{split} \label{eq:DPsi}
\bigg| - \D \Psi_R (x)- \frac 1 \pi \frac{J_0(k\aa)}{Y_0(k\aa)} V (x)  \log(m_R(x)/\aa ) \bigg| \leq C (1+V(x)) k^2\,,
\end{split}\]
for all $|x| \leq R_0$ and for $R$ sufficiently large. With (\ref{eq:trialPsi}), we conclude that, for $0 \leq |x| \leq R_0$,  
\be \begin{split} \label{eq:hPsi}
\Big|(- \D + \frac 12 V )\Psi_R (x) \Big| \leq C (1+V(x)) k^2  \,.
\end{split}\ee

With \eqref{eq:splitting}, \eqref{eq:hPsi} and the upper bound 
\begin{equation}\label{eq:upp-psi} 
|\Psi_R(r)| \leq \frac{C}{|\log(k \aa)|} 
\end{equation}
for all $|x| \leq R_0$ (which follows from \eqref{eq:trialPsi} and \eqref{eq:mR-bnd}), we get 
\[
\bmedia{\Psi_R, {\mathfrak h} \Psi_R} \leq  k^2  \bmedia{\Psi_R, \Psi_R}+ \frac {C k^2} {|\log (k\aa) |} \int_{|x| \leq R_0} (1+ V(x)) \, \di x  \,.
\]
On the other hand, Eq.\eqref{eq:trialPsi}, together with $m_R(x)=|x|$ for $|x| \geq R_0$, implies the lower bound 
\[ \begin{split} 
\bmedia{\Psi_R,\Psi_R} &\geq  \int_{R_0 \leq |x| \leq R}  |\Psi_R(x)|^2  \di x \geq  \frac{C}{|\log (k\aa)|^2} \int_{R_0 \leq |x| \leq R}   \log^2(|x|/\aa)   \di x  \geq CR^2 \,.
\end{split} \]
Hence, with \eqref{eq:chosenk}, we conclude that 
\be \begin{split}\label{eq:lR-upp}
\l_{R}   \leq  \frac{\bmedia{\Psi_R, \mathfrak h \Psi_R}}{\bmedia{\Psi_R,\Psi_R}}&  \leq k^2 \, \left( 1 + \frac {C\, |\log (k\aa)|}{R^2} \right)   \\
& \leq  \frac{2}{R^2 \log(R/\aa)} \left( 1 + \frac{3}{4}\fra{1}{\log(R/\aa)}+  \frac{C}{\log^2(R/\aa)}  \right)\, 
\end{split}\ee
in agreement with \eqref{eq:eigenvalue2}. 

To prove the lower bound for $\l_R$ it is convenient to show some upper and lower bounds 
for $f_R$. We start by considering $f_R$ outside the range of the potential. We denote $\e_R= \sqrt{\l_R}\, R$. Keeping into account the boundary conditions at $|x| = R$, we find, for 
$R_0 \leq |x| \leq R$, 
\[ \label{eq:fellout}
 f_{R} (x) =A_R \, J_0( \e_R\,|x| /R) + B_R \, Y_0( \e_R\, |x| /R)\,, 
\]
with
\[ 
A_R =\left(J_0(\e_R)  - J_1(\e_R)\, \frac{Y_0(\e_R)}{Y_1(\e_R)} \right)^{-1} \,,
\]
and
\[
B_R  =\left( Y_0(\e_R)  - \frac{J_0(\e_R)}{J_1(\e_R)} Y_1(\e_R)  
\right)^{-1} \,.
\]
From \eqref{eq:lR-upp}, we have $|\e_R|\leq C\, |\log (R/\aa)|^{-1/2}$. Thus, we can expand $f_R$ for large $R$, using \eqref{eq:expansion} and, for $Y_0$, the improved bound 
$$\left|Y_0 (r)  -\frac2\pi\log(re^\g/2) \left(1-\frac 1 4 r^2\right)\right|\leq C \, r^2\, ,$$  
we find 
\be \label{eq:exp-AB}\begin{split}
\Big| A_R - 1 + \frac {\e_R^2} 4 \Big( 2 \log(\e_R e^\g/2) - 1 \Big)  \Big| &\leq C \e_R^4 (\log \e_R)^2\,, \\
 \Big| B_R - \frac \pi 4  \e_R^2 \left( 1 -\frac{\e_R^2}{8}\right)  \Big| &\leq  C \e_R^6 \, .
\end{split}\ee
which leads to 
\be \label{eq:exp-tlf-outsideR}
\begin{split}
&\left| f_R (x) -1 +\frac {\e_R^2}{4}  \left(2\log( R/|x|) - 1 +\frac{x^2}{R^2} \right)  -  \frac{\e_R^4}{16} \log(R/|x|) \left(1 + \frac{ 2x^2}{R^2} \right)  \right| \\ &\hspace{9cm} \leq C \e_R^4 (\log \e_R)^2 \,.
\end{split}
\ee
We can also compute the radial derivative 
\[
\partial_r f_R (x)= - \frac{\e_R} R \Big( A_R \, J_1( \e_R\,r/R) + B_R \, Y_1( \e_R\, r/R)\Big)\,.
\]
With the expansions \eqref{eq:expansion} and \eqref{eq:exp-AB} we conclude that for all $R_0 \leq |x| < R$ we have 
\be \label{eq:derivative-out}
\left| \partial_r f_R (x)- \frac {\e_R^2} {2|x|} \left( 1 - \frac{x^2}{R^2} + \frac{\eps_R^2 x^2}{2R^2} \log (R/|x|)  
\right)\right| \leq  C \e_R^4 \log \e_R\,.
\ee
The bound (\ref{eq:derivative-out}) shows that $\partial_r f_R (x)$ is positive, for, say, $R_0 < |x| < R/2$. 
Since $\partial_r f_R (x)$ must have its first zero at $|x| = R$, we conclude that $f_R$ is increasing in $|x|$, on $R_0 \leq |x| \leq R$. From the normalization $f_R (x) = 1$, for $|x| = R$, we conclude therefore that $f_R (x) \leq 1$, for all $R_0 \leq |x| \leq R$. 

From (\ref{eq:exp-tlf-outsideR}) and \eqref{eq:lR-upp} we obtain, on the other hand, the lower bound 
\be \begin{split}\label{eq:tlf-lower} 
f_R (x) &\geq  1 -\frac {\e_R^2}{2} \log( R/ |x|)  -C \e_R^4 (\log \e_R)^2 \\
 &\geq 1  - \frac{\log( R/|x|)}{\log (R/\aa)}\left (1 + \frac 3 4 \frac{1}{\log(R/\aa)} +  \frac{C}{\log^2(R/\aa)}\right) - C \,\fra{(\log\log(R/\aa))^2}{\log^2(R/\aa)}\\
& \geq \fra{\log(|x| /\aa)}{\log(R/\aa)} - \frac 3 4   \fra{\log(R/|x|)}{\log^2(R/\aa)} - C \frac{\log (R/|x|)}{\log^3 (R/a)} - C\, \fra{ (\log\log(R/\aa))^2}{\log^2(R/\aa)}\,,
\end{split}\ee
for $R$ sufficiently large. Let $R_*=\max \{R_0, e \aa\}$. Then Eq. \eqref{eq:tlf-lower}  implies in particular that, for $R$ large enough,
\begin{equation}\label{eq:lower-fR} \begin{split}
f_R (x)  \geq \frac{C}{\log(R/\aa)}\,.
\end{split}\end{equation} 
for all $R_* < |x| \leq R$.

Finally, we show that $f_R (x) \leq 1$ also for $|x| \leq R_0$. First of all, we observe that, by elliptic regularity, as stated for example in \cite[Theorem 11.7, part iv)]{Lieb-Loss}, there exists $0<\a<1$ and $C > 0$ such that   
\[ \left| f_R (x)- f_R(y) \right| \leq C  \| (V- 2 \l_R) f_R \|_2 \, |x-y|^\a  \]
With $\| V f_R \|_2 \leq \| V \|_3 \| f_R \|_6 \leq C \| f_R \|_{H^1} \leq C (1 + \lambda_R ) \| f_R \|_2$, we conclude that $0 \leq f_R (x) \leq 1 + C \| f \|_2$ for all $|x| \leq R_0$ (because we know that $f_R (x) \leq 1$ for $R_0 \leq |x| \leq R$). To improve this bound, we go back to the differential equation (\ref{tlf2}), to estimate 
\begin{equation}\label{eq:subhar} \Delta f_R = \frac{1}{2} V f_R - \l_R f_R \geq - \l_R (1+C \| f \|_2) \end{equation} 
This implies that $f_R (x) + \l_R (1+C \| f \|_2) x^2 /2$ is subharmonic. Using (\ref{eq:exp-tlf-outsideR}), we find $f_R (x) \leq 1 - C \eps_R^2$ for $|x| = R_0$. From the maximum principle, we obtain therefore that 
\begin{equation}\label{eq:fR-up} f_R (x) \leq 1 - C \eps^2_R + C \l_R (1+C \| f_R \|_2)  \end{equation}
for all $|x| \leq R_0$. In particular, this implies that $\| f_R {\bf 1}_{|x| \leq R_0} \|_2 \leq C + C \l_R  \| f_R \|_2$, and therefore that 
\[ \| f_R {\bf 1}_{R_0 \leq |x| \leq R} \|_2 \geq \| f_R \|_2 (1 - C \l_R) - C  \]
With $f_R (x) \leq 1$ for $R_0 \leq |x| \leq R$, we find, on the other hand, that 
$\| f_R {\bf 1}_{R_0 \leq |x| \leq R} \|_2 \leq C R$. We conclude therefore that $\| f_R \|_2 \leq CR$ and, from (\ref{eq:fR-up}), that $f_R (x) \leq 1 - C \eps_R^2 + C/R \leq 1$, for all $|x| \leq R_0$, if $R$ is large enough. 

We are now ready to prove the lower bound for $\l_R$. We use now that any function $\Phi$ satisfying Neumann boundary conditions at $|x|=R$ can be written as $\Phi(x)=q(x)\Psi_R(x)$, with $\Psi_R(x)$ the trial function used for the upper bound and $q>0$ a function that satisfies Neumann boundary condition at $|x|=R$ as well. This is in particular true for the solution $f_R(x)$ of \eqref{tlf2}. In the following we write
\[
f_R (x) = q_R (x) \Psi_R(x)
\]
where $q_R$ satisfies Neumann boundary conditions at $|x|=R$. From (\ref{eq:trialPsi}), we find $|\Psi_R (x)| \geq C / \log (ka)$. The bound $f_R(x)\leq 1$ implies therefore that there exists $c>0$ such that
\be \label{eq:qR-upper}
q_R(x) \leq C \log(k\aa)    \qquad \forall \, |x| \leq R_0 \,.
\ee
From the identity
\[
\mathfrak h f_R= (\mathfrak h \Psi_R) q_R - (\D q_R) \Psi_R - 2 \nabla q_R \nabla \Psi_R
\]
we have 
\[
\int_{|x|\leq R} \di x \, f_R \mathfrak h f_R = \int_{|x|\leq R} \di x \, |\nabla q_R|^2 \Psi_R^2 + \int_{|x|\leq R} \di x \, |q_R|^2 \Psi_R \mathfrak h \Psi_R\,.
\]
From \eqref{eq:hPsi} and \eqref{eq:upp-psi}, we have
\[ \begin{split} \label{Psi-h-Psi}
\left| \Psi_R (x)  (\mathfrak h \Psi_R) (x) - k^2 \Psi^2_R (x) \right| \leq \, C \frac{k^2}{|\log k a |}  (1+ V(x)) \chi (|x| \leq R_0)\,.
\end{split}\]
Hence
\be \begin{split} \label{eq:lb1}
&\int_{|x|\leq R} \di x\,  f_R \mathfrak h f_R  \geq\;  k^2\|f_R\|^2 - \frac{C k^2}{|\log k|} \int_{|x|\leq R_0} \di x \, (1+V(x)) |q_R (x)|^2  \,.
\end{split}\ee
With (\ref{eq:qR-upper}), we obtain 
\[ \begin{split} \label{eq:lb2}
\int_{|x|\leq R}\di x\, f_R & \mathfrak h f_R \geq   k^2\|f_R\|^2 - C k^2 \log(k\aa)\,.
\end{split}\]
With (\ref{eq:lower-fR}) (recalling that $R_*=\max\{R_0, e \aa\}$), we bound 
\[
\|f_R\|^2 \geq \int_{R_* \leq |x| \leq R}  |f_R (x)|^2 \,  \di x  \geq \frac{C R^2  }{\log^2(R/\aa)} 
\]
and, inserting in (\ref{eq:lb1}), we conclude that
\[ \label{eq:lambdaR-LB}
\begin{split}
\l_{R} =  \frac{\bmedia{f_R, \mathfrak h f_R}}{\bmedia{f_R,f_R}} &\geq k^2 \left( 1 - \frac{C \log^3(R/\aa)}{R^2 }  \right) \\&\geq  \frac{2}{ R^2 \log (R/\aa)} \,\left( 1+\fra34\fra{1}{\log(R/\aa)} -\, \frac{C}{\log^2 (R/\aa)} \right)\,,
\end{split}\]
where in the last inequality we used \eqref{eq:chosenk}. \\

To prove \eqref{eq:intpotf2} we use the scattering equation \eqref{tlf2} to write 
\[
\int \di x\, V(x)  f_{R}(x) =  2 \int_{|x| \leq R} \di x\, \Delta  f_{R}(x)  + 2 \int_{|x| \leq R} \di x  \lambda_{R}  f_{R}(x)\,.
\]
Passing to polar coordinates, and using that $\Delta f_R (x) = |x|^{-1} \partial_r |x| \partial_r f_R (x)$, we find that the first term vanishes. Hence 
\[
\int \di x\, V(x)  f_{R}(x)= 2 \lambda_R \int \di x \,  f_{R} (x) \,. \label{B.29}
\]
With the upper bound $f_R(r) \leq 1$ and with (\ref{eq:eigenvalue2}), we find  
\[ \label{eq:int-lambda-upp}
\int \di x\, V(x)  f_{R} (x) \leq 2 \pi R^2  \lambda_R  \leq  
\frac{4 \pi}{\log (R/\aa)}  \left( 1 +  \frac{C}{\log(R/\aa)} \right)\,.
\]
To obtain a lower bound for the same integral we use that $ f_{R}(r) \geq 0$ inside the range of the potential. Outside the range of $V$, we use \eqref{eq:exp-tlf-outsideR}. We find
\[
\int \di x\, V(x)  f_{R}(x) \geq 4 \pi  \lambda_R \int_{R_0}^R \di r \, r \,  (1-  C \eps_R^2 \log (R/r) )  \geq \frac{4 \pi}{\log (R/\aa)}  \left( 1 -  \frac{C}{\log(R/\aa)} \right) 
\]
We conclude that 
\[\left| \int \di x\, V(x)  f_{R}(x) - \frac{4\pi}{\log (R/\aa)} \right| \leq \frac{C}{\log^2 (R/\aa)}.\]

Finally, we show the bounds in \eqref{eq:w-bounds}. For $r\in [R_0, R]$, from \eqref{eq:exp-tlf-outsideR} we have 
\begin{equation}\label{eq:wRbd}
\left|\, w_R (x)- \frac{\log(R/|x|)}{\log(R/\aa)}\,\right| \leq  \frac{C}{\log(R/\aa)}\,.
\end{equation}
As for the derivative of $w_R$ we use \eqref{eq:derivative-out} to compute
\[
\left| \partial_r f_R (x)\right|  \leq  \frac{C}{|x|}  \frac{1}{\log(R/\aa)}\,.
\]
Moreover $\partial_r f_R (x)= 0$ if $|x| = R$, by construction. 

On the other hand, if $|x| \leq R_0$, we have $w_R (x)= 1- f_R(x)\leq 1$.
As for the derivative, we define $\wt{f}_R$ on $\bR_+$ through $\wt{f}_R (r) = f_R (x)$, if $|x| = r$, and we use the representation  
\[ 
\wt{f}'_R (r) = \frac 1 r \int_0^r \di s \big( \wt{f}''_R(s) s + \wt{f}'_R (s) \big)\,.
\]
With (\ref{tlf2}), we have (with $\wt{V}$ defined on $\bR_+$ through $V(x) = \wt{V}(r)$, if $|x| = r$) 
\[
\wt{f}''_R (r) + \frac 1 r \wt{f}'_R (r) = \l_R \wt{f}_R (r) - \frac 1 2 \wt{V} (r) \wt{f}_R (r)\,,
\] 
By (\ref{eq:wRbd}), we can estimate $\wt{f}_R (R_0) \leq C / \log (R/\aa)$. 
From (\ref{eq:subhar}), we also 
recall that 
\[ \wt{f}_R (r) \leq \wt{f}_R (R_0) + C R \l_R \leq C /\log (R/\aa) \] 
for any $r < R_0$. We conclude therefore that 
 \[ \begin{split}
|\wt{f}'_R(r) | & =  \Big|\frac 1 r \int_0^r \di s s \big( \l_R \wt{f}_R(s) - \frac 1 2 \wt{V} (s) \wt{f}_R(s)\big)  \Big| \\
& \leq  \frac{\l_R}{r}  \int_0^r r dr  + \frac {C}{r \log (R/\aa)}   \int_0^r  dr \, r  \wt{V} (r)  \\
& \leq  \frac{C}{\log(R/\aa)}  +  C\, \| V \|_2  \frac{\log (R_0/\aa)}{\log(R/\aa)}   \leq   \frac{C}{\log(R/\aa)}  \,.
\end{split}\]

\end{proof}

\end{document}